\PassOptionsToPackage{unicode}{hyperref}
\PassOptionsToPackage{hyphens}{url}
\PassOptionsToPackage{dvipsnames,svgnames,x11names}{xcolor}
\documentclass[
  12pt,
  letterpaper,
]{article}
\usepackage{rotating}
\usepackage{pdflscape}
\usepackage{afterpage}
\usepackage{amsmath,amssymb}
\usepackage{setspace}
\usepackage{iftex}
\ifPDFTeX
  \usepackage[T1]{fontenc}
  \usepackage[utf8]{inputenc}
  \usepackage{textcomp} 
\else 
  \usepackage{unicode-math} 
  \defaultfontfeatures{Scale=MatchLowercase}
  \defaultfontfeatures[\rmfamily]{Ligatures=TeX,Scale=1}
\fi
\usepackage{lmodern}
\ifPDFTeX\else
\fi
\IfFileExists{upquote.sty}{\usepackage{upquote}}{}
\IfFileExists{microtype.sty}{
  \usepackage[protrusion=false]{microtype}
  \UseMicrotypeSet[protrusion]{basicmath} 
}{}
\usepackage{xcolor}
\usepackage[margin=1in]{geometry}
\usepackage{cmap}
\usepackage[T1]{fontenc}
\usepackage[utf8]{inputenc}
\input{glyphtounicode}
\usepackage{mathpazo}
\usepackage[authoryear,round]{natbib}
\bibpunct{(}{)}{;}{a}{,}{,}
\usepackage{caption}
\usepackage{placeins}
\usepackage{float}
\usepackage{graphicx}
\usepackage{booktabs}
\usepackage{tabularx}
\usepackage{xcolor}
\definecolor{LJETBlue}{HTML}{0080AC}
\newtheorem{theorem}{Theorem}
\newtheorem{proposition}{Proposition}
\newtheorem{lemma}{Lemma}

\newtheorem{corollary}{Corollary}

\ifLuaTeX
  \usepackage{selnolig}  
\fi
\IfFileExists{bookmark.sty}{\usepackage{bookmark}}{\usepackage{hyperref}}
\IfFileExists{xurl.sty}{\usepackage{xurl}}{} 
\hypersetup{
  colorlinks=true,
  linkcolor={Maroon},
  filecolor={Maroon},
  citecolor={Blue},
  urlcolor={Blue},
  pdfcreator={LaTeX via pandoc}}

\author{}
\date{}

\begin{document}

\setstretch{1.5}
\hypersetup{
  pdftitle={The Division of Understanding: Specialization and Democratic Accountability},
  pdfauthor={Giampaolo Bonomi},
  linkcolor=LJETBlue,
  citecolor=LJETBlue,
  filecolor=LJETBlue,
  urlcolor=LJETBlue
}

\title{The Division of Understanding: Specialization and Democratic Accountability\thanks{I am grateful to Arjada Bardhi, Renee Bowen, Roberto Corrao, Carlo Cusumano, Kfir Eliaz, Germán Gieczewski, Aram Grigoryan, Matias Iaryczower, Gleason Judd, Samuel Kapon, Alessandro Lizzeri, John Londregan, Pietro Ortoleva, Freddie Papazyan, Harry Pei, John Quah, Kris Ramsay, Alessandro Riboni, Denis Shishkin, Joel Sobel, Guido Tabellini, Isabel Trevino, Joel Watson, Galina Zudenkova, and participants in the Political Economic Theory session at SITE 2026 for helpful comments and suggestions.}}
\author{Giampaolo Bonomi\thanks{Princeton University. Email: bonomi@princeton.edu}}
\date{September 2026}
\maketitle
\vspace{-1em}
\begin{abstract}
\vspace{.5em}
\small
\begin{spacing}{1.1}
I develop a model in which occupational knowledge links the organization of production to democratic performance. In production, specialized knowledge allows workers to perform narrow tasks efficiently, while cross-field knowledge allows coordinators to keep the system coherent. In politics, voters bring their occupational knowledge to the ballot. Voters lacking the relevant knowledge have difficulty distinguishing good policies from bad ones, relaxing the electoral pressure towards high-quality policymaking. I show that productive efficiency can sacrifice knowledge needed to understand cross-domain policies. Competitive markets reward knowledge for its productive contribution while leaving its political value unpriced; I identify when this externality creates over-specialization and show that its welfare cost can grow with policy complexity. When policies are in the common interest, I show that expanding coordinators' responsibilities can maximize welfare by increasing their employment share and voters' average policy understanding. When policies have redistributive consequences, broadening specialists' knowledge can instead be preferable: reducing inequality in policy understanding improves the allocation of public spending. I use the framework to study the political effects of technology shocks, showing how AI adoption at different layers of the organizational hierarchy affects welfare and when firms' incentives conflict with social efficiency.

\vspace{.20cm}


\end{spacing}
\end{abstract}

\newpage

\hypertarget{introduction}{%
\section{Introduction}\label{introduction}}
\label{sec:introduction}

\begin{flushright}
\begin{minipage}{0.74\textwidth}
\small\itshape
``The understandings of the greater part of men are necessarily formed by their ordinary employments.''

\par\smallskip
\raggedleft\upshape --- Adam Smith \citeyearpar{smith-1776}, Book V, Chapter I
\end{minipage}
\end{flushright}
\vspace{0.4em}

\emph{The Wealth of Nations} contains two accounts of the division of labor.
Book I explains its productive appeal: narrow assignments allow workers to avoid switching tasks, specialize in an operation, and master tools designed for it.
Book V asks what the same organization does to citizens. A worker confined
to a few narrow operations, Smith warns, may become ``incapable of judging
the great and extensive interests of his country.'' Together, the two
passages motivate the question studied here: can the division of labor be
efficient for production yet inefficient for the functioning of democracy?

Consider a port. A customs specialist needs to understand inspection
procedures; a terminal operator needs to understand how to load and unload
ships. Someone must also work out how inspections, berth capacity, and
inland transport fit together. Assigning this coordinating responsibility
to a small group allows everyone else to concentrate on a narrower job.
The distinction between system-wide and local knowledge matters when these workers evaluate a proposed port expansion policy or a change in trade policy. Being an expert in one operation is different from understanding how a policy will affect the system through the connections among its parts. Hence, the central tension of this paper: voters may need to combine knowledge that productive incentives divide among workers. 

I study this tension in a model in which citizens work and vote. The knowledge
they use in production determines which public
policies they understand well.\footnote{Section
\ref{sec:understanding} reports descriptive evidence. In the Dutch LISS
panel, workers answer questions about family policies more accurately
when their occupational knowledge better matches the questions. In the ECB Consumer
Expectations Survey, supervisory and strategic responsibilities are
associated with greater knowledge of inflation and the ECB's remit.
Online Appendix OF reports the data and specifications. These associations
do not establish causality or test the electoral mechanism.}
The electorate's understanding, in turn, shapes the quality of equilibrium policymaking:

\begin{center}
\emph{organization of production}
\(\rightarrow\)
\emph{knowledge}
\(\rightarrow\)
\emph{policy understanding}
\(\rightarrow\)
\emph{electoral accountability}
\end{center}
I characterize when organizations that maximize output are socially too
specialized, which changes in work, political authority, and information
sharing can improve welfare, and where those remedies fall short. Finally, I study how technology shocks affect democracy by changing
the knowledge workers acquire. The case of artificial intelligence serves as the leading example.

The model starts by deriving the productively efficient distribution of
knowledge. Workers have a fixed stock of knowledge to distribute across
knowledge domains, while the production system specifies how much of each domain is required to produce and the domain interdependencies. An organization is a tree of responsibilities, specifying who performs each part and
with what expertise. Within this tree, coordinators draw on several
knowledge domains to resolve interdependencies and reduce the complexity of the
remaining work, allowing it to be divided into narrower assignments that
require less labor per unit of work. Efficiency requires minimizing labor
per unit of output while preserving the required knowledge mix in
aggregate at every division, giving the problem a
martingale structure.\footnote{The mathematical antecedents include
efficient unbiased sparsification and the \(k\)-support norm
(\citealt{barnes-et-al-2024};
\citealt*{argyriou-foygel-srebro-2012}),
and martingale optimal transport
\citep{beiglbock-juillet-2016}.
Here, intermediate responsibilities, delegation points, and staffing
are endogenous.}
I show that any productively efficient
organization is a hierarchy of progressively narrower responsibilities,
with workers' knowledge matched to their assignments. Broad knowledge is valuable precisely because it
allows much of the remaining work to be done efficiently without it.

Each knowledge profile has a political value, which depends on the policy questions citizens face. At the port, customs expertise may suffice to assess an
inspection rule. Evaluating a port expansion requires understanding how
inspections, berth capacity, and inland transport interact. In an
efficient hierarchy, specialists understand issues confined to one knowledge domain
as well, on average, as the coordinators above them.\footnote{The deep understanding of inspection rules by a customs expert is compensated by the lack of understanding of the same rules by specialists in other knowledge domains.} When an issue
requires several knowledge domains to be considered together, expertise in one knowledge domain does not compensate a voter for
missing knowledge from another. Coordinators can therefore understand
these issues better on average, even when specialists know more about
their own knowledge domains.

Equipped with their occupational knowledge and the resulting policy understanding, the workforce becomes the electorate, called to evaluate the platforms of two office-motivated politicians competing in an election. I first consider policies that serve a common interest: everyone benefits
equally from an improvement in their quality. Competing candidates choose
how much effort to devote to designing better policies, and voters
with a better policy understanding assess quality more accurately. A more discerning electorate
therefore gives both candidates a greater incentive to improve their platforms.
In the unique pure strategy equilibrium, effort and quality are strictly increasing in the electorate's
average understanding. The distribution of knowledge in society, therefore, affects citizens' welfare not
only through the income and public resources it generates, but also through
the quality of government that elections sustain.

Knowledge acquisition therefore impacts utilitarian welfare via two margins: the productive margin and the civic one. Can its decentralization balance such margins optimally? Unsurprisingly, competitive markets reward occupational knowledge for its productive
contribution. Its contribution to better government, however, represents an externality: it is absent from the wage comparisons of individual
workers and the staffing decisions of competing firms, as their impact on election outcomes is too little to justify the corresponding losses in income and productivity. Markets, therefore, implement a productively efficient organization of knowledge, disregarding any effect on democracy. 

Different organizations may produce the
same output with different distributions of knowledge, allowing a social
planner to improve accountability at no production cost. Beyond such improvements, the planner faces a real
trade-off. To further improve understanding, the planner has different options. She can assign more coordination work to positions with the broadest responsibilities, replacing some narrower-scope workers. This increases the share of workers
who need to understand the whole system. Alternatively, she can cross-train
specialists, broadening their knowledge while leaving their jobs narrow.
Neither is free: the first gives up some of the labor savings from
delegation; the second diverts part of a worker's limited knowledge from
the knowledge domains her job requires. Which adjustment delivers the greater
political benefit for a given production cost?

I show that, under some symmetry conditions, assigning more workers to positions with the broadest responsibilities is at least as effective as cross-training narrow workers when
policy serves a common interest. An optimum follows a simple cutoff rule: while
increasing employment in system-wide positions is worthwhile, the
planner begins by augmenting system-wide responsibilities with the broadest responsibilities that would
otherwise be delegated downstream. Responsibilities narrower than a cutoff are delegated to specialists, as the political value of undoing delegation is not worth the production cost. The more the political agenda becomes complex, covering system-wide policies, the stronger the incentive to further reduce delegation and convert intermediate-breadth positions into system-wide ones. If all policy issues are sufficiently complex and broad coordination is sufficiently inexpensive, the resulting organization strictly improves welfare.

The answer changes when policy issues have strong redistributive components, so that the policy platforms differ not only in the total value generated for citizens, but also in how such a value is distributed. The groups with higher policy understanding are
more responsive to policy changes, giving
candidates an incentive to direct more spending toward them. Unequal understanding then has a
distributive cost: because targeted spending has a diminishing marginal value for voters,
concentrating it on the most electorally responsive groups lowers
welfare. Raising societal understanding is no longer the planner's only
political concern. Cross-training can address the allocative margin by improving the judgment
of workers who would otherwise be politically disadvantaged. I show that
an optimum can leave their productive responsibilities narrow while broadening the
knowledge they possess.\footnote{Section
\ref{sec:institutions} characterizes optimal organizations in a symmetric
benchmark with targeted spending. Online Appendix OD gives an example in
which cross-training both specialist groups is globally optimal.}
The workers give up some productive expertise to better evaluate policies
outside their immediate assignments. Hence, a consequential insight:
the political case for broader knowledge need not be a case for less
specialized jobs.

Reorganizing production comes with an output loss. Could changing political institutions instead achieve these gains for free? One could rightly observe that the market failure is a failure of democracy. In production, specialists entrust the coordinators at the top of the knowledge hierarchy with the right system-wide decisions; in democracy, each vote is worth the same. It is the latter, the argument goes, that should be more hierarchical. Indeed, giving greater weight to the vote of the ``knowledge elite'' would result in better designed policies, which would, ceteris paribus, benefit everyone. Once the redistributive margin is accounted for, however, the same remedy also redirects resources toward
those whose authority increases, exacerbating the allocative inefficiency, potentially leading to an overall welfare loss. Information sharing faces a related obstacle. By disclosing their policy analysis, a better informed group can help other citizens
judge policy, increasing electoral discipline. The same disclosure, however, makes the receivers more aware of their redistributive stakes. When such stakes are large
enough, informed groups prefer to keep others uninformed. Neither greater authority for the knowledgeable nor voluntary
communication is therefore assured to remove the political cost of how
work is organized.

The framework connects the organization of production to the health of the political system. As such, it naturally lends itself to studying the political effects of
technology shocks. A new technology can change the productive value of
different skills and thereby alter the knowledge people acquire for their
jobs and bring to the ballot. AI is a natural and salient example and its effect, I argue, depends on whether
it takes over the narrow tasks performed by specialists or the
synthesis work performed by coordinators. I study the implications of AI adoption for the composition of
human occupational knowledge. Firms allocate a limited AI capacity where it has a comparative advantage relative to humans. Automating specialist work shifts
employment toward system-wide responsibilities, increasing human understanding of complex issues; automating coordination
shifts it toward narrower ones. Even a small
productive gain from automating coordination can induce firms to
choose it when automating specialist work yields higher welfare.

In short, the division of labor is also a division of understanding. Modern economies can be expected to push workers towards the type of knowledge that makes production efficient, not the one that makes them empowered citizens. Policy and technological changes should be evaluated accordingly.

The remainder of the paper is organized as follows. Section \ref{sec:related-literature} discusses the related literature.
Section \ref{sec:production} studies the productively efficient organization.
Section \ref{sec:understanding} links occupational knowledge to policy
understanding and presents descriptive evidence.
Section \ref{sec:political-competition} analyzes electoral competition.
Section \ref{sec:social-optimum} compares the competitive and socially
optimal organizations of production.
Section \ref{sec:institutions} studies targeted spending, political
institutions, and optimal organization.
Section \ref{sec:ai} examines AI adoption.
Section \ref{sec:conclusion} concludes.

\hypertarget{related-literature}{%
\section{Related Literature}\label{related-literature}}
\label{sec:related-literature}

The paper contributes to theories of how coordination sustains
specialization. \citet{becker-murphy-1992}
identify the trade-off between specialization and coordination costs;
\citet{garicano-2000} and
\citet{garicano-rossi-hansberg-2006}
derive hierarchies that refer problems to workers with the expertise
to solve them. \citet{ide-talamas-2025} use such a framework to study
AI's effects on organization and wages. In the present paper, coordination changes what can be delegated:
resolving interdependencies makes delegation to specialists possible.\footnote{
Related models study coordination and adaptation to local information
\citep{dessein-santos-2006}, authority over coordinating decisions
\citep{hart-moore-2005} and the breadth of
managerial expertise \citep{ferreira-sah-2012}.
The approach also connects to decomposability and modular production
\citep*{simon-1962,matouschek-powell-reich-2025}
and to architectural knowledge and systems integration
\citep*{henderson-clark-1990,brusoni-prencipe-pavitt-2001}.}
I solve jointly for the division of production into responsibilities,
the knowledge each job requires, and its staffing. Although several
responsibility trees can minimize costs, the analysis pins down labor
demand at each breadth of expertise.

The closest political-economy connection is to
\citet*{ponzetto-petrova-enikolopov-2020}: both papers treat political knowledge as a byproduct
of economic activity.\footnote{This
mechanism complements research linking work to political preferences
\citep*{iversen-soskice-2001,kitschelt-rehm-2014} and civic
participation (\citealt*{brady-verba-schlozman-1995}; \citealt{selenko-et-al-2025}).
It also identifies a role for knowledge breadth alongside the civic
returns to schooling \citep*{dee-2004,milligan-moretti-oreopoulos-2004,glaeser-ponzetto-shleifer-2007}, connecting with arguments for liberal education
\citep*{nussbaum-2010,zakaria-2015}.}
In their model, producers learn proposed trade policies to guide
investment, and politicians favor the sectors whose beneficiaries are
better informed. Their mechanism concerns which policies voters learn
about; mine concerns the expertise with which voters evaluate policy
consequences. Productive requirements determine which knowledge domains
workers master, while policy requirements determine which domains they
must understand together. The resulting mismatch exposes a political
externality in the division of labor. By deriving both the allocation
of expertise and the productive cost of changing it, the model identifies
when changing responsibilities or cross-training improves welfare.

The electoral mechanism shares the emphasis on informed evaluation in
models of accountability \citep*{barro-1973,ferejohn-1986} and research on media and
responsiveness \citep*{stromberg-2004,snyder-stromberg-2010}.
\citet{matejka-tabellini-2021}
show how voters allocate attention according to their political stakes.
Specialized learning can reduce responsiveness to platforms
\citep{yuksel-2022}, and media competition can shift
information toward divisive issues
\citep{perego-yuksel-2022}.
In my model, firms' demand for skills determines the expertise voters bring to
policy evaluation. Technological shocks can therefore alter electoral
accountability by reallocating knowledge across jobs.

The distribution of this expertise also matters for distributive
politics. Electoral incentives shape the provision of public goods and
targeted benefits \citep{lizzeri-persico-2001}, favoring responsive groups
\citep*{lindbeck-weibull-1987,dixit-londregan-1996}; informational
heterogeneity can consequently affect the allocation of resources
\citep{blumenthal-2026}.
I endogenize this source of influence through productive organization.
Unequal understanding then has an allocative cost, so the social return
to occupational knowledge depends on its distribution as well as its
mean. This creates a motive for cross-training specialists even when
their productive responsibilities remain narrow.

These distributive effects also qualify the political use of expertise.
While scholars have argued that credible advice can guide an otherwise ignorant electorate
\citep{lupia-mccubbins-1998}, it is well known that
conflicting interests can impede information transmission
\citep*{crawford-sobel-1982,agranov-eilat-sonin-2025}.
Here an informed group may withhold even verifiable, costless policy analysis
because informing others strengthens their claim on targeted spending.
\citet{maskin-tirole-2004} study how
accountability affects the allocation of decisions between elected and
insulated officials. I show that distributive concerns can make it
optimal to give better-informed groups less formal political authority.

\hypertarget{production-and-the-allocation-of-responsibility}{%
\section{Production and the Allocation of
Responsibility}\label{production-and-the-allocation-of-responsibility}}

\label{sec:production}

Specialists can work separately only if someone keeps their tasks coherent. At a port, a customs delay affects both berth availability and rail departures; a coordinator who resolves such interdependencies allows clearance, terminal operations, and inland transport to be run by separate specialists. The model describes this division of responsibilities, the knowledge each requires, and the workers who perform them. The firm chooses among such organizations the one that produces a unit of output with the least labor.


\hypertarget{production-systems-and-responsibilities}{%
\subsection{Production Systems and
Responsibilities}\label{production-systems-and-responsibilities}}

The firm operates a constant-returns production system that draws on
\(K\geq2\) knowledge domains. Write
\(\Delta_K\) for the unit simplex.
An exogenous production profile \(q\in\Delta_K\) records the share of
the knowledge required for one unit of output that is drawn from each
domain. At a port, the domains include terminal operations, inland
transport, capacity planning, and environmental regulation. These
requirements interact: freight flows must be reconciled with berth
capacity and environmental rules. Production therefore involves two
kinds of work: executing components and resolving their interfaces.

The production system can be broken down into different responsibilities. The division is represented by a finite rooted tree \(T\),
with root \(0\) and nodes \(v\in T\). Each node is a responsibility
with scale \(\omega_v>0\) and requirement profile
\(p_v\in\Delta_K\):
the scale measures the requirements it covers, and the profile their mix
across domains. Its requirement from domain \(k\) is therefore
\(\omega_vp_{vk}\). Together, these objects define a responsibility
tree \((T,\omega,p)\), whose root \((\omega_0,p_0)=(1,q)\)
covers the whole system. An internal node divides its responsibility
among its children; a terminal node performs it without further division.
Writing \(\operatorname{ch}(v)\) for the children of an internal node
\(v\), the tree is coherent if every division satisfies
\[
\sum_{c\in\operatorname{ch}(v)}\omega_c=\omega_v,
\qquad
\sum_{c\in\operatorname{ch}(v)}\omega_cp_c=\omega_vp_v.
\]

Coherence requires the children's requirements to add up, domain by
domain, to those of the parent. The parent always covers the children. When a freight-flow responsibility is
divided into terminal operations and inland transport, for example, it still covers both activities. Its own work is, however, only resolving the interfaces that make such activities separable; each child receives
its requirements profile and scale, and the corresponding unresolved interfaces.

For a profile \(p\), the interface load per unit of scale is
\(d(p)=1-\lVert p\rVert_2^2\), the probability that two requirements
drawn independently from \(p\) fall in different domains. It is zero
for a single-domain profile and largest for a uniform one, so a more
balanced responsibility carries more interface work: \(\omega d(p)\)
units at scale \(\omega\). When a responsibility is divided, the
parent's own work is the load it removes: its interface load less what
its children inherit. A terminal responsibility instead executes its
components and resolves all the interfaces it retains. The work
performed at node \(v\) is therefore
\[
\begin{aligned}
z_v&=
\begin{cases}
\displaystyle
\omega_vd(p_v)-\sum_{c\in\operatorname{ch}(v)}\omega_cd(p_c),
&\text{if \(v\) is divided},\\[8pt]
\omega_v[1+d(p_v)],
&\text{if \(v\) is terminal},
\end{cases}\\[3pt]
\sum_{v\in T}z_v&=
\underbrace{1}_{\substack{\text{component}\\\text{execution}}}
+
\underbrace{d(q)}_{\substack{\text{interface}\\\text{resolution}}}.
\end{aligned}
\tag{1}\label{eq:work-conservation}
\]
Coherence and the concavity of \(d\) make \(z_v\geq0\). Equation~\eqref{eq:work-conservation}
counts each unit of work once, since interfaces resolved by the parent
are removed from what its children inherit. Every organization
therefore performs the same total work; what differs is how much of it
can be entrusted to narrow specialists.
\citet{brusoni-prencipe-pavitt-2001} document this pattern in aircraft
engine controls, where firms retain knowledge of outsourced components
in order to coordinate specialized suppliers. In the model,
responsibilities spanning several domains carry the work that makes
such specialization possible. The next step is to assign this work to
jobs and determine the labor it requires.

\hypertarget{occupational-profiles-and-processing}{%
\subsection{Positions, Knowledge, and Staffing}\label{occupational-profiles-and-processing}}

The firm staffs the tree by grouping responsibilities into positions.
Formally, it chooses a finite set of positions \(\mathcal H\) and an onto map
\(\iota:T\to\mathcal H\) assigning each responsibility to a position.
The responsibilities in a position form a connected set:
if it contains two nodes, it contains every node on the path between
them. A position is a job type; how many workers it employs depends on
the labor its assigned work requires. A division of \(v\) among its
children is a handoff to child \(c\) when \(\iota(c)\neq\iota(v)\);
otherwise the work on both sides of the division belongs to the same
job.

Workers in position \(h\) share an occupational knowledge profile
\(s_h\in\Delta_K\), which allocates each worker's fixed unit stock of
knowledge across domains. Whereas \(p_v\) describes what knowledge mix a
responsibility requires, \(s_h\) describes what the workers assigned
to it know. Responsibilities grouped into one position are therefore
performed with the same knowledge, even when their requirements differ,
while different positions may use the same profile. In this section, I derive the knowledge profiles firms demand and assume that they are
supplied by the citizens; Section~\ref{sec:social-optimum} shows how competitive wages
induce workers to supply them.

\afterpage{%
    \clearpage
\begin{landscape}
\begin{figure}[p]
\centering
{\footnotesize
\(\omega_v\): responsibility scale\hfill
\(p_v\): requirement profile\hfill
\(s_h\): knowledge profile\par}
\vspace{12pt}
\begin{minipage}[t]{0.45\linewidth}
\centering
{\small\textbf{A. Combined position}\par}
\vspace{8pt}
\includegraphics[width=\linewidth]{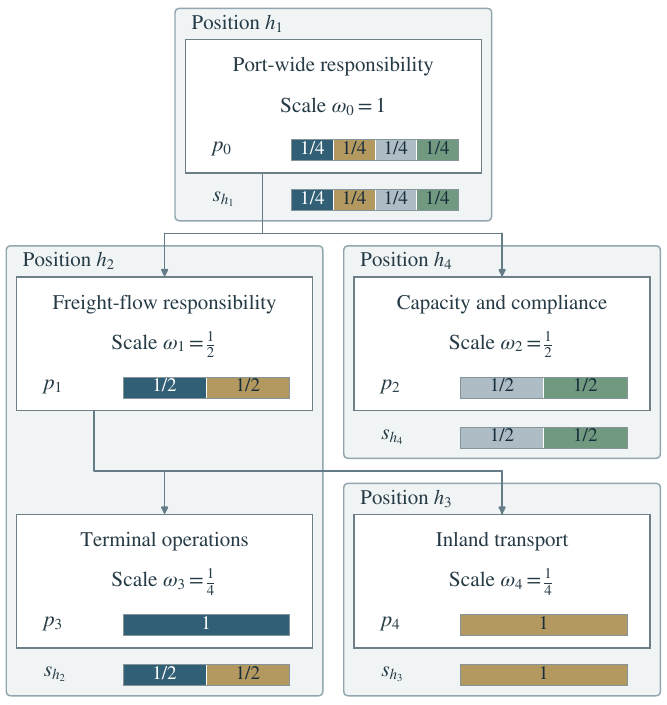}
\end{minipage}\hfill
\begin{minipage}[t]{0.45\linewidth}
\centering
{\small\textbf{B. Separate positions}\par}
\vspace{8pt}
\includegraphics[width=\linewidth]{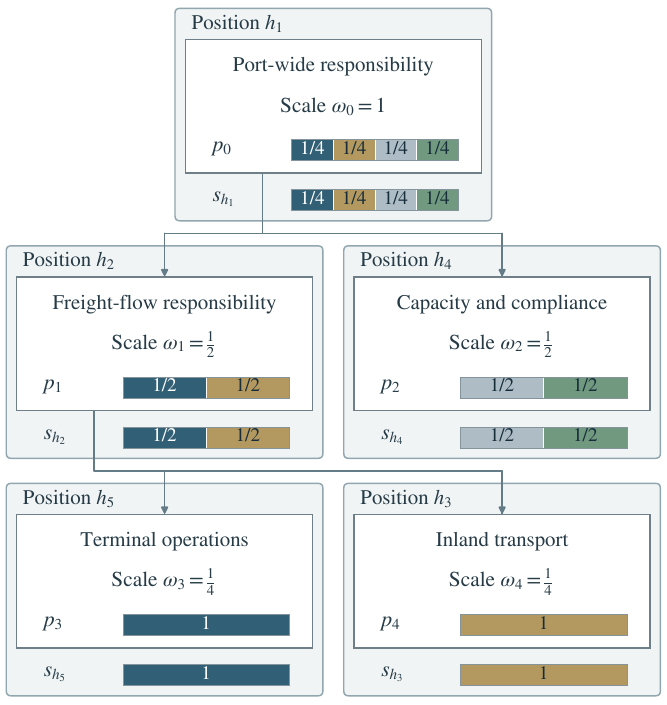}
\end{minipage}
\par\vspace{12pt}
{\footnotesize
\begin{tabular}{@{}ll@{\qquad}ll@{}}
\textcolor[HTML]{315F75}{\rule{8pt}{7pt}} & Terminal operations &
\textcolor[HTML]{B39860}{\rule{8pt}{7pt}} & Inland transport\\[2pt]
\textcolor[HTML]{AEBCC6}{\rule{8pt}{7pt}} & Capacity planning &
\textcolor[HTML]{709980}{\rule{8pt}{7pt}} & Environmental regulation
\end{tabular}\par}
\vspace{12pt}
\caption{A port-wide responsibility and its division into jobs.
White boxes are responsibilities and shaded regions are positions.
Both panels use the same responsibility tree. In A, position \(h_2\)
combines freight coordination and terminal execution; in B, terminal
execution moves to position \(h_5\) with matching knowledge.
Bars show domain shares. Arrows divide a parent's requirements among its
children; those crossing position boundaries are handoffs.}
\label{fig:port-responsibilities}
\end{figure}
\end{landscape}
}

Figure~\ref{fig:port-responsibilities} applies these distinctions to the
port-wide example. The capacity-and-compliance responsibility remains
undivided, performing its components and resolving the interfaces it retains. In panel A,
\(\iota(1)=\iota(3)=h_2\): freight coordination and terminal execution
share workers with profile \(s_{h_2}=(1/2,1/2,0,0)\), although the
latter requires only \(p_3=(1,0,0,0)\). Panel B assigns terminal
execution to a separate position \(h_5\) with \(s_{h_5}=p_3\).
The labor cost of this choice depends on how well workers' knowledge
fits their assignments.

Consider a responsibility with profile \(p\), performed by workers
with profile \(s\). Its scope, \(r(p)=|\{k:p_k>0\}|\), is the number
of domains it requires. Let \(\ell_r\), with
\(1=\ell_1<\cdots<\ell_K\), be the labor intensity of activating and
combining knowledge from \(r\) domains. Productive fit and the labor
required per unit of work are
\[
\beta(s,p)=\min_{k:p_k>0}\frac{s_k}{p_k},
\qquad
a(p,s)=\frac{\ell_{r(p)}}{\beta(s,p)}.
\tag{2}
\]
Fit is bounded by the domain in which the workers' knowledge is scarcest
relative to the responsibility's needs; knowledge in other domains
cannot make up the shortfall. In panel A of
Figure~\ref{fig:port-responsibilities}, the workers in position
\(h_2\) match freight coordination exactly, but their fit for terminal
execution is only \(1/2\). Panel B assigns that work to terminal
specialists, which raises fit to one and halves its labor requirement.
The tree is the same in both panels, the gain comes from
staffing. Even when \(s=p\), work that draws on several domains
takes longer, because the worker must switch between bodies of
knowledge and combine their implications. The increasing \(\ell_r\)
captures this burden, consistent with evidence on the time cost of
switching cognitive rules \citep*{rubinstein-meyer-evans-2001}. It is
also what gives firms a reason to delegate: once coordination makes a
narrower assignment feasible, performing the work at that scope saves
labor, the gain from concentration emphasized by
\citet{becker-murphy-1992}. Component execution and interface
resolution use the same technology.

An organization combines the division of responsibilities, their
assignment to positions, and the knowledge used in each position:
\[
\mathcal A=\bigl((T,\omega,p),\mathcal H,\iota,(s_h)_{h\in\mathcal H}\bigr).
\]
It is feasible if the tree is coherent and rooted at \((1,q)\),
positions are connected, and \(\beta(s_{\iota(v)},p_v)>0\) at every
node; Online Appendix OA shows that the last two restrictions are
without loss. Write \(s_v=s_{\iota(v)}\) for the knowledge profile
used at node \(v\). Its work requires \(z_va(p_v,s_v)\) units of
labor, so the organization's labor requirement per unit of output is
\[
C(\mathcal A)
=\sum_{v\in T}z_v
\frac{\ell_{r(p_v)}}{\beta(s_v,p_v)}.
\]

The workforce is a unit mass of citizens. Each citizen \(i\in[0,1]\)
supplies one unit of labor inelastically, so organization
\(\mathcal A\) produces \(Y(\mathcal A)=1/C(\mathcal A)\) units of
output; each position employs \(Y(\mathcal A)\) times the labor its
assigned nodes require. The tree thus allocates work, knowledge
determines the time it takes, and staffing assigns the workforce
accordingly. I normalize the output price and its conversion into
disposable income so that \(Y\) also measures aggregate disposable
income.\footnote{Section~\ref{sec:political-competition}
  states an alternative fiscal normalization.}

Letting \(\Omega(q)\) denote the set of feasible organizations, the
minimum labor requirement and maximal output for the production system \(q\) are
\[
C^M(q)
=\min_{\mathcal A\in\Omega(q)}\quad C(\mathcal A),
\qquad
Y^M(q)=\frac{1}{C^M(q)}.
\tag{3}
\]

Note that, since \(\beta(s,p)\leq1\), with equality only at \(s=p\), and
positions can be split at no cost, every labor minimizer sets
\(s_v=p_v\) wherever \(z_v>0\). Hence, efficient production requires workers' profiles to be exactly matched to each of their responsibilities.\footnote{Different positive-work responsibilities can still be combined in the same position if their profiles coincide, as pooling such nodes leaves the labor requirement unchanged.} What remains to be chosen is how to
divide production and where to resolve its interfaces. 

Multiple organizations may achieve \(C^M(q)\). As a selection criterion, I minimize the total scale of handoffs between
distinct positions. Let
\(\mathcal M(q)=\{\mathcal A\in\Omega(q):C(\mathcal A)=C^M(q)\}\)
be the set of labor minimizers and \(
E^{\mathrm{ext}}(\mathcal A)
=
\{(v,c):c\in\operatorname{ch}(v),\ \iota(v)\neq\iota(c)\}
\) the set of handoffs between positions in organization \(\mathcal{A}\). The total handoff scale and the selection rule are, respectively,
\[
J(\mathcal A)
\equiv
\sum_{(v,c)\in E^{\mathrm{ext}}(\mathcal A)}\omega_c,
\qquad
\mathcal A^M
\in
\arg\min_{\mathcal A\in\mathcal M(q)}\quad J(\mathcal A).
\]

The tie-break rule identifies minimizers that survive small communication and supervision costs at handoffs (Online Appendix OA). Pruning zero-work nodes and
merging connected nodes assigned to the same position turns a selected
organization into a tree of positions, which I call a \textit{productively
optimal hierarchy}. The next subsection characterizes it.\footnote{Without selection,
  labor minimizers still have single-domain execution, exact matching at
  positive-work responsibilities, and identical workloads and employment
  by scope, labor requirements, and output. Trees and within-scope
  profiles---and policy understanding under general agendas---can differ
  (see the two-domain example in Online Appendix OA). Adjacent-scope
  handoffs, the layer ordering, and post-AI employment composition require
  selection. The planner ranges over \(\Omega(q)\).}

\hypertarget{the-efficient-organization-of-production}{%
\subsection{The Efficient Organization of
Production}\label{the-efficient-organization-of-production}}

Let \(n=r(q)\geq2\). In a productively optimal hierarchy, \(g_r(q)\)
is the interface work performed at scope \(r=2,\ldots,n\) per unit
of output; \(g_1(q)=1\) is component execution.


\begin{theorem}[The efficient organization of production]\label{thm:efficient-organization}
A productively optimal hierarchy exists. Every such hierarchy passes
responsibility only between adjacent scopes, dropping one knowledge
domain at each handoff. Single-domain specialists execute components,
and broader coordinators resolve interfaces. At every positive-work
responsibility, the assigned occupational profile exactly matches the
responsibility profile. The aggregate workloads \(g(q)\), minimum labor
requirement, and output satisfy
\begin{gather}
d(q)=\sum_{r=2}^n g_r(q),
\qquad
0<g_n(q)<\cdots<g_2(q)<g_1(q)=1.
\tag{4}\\[6pt]
C^M(q)=
\underbrace{\ell_1g_1(q)}_{\substack{\text{component}\\\text{execution}}}
+
\underbrace{\sum_{r=2}^n\ell_rg_r(q)}_{\substack{\text{interface}\\\text{resolution}}},
\qquad
Y^M(q)=\frac{1}{C^M(q)}.
\tag{5}
\end{gather}
\end{theorem}

The efficient hierarchy assigns each layer only the coordination work
needed to make narrower responsibilities possible. These workloads
depend only on the production requirements \(q\), even when several
trees are optimal. To characterize them, consider any coherent responsibility
tree and a \emph{cut}: a set of nodes containing exactly one node on
each path from the root to a terminal node. The responsibilities on a
cut cover the system's requirements: the interfaces between them have
been resolved above the cut, and those within each remain. A cut whose
responsibilities span at most \(r\) domains thus shows how much
interface work can be left to workers who know \(r\) or fewer domains.
By coherence, the scales on a cut sum to one and the scale-weighted
average of its profiles is \(q\); selecting a responsibility with
probability equal to its scale therefore gives a random profile
\(P\in\Delta_K\) with \(\mathbb E[P]=q\). For \(1\leq r\leq n\), define the scope-\(r\)
frontier by

\[
\begin{aligned}
\Phi_r(q)
=\min_P\quad
&\mathbb E\lVert P\rVert_2^2\\
\text{subject to}\quad
&\mathbb E[P]=q,\\
&r(P)\leq r\quad\text{almost surely}.
\end{aligned}
\]

Since \(d(P)=1-\lVert P\rVert_2^2\), the quantity \(1-\Phi_r(q)\) is
the greatest interface load that can be left to responsibilities
spanning at most \(r\) domains. The efficient organization attains
every frontier at once: the scope-\(r\) layer resolves
\(g_r(q)=\Phi_{r-1}(q)-\Phi_r(q)\) units of interface work, the
difference between the load that can be left at scope \(r\) and the
load that can be left at scope \(r-1\), and passes the rest down. With
exact matching, this layer requires \(\ell_rg_r(q)\) units of labor
per unit of output, so its employment share,
\(m_r(q)=\ell_rg_r(q)/C^M(q)\), is also its share of the electorate.

\begin{corollary}[Narrower layers employ more workers]\label{cor:employment-hierarchy}
If \(q\) is uniform over all \(K\) knowledge domains, every scope-\(r\) profile is uniform over its support. The scope
workloads are
\vspace{-1em}
\[
g_1=1,
\qquad
g_r=\frac{1}{r(r-1)}
\quad\text{for }r=2,\ldots,K.
\tag{6}\label{eq:uniform-scope-workloads}
\]
For any production profile \(q\), if the labor intensities are affine,
\(\ell_r=1+\kappa(r-1)\) with \(0<\kappa<1\), the corresponding
employment shares satisfy

\[
m_1>m_2>\cdots>m_n>0.
\]
\end{corollary}

With three equally important domains and \(\kappa=1/2\), workloads
\((1,1/2,1/6)\) at scopes \(1,2,3\) require labor
\((1,3/4,1/3)\) per unit of output. Employment shares are therefore
\(48\%\), \(36\%\), and \(16\%\): specialists form the largest layer. In general cases, whether narrower layers employ more workers therefore depends on how
quickly labor intensity rises with breadth.

\hypertarget{occupational-knowledge-and-policy-understanding}{%
\section{Occupational Knowledge and Policy
Understanding}\label{occupational-knowledge-and-policy-understanding}}

\label{sec:understanding}

Citizens can draw on the knowledge they acquire at work to evaluate
policies, but policies differ in the knowledge they call for. An
inspection rule can be judged with customs expertise alone; its
effects on rail arrivals and berth congestion can be traced only by
someone who knows how port operations connect, as a port coordinator
does. Many policies are of the second kind: integrated models of
climate and pandemic policy trace consequences across energy,
emissions, health, behavior, employment, and consumption, knowledge
that no single occupation supplies
\citep*{nordhaus-2013,eichenbaum-rebelo-trabandt-2021}.

To capture these requirements, represent each policy issue by a
profile \(u\in\Delta_K\), where \(u_k\) is the share of the required
knowledge that comes from domain \(k\). A domain-specific issue draws
on one domain, whereas a cross-domain issue draws on several at once. The policy agenda is an exogenous probability measure \(\nu\) over issue profiles. It
determines which combinations of expertise matter for evaluating
policy, hence the political value of a worker's occupational
knowledge.

For a worker with profile \(s\), understanding of issue \(u\) is
\(\beta(s,u)/\ell_{r(u)}\): citizens evaluate policy with the same
cognitive technology they use for their responsibilities at work.
Expertise missing in a required domain cannot be replaced by expertise
in another, and combining more domains is more burdensome. Averaging
over the agenda gives the worker's policy understanding. The average policy understanding in society depends on the employment share of different positions. Under organization
\(\mathcal A\), position \(h\) has employment share

\[
m_h(\mathcal A)
=
\frac{\displaystyle\sum_{v:\iota(v)=h}z_va(p_v,s_h)}{C(\mathcal A)}.
\]

The numerator gives its labor requirement per unit of output; dividing
by total labor gives its share of the electorate. Individual and average
understanding are therefore

\[
\begin{aligned}
b(s;\nu)&=\int\frac{\beta(s,u)}{\ell_{r(u)}}\,d\nu(u),
& b_i&=b(s_i;\nu),\\
B(\mathcal A;\nu)
&=\int_0^1b_i\,di
=\sum_{h\in\mathcal H}m_h(\mathcal A)b(s_h;\nu).
\end{aligned}
\tag{7}
\]

Both \(b_i\) and \(B(\mathcal A;\nu)\) lie in \([0,1]\). When the
agenda is fixed, I write \(b(s)=b(s;\nu)\) and
\(B(\mathcal A)=B(\mathcal A;\nu)\); when the organization is also
fixed, I write \(B\). In a productively optimal hierarchy, let
\(m_r=m_r(q)\) be the employment share at scope \(r\) and \(b_r\)
the average understanding of its workers, so
\(B=\sum_{r=1}^nm_rb_r\).

Delegation separates knowledge that a coordinator holds together. Under
exact matching and coherence, a coordinator's profile
\(p=\sum_j\alpha_jp_j\) averages her children's profiles \(p_j\) with
responsibility shares \(\alpha_j\). Concavity of understanding then
gives \(b(p)\geq\sum_j\alpha_jb(p_j)\): dividing complementary
knowledge weakly reduces average understanding. This holds issue by
issue for every coherent division, selected or not. In a productively optimal hierarchy, work within each scope is
proportional to scale and matched workers face the same labor
intensity, so responsibility shares are employment shares within a
layer, and the inequality yields the following result.

\begin{theorem}[The division of understanding]\label{thm:division-understanding}
In every productively optimal hierarchy, average understanding of any
given policy issue weakly increases with the scope of the layer. Hence,
for every policy agenda,
\[
b_1\leq b_2\leq\cdots\leq b_n.
\tag{8}\label{eq:scope-understanding-order}
\]
On a domain-specific issue in domain \(k\), every layer has the same
average understanding, \(q_k\).
\end{theorem}

Broader layers therefore hold no more expertise, on average, about any
single domain; their advantage lies in combining domains.\footnote{Online
  Appendix OA extends the layer ordering to any issue-understanding
  measure concave in occupational knowledge.} The next corollary isolates this advantage
with an agenda that mixes domain-specific issues with the cross-domain
issue.

\begin{corollary}[The coordinator's premium]\label{cor:coordinator-premium}
Suppose a matched coordinator with profile \(p\) spanning at least two
domains divides her responsibility among matched children with profiles
\(p_j\) and shares \(\alpha_j\), so that \(p=\sum_j\alpha_jp_j\).
Let the agenda place weight \(1-\lambda\) on domain-specific
issues drawn with probabilities \(p\) and weight \(\lambda\in[0,1]\)
on the cross-domain issue with profile \(p\). The coordinator's
understanding exceeds the share-weighted average of her children's by
\[
\frac{\lambda}{\ell_{r(p)}}
\left[
1-\sum_j\alpha_j\beta(p_j,p)
\right].
\tag{9}\label{eq:coordinator-premium}
\]
For \(\lambda>0\), it is strictly positive if and only if some
positive-share child's profile differs from \(p\).
\end{corollary}
The premium comes entirely from the cross-domain issue: the coordinator
holds knowledge in the required proportions, whereas delegation
preserves average expertise in each domain but splits it across
workers. Importantly, coordinators do not understand
every issue better. With \(p=(1/2,1/2)\), each specialist understands
an issue confined to her own domain best, and the coordinator is better
placed on one that requires both.

I use survey data to test two implications of the theory: workers should
know more about policies close to their occupational expertise, and
those with coordinating responsibilities should know more about
system-level questions. For the first, the Dutch LISS panel asks
respondents about parental leave and family policies. I relate
correct answers to these questions to the share of their occupational knowledge (from
O*NET) that lies in the domains each question draws on. I include respondent and
question fixed effects, so the comparison is whether the same person
does better on questions closer to her expertise. A
one-standard-deviation increase in this share is associated with 1.63
percentage points greater accuracy (\(p=.006\)), against a mean of 49 percent.\footnote{
  The outcome covers 26 family-policy components. The regressor is scaled
  by its standard deviation after removing respondent and component
  fixed effects.}

For the second, the ECB Consumer Expectations Survey records the duties
workers hold and tests their knowledge of inflation and the ECB's role.
In the pooled 2022--2025 waves, each of two coordinating duties a
worker reports, supervising a team and contributing to investment or
strategy decisions, is associated with 2.57 percentage points greater
accuracy on its eight questions, controlling for age, gender,
education, income, financial literacy, job characteristics, other
workplace duties, and country-by-year effects
(\(p<.001\)).\footnote{Job characteristics include employment type,
  industry, employer sector, and firm size; the other duties, price and
  wage setting and hiring, have small and imprecise coefficients.}
Within person, comparing years in which the same worker reports
different duties, the estimate is 1.30 points (\(p=.007\)). Neither
association is causal, but both run in the direction the theory
predicts. Online Appendix OF reports the data and specifications. 

I now turn to the link between policy understanding and democratic accountability.

\hypertarget{policy-understanding-and-electoral-accountability}{%
\section{Policy Understanding and Electoral
Accountability}\label{policy-understanding-and-electoral-accountability}}

\label{sec:political-competition}

The organization of production hands politics two things: the output
that funds policy and the knowledge with which voters judge it. To
connect them, consider an election in which two office-seeking
candidates \(j\in\{1,2\}\) choose policy effort, and the winner
implements her policy. Given \(\mathcal A\in\Omega(q)\), write
\(Y\equiv Y(\mathcal A)=\int_0^1y_i\,di=1/C(\mathcal A)\) for
aggregate personal income and \(B\equiv B(\mathcal A)=\int_0^1b_i\,di\)
for mean understanding. In this section, citizens' incomes \(y_i\)
and knowledge are taken as given.

\hypertarget{voter-evaluation-and-candidate-effort}{%
\subsection{Voter Evaluation and Candidate
Effort}\label{voter-evaluation-and-candidate-effort}}

Candidate \(j\) chooses effort \(e_j\geq0\), one choice for the whole
agenda, which determines how well she designs her policy platform.
The resources effectively available for public services rise with
both output and this effort: a fixed tax rate makes the public budget
proportional to \(Y\), with a factor I normalize to one,\footnote{At
  a fixed tax rate \(\tau\), the public budget is \([\tau/(1-\tau)]Y\),
  where \(Y\) is aggregate disposable income. I absorb this constant
  ratio into the units of \(e_j\).} and higher policy quality makes
each unit of that budget deliver more. If she wins, effective
public-service provision and voter \(i\)'s utility are
\[
R_j=Ye_j,
\qquad
u_i(j)=y_i+\log R_j,
\qquad
\int_0^1y_i\,di=Y.
\tag{10}\label{eq:policy-primitives}
\]
At the port, \(Y\) is the funding for a modernization and \(e_j\) the
quality of its planning, procurement, and implementation. This is the
common-interest environment: every citizen receives the same benefit
from better policy. Online Appendix OB allows a general
policy technology \(G(Y,e)\), including nonfiscal policies.

Citizens agree about desirable outcomes but may misjudge the platform's implications. Voter \(i\) perceives the payoff difference between the candidates' platform with an error of scale \(1/b_i\): the better she understands the
policy, the more accurate she is. Let
\(\eta_i\sim\mathrm{Logistic}(0,1)\). For \(b_i,e_1,e_2>0\), her perceived payoff difference and vote probabilities
are
\[
\begin{aligned}
\widetilde\Delta_i
&=u_i(1)-u_i(2)+\frac{\eta_i}{b_i}
=\log\frac{R_1}{R_2}+\frac{\eta_i}{b_i}
=\log\frac{e_1}{e_2}+\frac{\eta_i}{b_i},\\
\pi_{i1}
&=\Pr(\widetilde\Delta_i\geq0)
=\frac{e_1^{b_i}}{e_1^{b_i}+e_2^{b_i}},
\qquad \pi_{i2}=1-\pi_{i1}.
\end{aligned}
\tag{11}\label{eq:vote-probability}
\]
In the \(b_i=0\) limit she supports each candidate with probability
one half.\footnote{Appendix~\ref{app:proofs-politics} states the zero-effort
conventions.} Personal income and the size of the budget are common to
the two candidates and cancel from the comparison, although both enter
welfare.

Each candidate values office and bears a private effort cost \(c(e)\),
which she incurs whether or not she wins. A common popularity shock,
uniform over a range covering every feasible vote margin, makes her
probability of winning affine in her expected vote share; I normalize
its slope to one. With \(e_{-j}\) the opponent's effort and
\(\pi_{ij}(e_j,e_{-j})\) the support probability in
equation~\eqref{eq:vote-probability}, candidate \(j\) chooses
\[
\max_{e_j\geq0}\quad
\left\{
\int_0^1\pi_{ij}(e_j,e_{-j})\,di-c(e_j)
\right\}.
\tag{12}\label{eq:candidate-problem}
\]
The cost function is \(C^1\) at zero and \(C^2\) for positive effort,
with \(c(0)=c'(0)=0\) and \(c''(e)>0\) for \(e>0\), and effort uses no
productive labor or public-service resources. Online Appendix OB
derives the winning probability from the shock and discusses
extensions.
\hypertarget{equilibrium-policy-quality}{%
\subsection{Equilibrium Policy
Quality}\label{equilibrium-policy-quality}}

A voter's understanding makes her vote more responsive to improvements in
either candidate's policy. A
better-informed electorate therefore pushes candidates harder to
compete on policy, and average understanding \(B\) measures that
pressure.

\begin{proposition}[Understanding and policy effort]\label{thm:political-equilibrium}
The electoral game has a unique pure-strategy equilibrium. It is
symmetric, and the distribution of understanding affects equilibrium
effort only through its mean \(B\). The common effort is zero at
\(B=0\). For \(B>0\), it is positive, strictly increasing in \(B\),
and satisfies

\[
e^*(B)c'\!\left(e^*(B)\right)=\frac{B}{4}.
\tag{13}\label{eq:effort-condition}
\]

\end{proposition}

Better-informed voters give candidates a stronger reason to improve
policy. When both candidates choose \(e>0\), increasing one candidate's
effort raises voter \(i\)'s support probability at rate \(b_i/(4e)\).
The candidate therefore equates the aggregate electoral return
\(B/(4e)\) to her marginal cost. The full distribution of understanding
can matter away from a tie, but its mean becomes a sufficient statistic in equilibrium.\footnote{Here better understanding strengthens incentives to improve
a common policy outcome; the interaction with political selection can
give different welfare effects \citep{ashworth-bueno-de-mesquita-2014}.}
A related empirical link appears in
\citet{snyder-stromberg-2010}: weaker local press coverage reduces voter
knowledge, and less-covered representatives do less constituency work.

For the remaining analysis, I assume an isoelastic effort cost function, yielding the following:

\[
c(e)=\frac{\epsilon}{4}e^{1/\epsilon},
\qquad
e^*(B)=B^\epsilon,
\qquad
R(Y,B)=YB^\epsilon,
\qquad 0<\epsilon<1.
\tag{14}\label{eq:effective-resources}
\]

The elasticity \(\epsilon\) measures how strongly equilibrium policy
quality responds to aggregate understanding. Effective services thus depend
on both the economy's resources and the electoral incentives to use them
well. 

I evaluate social welfare by aggregating citizens' realized payoffs,
which depend on the quality delivered and not on their evaluation
errors.\footnote{The two candidates form a measure-zero subset of the
  population, so their private effort costs do not enter utilitarian
  welfare.} In equilibrium \(R=YB^\epsilon\), so a citizen's payoff is
\(y_i+\log Y+\epsilon\log B\), and integrating over the unit mass of
citizens, an organization \(\mathcal A\in\Omega(q)\) yields
\[
\begin{aligned}
W_0(\mathcal A)
&\equiv W_0\!\left(Y(\mathcal A),B(\mathcal A)\right)\\
&=
\underbrace{Y(\mathcal A)+\log Y(\mathcal A)}_{\substack{\text{welfare value}\\\text{of output}}}
+
\underbrace{\epsilon\log B(\mathcal A)}_{\substack{\text{quality return}\\\text{to understanding}}},
\end{aligned}
\tag{15}\label{eq:common-welfare}
\]
where \(W_0(Y,B)\equiv Y+\log Y+\epsilon\log B\), the subscript
\(0\) denotes the common-interest environment, and
\(W_0(Y,0)=-\infty\) is the limit as \(B\downarrow0\). The linear
term aggregates personal income, \(\log Y\) values the public budget,
and \(\epsilon\log B\) values the discipline that an informed
electorate imposes on government. The organization of
production must therefore be evaluated for both the resources it
generates and the knowledge it leaves with voters.

\hypertarget{competitive-implementation-and-the-social-optimum}{%
\section{Competitive Implementation and the Social
Optimum}\label{competitive-implementation-and-the-social-optimum}}
\label{sec:social-optimum}
Competitive wages reward knowledge for what it produces; what it
contributes to the quality of government accrues to everyone but is not priced. I first establish
which organizations competition supports, then how a utilitarian
planner maximizing welfare in equation~\eqref{eq:common-welfare} would
change responsibilities and knowledge.
\hypertarget{competitive-implementation}{%
\subsection{Competitive
Implementation}\label{competitive-implementation}}
Workers may choose any occupational profile \(s\in\Delta_K\), and
\(w(s)\) is the wage for one unit of labor supplied with profile
\(s\). Given this schedule, firms choose \(\mathcal A\in\Omega(q)\)
to minimize payroll per unit of output; the output price is one. A
positive-output competitive equilibrium requires that active firms
minimize payroll, that workers choose wage-maximizing profiles, that
the labor market clears, and that free entry drives
profits to zero. 

\begin{proposition}[Competitive implementation]\label{thm:competitive-implementation}
A feasible organization can be supported in a positive-output competitive equilibrium if and only if it minimizes the unit labor requirement. In any such equilibrium, equilibrium output is concentrated on organizations satisfying \(C(\mathcal A)=C^M(q)\); aggregate output is \(Y^M(q)\), and every occupied profile earns the common wage \(1/C^M(q)\).
\end{proposition}

Free occupational choice equalizes wages across the profiles used, so
competition implements the organizations that require the least labor.
Since an individual worker cannot affect political outcomes, she
chooses her profile to maximize her wage. Atomistic firms likewise
cannot affect political outcomes through their labor demand, so they
are taken to maximize profits. When firms are sizeable and maximize their workforce's utility, free entry by sufficiently small worker-owned firms still implies that only productively efficient organizations are implementable.\footnote{
Online Appendix~OC.1 establishes this result for firms
that distribute profits to their workers as dividends. If an incumbent
uses a labor-inefficient organization, an arbitrarily small group of its
workers can form an efficient firm and share the resulting income gain,
while the associated change in public-service provision vanishes with the
entrant's size.}

Competition supports every
organization in \(\mathcal M(q)\), including those outside the
minimum-handoff selection, without choosing among them for
understanding. Write \(C^M=C^M(q)\) and define the best understanding
available at maximal output by

\[
B^M
\equiv
\sup_{\substack{\mathcal A\in\Omega(q)\\C(\mathcal A)=C^M}}
\quad B(\mathcal A).
\]

The benchmark \(B^M\) gives competition the most favorable allocation
of knowledge consistent with maximal output. The market does not favor
this allocation, but raising understanding above it necessarily requires
\(C>C^M\).

\hypertarget{the-planners-trade-off}{%
\subsection{The Planner’s Trade-off}\label{the-planners-trade-off}}

The planner has two instruments. She can delegate less, so that more interface work is retained in higher layers, increasing their employment share; or she can
cross-train, staffing a narrow responsibility with workers whose
knowledge extends beyond it. Neither changes total work, \(1+d(q)\),
and each costs labor: a unit of work at scope \(r'\) rather than
\(r\) costs \(\ell_{r'}\) instead of \(\ell_r\), and a cross-trained
worker's fit falls below one, so her unit costs \(\ell_r/\beta>\ell_r\).
Either therefore raises labor per unit of output and lowers
\(Y(\mathcal A)=1/C(\mathcal A)\). In general, the planner solves
\[
\sup_{\substack{\mathcal A\in\Omega(q)\\B(\mathcal A)>0}}
\quad W_0(\mathcal A).
\]
The least labor at which understanding \(\bar B>0\) can be delivered
is the frontier
\[
\mathcal C_q(\bar B)
=
\inf_{\substack{\mathcal A\in\Omega(q)\\B(\mathcal A)\geq\bar B}}
\quad C(\mathcal A).
\tag{16}\label{eq:understanding-cost-frontier}
\]

This frontier records the least labor needed to deliver a given level
of understanding, allowing both the division of responsibilities and
workers' knowledge to adjust. The infimum is over finite organizations;
it need not be attained and equals \(+\infty\) if the target is
infeasible. Online Appendix OC establishes existence in a relaxed
formulation and illustrates why a finite optimum need not exist.

A feasible organization with outcomes \((C,B)\), where \(C\geq C^M\) and
\(B>B^M>0\), yields welfare above the supremum attainable by
labor-minimizing organizations if and only if\footnote{If \(B^M\) is not
  attained, equation \eqref{eq:reform-hurdle} compares the redesign with
  the supremum over labor minimizers. If \(B^M=0\), every feasible
  organization with positive understanding dominates them under
  logarithmic welfare.}

\[
W_0(1/C,B)>W_0(1/C^M,B^M)
\quad\Longleftrightarrow\quad
\epsilon\log\frac{B}{B^M}
>
\left(\frac{1}{C^M}-\frac{1}{C}\right)
+\log\frac{C}{C^M}.
\tag{17}\label{eq:reform-hurdle}
\]

The electoral gain from better understanding must compensate for both
the loss of private income, \(1/C^M-1/C\), and the reduction in public
services caused by a smaller budget, \(\log(C/C^M)\). The more
strongly candidate effort responds to understanding, as measured by
\(\epsilon\), the greater the output loss the understanding gain can justify.

\hypertarget{the-planners-delegation-cutoff}{%
\subsection{The Planner’s Delegation
Cutoff}\label{the-planners-delegation-cutoff}}

To solve the planner's problem, I impose two technical restrictions.
First, symmetry: all production domains are equally important, and the
agenda is \emph{symmetric and dimensional}, giving equal probability to
issues of the same scope and weighting their domains equally, with
arbitrary weights across scopes. Matched scope-\(r\) workers in a
productively optimal hierarchy then share understanding \(b_r\), given
in equation~\eqref{eq:proof-br}. Second, labor intensities are affine
in scope. These restrictions yield the cutoff
characterization below. In the theorem, scope-\(r\) workload refers to
the amount \(g_r\) in a productively optimal hierarchy.

\begin{theorem}[A cutoff representation of an optimum]\label{thm:optimal-responsibility}
Suppose \(q\) is uniform over all \(K\) knowledge domains; the agenda is symmetric and dimensional; and labor intensities satisfy \(\ell_r=1+\kappa(r-1)\), with \(0<\kappa<1\). The planner's common-interest problem then has an optimal organization characterized by a cutoff \(\widehat r\in\{0,\ldots,K-1\}\). All workload generated at scopes \(r>\widehat r\) is performed in system-wide positions, and all workload generated at scopes \(r<\widehat r\) is performed in exactly matched scope-\(r\) positions. The workload generated at scope \(\widehat r\) may be split between system-wide and exactly matched scope-\(\widehat r\) positions. The case \(\widehat r=0\) assigns all workload to system-wide positions.
\end{theorem}

Moving a unit of scope-\(r\) work to the root costs \(\ell_K-\ell_r\)
extra labor and raises labor-weighted understanding by
\(\ell_Kb_K-\ell_rb_r\). The planner ranks source scopes by the ratio,
the understanding obtained per unit of extra labor, and under the
symmetric agenda it weakly rises with \(r\): a cross-domain issue counts only
for a worker covering all of its domains, so each additional domain
adds weakly more understanding the more domains she already spans. The
planner therefore reassigns the broadest delegated work to system-wide
positions first, then moves down, scope by scope, until it no longer pays.

For any feasible labor requirement along the cutoff frontier, an organization with exactly matched knowledge delivers at least as much average understanding as any alternative, including cross-training narrow specialists, which is therefore not required. If the cutoff sits at \(\widehat r=K-1\), the resulting organization can be a labor minimizer. A sufficient condition for a strict improvement over any productive optimum is that the return from the first marginal reassignment exceeds its cost. Write
\[
\sigma_{K-1}
\equiv
\frac{\ell_Kb_K-\ell_{K-1}b_{K-1}}
{\ell_K-\ell_{K-1}}.
\]

Moving a unit of scope-$(K-1)$ work to system-wide positions raises labor-weighted understanding by $\ell_K b_K-\ell_{K-1}b_{K-1}$ and labor required per unit of output by $\ell_K-\ell_{K-1}$, so $\sigma_{K-1}$ is the marginal increase in labor-weighted understanding per unit of additional labor; average understanding therefore rises precisely when $\sigma_{K-1}>B^M$, while by equation~\eqref{eq:common-welfare} each
additional unit of labor costs income and public resources in proportion to
\(1+1/C^M\).
\begin{corollary}[Strict over-specialization]\label{cor:strict-over-specialization}
Under the assumptions of Theorem \ref{thm:optimal-responsibility}, suppose
\[
\epsilon\bigl(\sigma_{K-1}-B^M\bigr)
>
B^M\left(1+\frac{1}{C^M}\right).
\]
Then every labor minimizer---and hence every organization implemented
by a positive-output competitive equilibrium---is strictly dominated in
common-interest welfare. An optimal cutoff assigns to system-wide
positions some of the workload generated at scope \(K-1\), and hence
strictly more workload than any labor minimizer.
\end{corollary}
The condition says that the electoral value of the understanding the
first reassignment adds exceeds its cost in income and public
resources. Since \(B^M\) is the greatest understanding attainable at
\(C^M\), the conclusion holds whichever labor minimizer competition
implements. If every policy issue requires all \(K\) domains, then
\(b_{K-1}=0\) and \(\sigma_{K-1}=1/\kappa\), so the condition holds
for any \(\epsilon>0\) once \(\kappa\) is small enough: when only
system-wide knowledge counts for policy and broad coordination is
cheap, competition over-specializes. At the port, the planner keeps
port-wide positions responsible for work that firms would delegate
once its interfaces are resolved; more workers then acquire
system-wide knowledge through their jobs, at the cost of doing that
work less efficiently.

As policy becomes more integrated \citep{cejudo-trein-2023}, the
planner--market wedge widens: with more weight on cross-domain issues,
the planner assigns more work to system-wide positions and the welfare
cost of the market's organization grows. Let the agenda place weight
\(1-\chi\) on the \(K\) domain-specific issues, equally weighted, and
\(\chi\in[0,1]\) on a fixed distribution over cross-domain issues;
production is unaffected. At \(\chi=0\), every profile delivers
understanding \(1/K\), so every labor minimizer is socially optimal. As
\(\chi\) or \(\epsilon\) rises, the least and greatest amounts of
system-wide work among optimal cutoff allocations increase weakly, and
strictly within a unique partial-cutoff segment. The planner's welfare
advantage over the labor-minimizing benchmark also rises weakly with
\(\chi\); welfare itself need not, since cross-domain issues carry a
larger processing burden. Online Appendix OC states these results
precisely.

Less delegation buys understanding with output. Political institutions
might instead draw on the expertise production already creates, by
giving informed citizens more influence or by sharing their knowledge.
The next section introduces targeted spending, which changes both the
value of broadening knowledge and the case for either response.

\hypertarget{targeted-spending-and-the-politics-of-knowledge}{%
\section{Targeted Spending and the Politics of
Knowledge}\label{targeted-spending-and-the-politics-of-knowledge}}

\label{sec:institutions}

With common-interest policies, greater authority for informed citizens
or wider sharing of their knowledge strengthens accountability. Targeted
spending complicates both responses. A port expansion raises regional
capacity, but its access roads, procurement, and permits favor
particular groups; candidates then compete over who benefits as well as
policy quality \citep{lizzeri-persico-2001}. Because resources follow
electorally responsive groups
\citep*{lindbeck-weibull-1987,dixit-londregan-1996}, understanding also
confers a distributive advantage \citep{blumenthal-2026}. I examine how
it changes the value of authority for the informed, of knowledge
sharing, and of broader occupational knowledge.

\hypertarget{targeted-spending-and-political-influence}{%
\subsection{Targeted Spending and Political
Influence}\label{targeted-spending-and-political-influence}}

Fix an organization \(A\in\Omega(q)\). Each occupied position is a
separately targetable occupational group, indexed by \(g\). Group \(g\)
has mass \(m_g>0\), understanding \(b_g=b(s_g)\in(0,1]\), and formal
authority per person \(v_g\geq 0\).\footnote{Allowing candidates to
target individuals directly leaves all results unchanged. Since workers
within a position share \(b_g\) and \(v_g\), individual-level targeting
gives every member \(i\) the same equilibrium provision,
\(t_i=fRv_gb_g/B_v=t_g\).}
Authority may take the form of equal voting, issue-specific voting
weights, or a restricted franchise; \(v_g=0\) excludes the group.
Suppressing dependence on \(A\), normalize population and authority and
define politically weighted understanding by
\[
    \sum_g m_g=1,
    \qquad
    \sum_g m_gv_g=1,
    \qquad
    B_v=\sum_g m_gv_gb_g.
\]

Maintain \(0<\epsilon<1\). Let \(f\in[0,1]\) be the fraction of
equally valued public services whose beneficiaries candidates can choose;
the remaining services benefit everyone equally. Logarithmic
utilities make \(f\) also the resource share devoted to targetable services.\footnote{Let \(x_{gj}(a)\) denote group \(g\)'s provision of service
  \(a\in[0,1]\), with
  \(u_i(j)=y_i+\int_0^1\log x_{gj}(a)\,da\).
  With unit provision costs, total common provision \(R_j^c\) satisfies
  \(R_j^c+\sum_gm_gt_{gj}=R_j\). Concavity implies equal provision
  within each category. Holding relative allocations across groups fixed,
  every voter prefers \(R_j^c=(1-f)R_j\), so candidates choose this
  split. Each common service then provides \(R_j\), while each
  targetable service provides \(t_{gj}/f\) to group \(g\), giving
  equation~\eqref{eq:targetable-utility}.}
Candidate \(j\) chooses effort \(e_j\) and per-capita targeted
provision \(t_{gj}\) for each group \(g\). A member \(i\) of group \(g\) receives
\(u_i(j)=y_i+U_{gj}^{\mathrm{svc}}\), with
\(\int_0^1y_i\,di=Y\) and the following reduced form for utility
from public services:

\[
U_{gj}^{\mathrm{svc}}
=
\begin{cases}
\log R_j, & f=0,\\[2pt]
(1-f)\log R_j+f\log\dfrac{t_{gj}}{f}, & f>0,
\end{cases}
\qquad
R_j=Ye_j.
\tag{18}\label{eq:targetable-utility}
\]
Equal per-capita provision, \(t_{gj}=fR_j\), recovers the
common-interest benchmark. As \(f\) rises, more of the services citizens
already value become subject to distributive politics. For \(f>0\),
candidate \(j\) jointly chooses effort and provision across groups:

\[
\begin{aligned}
\max_{e_j\geq0,(t_{gj})_g\geq0}\quad
&\left\{\sum_gm_gv_g\pi_{gj}-c(e_j)\right\}\\
\text{subject to}\quad
&\sum_gm_gt_{gj}=fYe_j,
\end{aligned}
\tag{19}\label{eq:targeted-candidate-problem}
\]
where \(\pi_{gj}\) uses the evaluation technology in
Section~\ref{sec:political-competition}, applied to group \(g\)'s
utility with precision \(b_g\). At \(f=0\), the objective remains
\(\sum_gm_gv_g\pi_{gj}-c(e_j)\), with effort the only choice.
Online Appendix OD proves that the unique symmetric equilibrium
satisfies \(ec'(e)=B_v/4\). With the isoelastic cost and
\(t_g\equiv t_{g1}=t_{g2}\) for \(f>0\),

\[
e=B_v^\epsilon,
\qquad
R=YB_v^\epsilon,
\qquad (f>0):\quad
t_g=fR\frac{v_gb_g}{B_v}.
\tag{20}\label{eq:targeted-allocation}
\]
Policy understanding now affects both the quality and the distribution
of public spending: \(B_v\) governs effort, while \(v_gb_g\) governs
targeted spending per person. Under equal voting, \(B_v=B\) and
\(t_g=fRb_g/B\): better-informed groups receive more, as their votes
respond more to spending. Since utility from targeted spending is
concave, an equal per-capita allocation of the same total would yield
higher welfare.

\hypertarget{political-authority-and-information-sharing}{%
\subsection{Political Authority and Information
Sharing}\label{political-authority-and-information-sharing}}

Giving informed groups greater authority strengthens effort incentives,
but allows such groups to attract more spending. The planner balances these effects by choosing \(v=(v_g)_g\):

\[
\begin{aligned}
\max_{v}\quad
&Y+\log Y+(\epsilon-f)\log B_v
+f\sum_gm_g\log(v_gb_g)\\
\text{subject to}\quad
& v_g\geq0 \quad\text{for every }g,\\
& \sum_gm_gv_g=1.
\end{aligned}
\tag{21}\label{eq:authority-policy}
\]
Related considerations are faced by the most knowledgeable groups when they evaluate disclosing a policy analysis that would raise the policy understanding of the less informed: sharing increases equilibrium policy quality by improving other groups' evaluations but it also strengthens their claims on spending. The next proposition explores these trade-offs.

\begin{proposition}[Political authority and information sharing]\label{prop:institutions}
The two political responses to unequal understanding have the following properties:

\begin{enumerate}
\renewcommand{\labelenumi}{\textup{(\roman{enumi})}}
\item \textup{\textbf{Authority.}} At \(f=0\), every optimum assigns formal authority only to groups with maximal understanding. For \(f>0\), the optimum is unique and interior; denote it by \(v^*\). If \(b_g>b_h\), then \(v_g^*>v_h^*\) when \(0<f<\epsilon\), \(v_g^*=v_h^*=1\) when \(f=\epsilon\), and \(v_g^*<v_h^*\) when \(f>\epsilon\). In every case with \(f>0\), group \(g\) has greater per-person effective influence: \(v_g^*b_g>v_h^*b_h\).

\item \textup{\textbf{Information sharing.}} Suppose group \(h\) can provide a verifiable policy analysis that raises every other group's understanding and changes politically weighted understanding from \(B_v\) to \(B_v'>B_v\), holding the occupational distribution, authority, output, \(f\), and \(b_h\) fixed; when \(f>0\), assume \(v_h>0\). Group \(h\)'s gross utility gain is

\[
(\epsilon-f)\log\frac{B_v'}{B_v}.
\tag{22}\label{eq:information-sharing}
\]

At a communication cost \(k\geq0\), the group strictly prefers to share when this expression exceeds \(k\), strictly prefers to withhold when it is below \(k\), and is indifferent at equality. Hence, with no communication costs, the analysis is shared when \(f<\epsilon\), and withheld when \(f>\epsilon\).
\end{enumerate}
\end{proposition}

In equilibrium, equations~\eqref{eq:targetable-utility} and
\eqref{eq:targeted-allocation} give group \(g\)'s service utility as
\(\log Y+(\epsilon-f)\log B_v+f\log(v_gb_g)\): understanding
elsewhere in the electorate improves policy for everyone, with
elasticity \(\epsilon\), but dilutes the group's share of targeted
spending, which has weight \(f\). When \(f>\epsilon\) the second
effect dominates. The planner then gives better-informed groups less
formal authority, since their understanding already attracts spending,
though their effective influence \(v_gb_g\) remains greater; and an
informed group withholds even costless, verifiable analysis, since
better policy is worth less to it than the spending it would forgo.
The social return keeps its sign: under equal voting, raising a
below-average group's understanding at fixed output raises welfare for
every \(f\). Since neither authority nor
communication closes the gap, I return to its source, the organization
of production.

\hypertarget{the-organization-of-production-with-targeted-spending}{%
\subsection{The Organization of Production with Targeted
Spending}\label{the-organization-of-production-with-targeted-spending}}

I now return to equal voting, \(v_g=1\) for every group. An organization
\(\mathcal A\in\Omega(q)\) induces occupational groups with masses
\(m_g(\mathcal A)\), understanding levels \(b_g(\mathcal A)\), and mean
understanding \(B(\mathcal A)\). When every positive-mass group has
positive understanding, define

\[
\Delta(\mathcal A)
\equiv
\log B(\mathcal A)
-\sum_gm_g(\mathcal A)\log b_g(\mathcal A)
\geq0.
\]
The gap \(\Delta\) measures inequality in understanding and vanishes
when all groups understand equally well. Because targeted
provision is proportional to understanding, \(f\Delta\) is the welfare
loss from its unequal allocation.\footnote{If the targetable groups were instead exogenous and
ascriptive, group identity would have no productive role, so the planner
could assign each group proportionally across positions. All groups
would then have mean understanding \(B\), implying \(t_g=fR\) and
eliminating the allocative loss; the planner would therefore choose a
common-interest optimum. Competition could support but would not select
this independent assignment, because group identity would affect
neither labor demand nor wages.} The planner chooses \(\mathcal A\in\Omega(q)\) to maximize

\[
\begin{aligned}
W_f(\mathcal A)
&=
\underbrace{Y(\mathcal A)+\log Y(\mathcal A)}_{\substack{\text{welfare value}\\\text{of output}}}
+
\underbrace{\epsilon\log B(\mathcal A)}_{\substack{\text{quality return}\\\text{to understanding}}}
-
\underbrace{f\Delta(\mathcal A)}_{\substack{\text{allocative loss}\\\text{from targeted spending}}}\\
&=W_0(\mathcal A)-f\Delta(\mathcal A).
\end{aligned}
\tag{23}\label{eq:targeted-organization-welfare}
\]
Understanding now has two welfare margins: policy quality and allocative efficiency. 

As more provision becomes targetable, the planner is willing to accept a lower common-interest payoff to reduce inequality in understanding. Among
organizations with positive understanding in every group, if optima
exist at \(0\leq f<f'\leq1\), every optimum at \(f'\) has weakly
lower \(\Delta\) and \(W_0\) than every optimum at \(f\). At fixed
output and mean understanding, any feasible, nontrivial mean-preserving
equalization strictly raises welfare when \(f>0\). Online Appendix OD
proves these comparisons.

Under the conditions of Theorem~\ref{thm:optimal-responsibility},
broadening responsibilities is the least costly way to improve average
understanding. Targeted spending creates a further return to learning:
holding output and group sizes fixed, the marginal welfare gain per
worker from raising \(b_g\) is \((\epsilon-f)/B+f/b_g\). The
\(f/b_g\) term gives learning greater marginal value in less-informed
groups, which can justify cross-training despite its cost in productive
fit.

The next result characterizes an optimal organization for
\(0\leq f\leq\epsilon\).\footnote{For \(f>\epsilon\), greater
\(B\) lowers welfare when output and mean log understanding are held
fixed, so the profile-dominance argument used below no longer applies.} An organization is \emph{pruned} if it omits unused scopes and
\emph{broad-to-narrow} if the domains required by active responsibilities
are nested along each delegation path, possibly skipping scopes. The
construction below may join several such component hierarchies as parallel
branches of one organization.

\begin{theorem}[An optimal organization with targeted spending]\label{thm:targetable-production}
Suppose voting weights are equal, \(q=(1/K,\ldots,1/K)\), the agenda is
symmetric and dimensional with positive single-domain weight, and
\(\ell_r=1+\kappa(r-1)\), with \(0<\kappa<1\).
For every \(f\in[0,\epsilon]\), some optimum is implemented by one
organization containing at most three pruned, broad-to-narrow component
hierarchies; at the endpoints, two suffice. In each component hierarchy,
every responsibility profile is uniform on the domains it requires, and
the scope-\(K\) position is exactly matched. Every positive-work
scope-\(r<K\) responsibility is exactly matched
or staffed by a cross-trained worker who assigns a total share in
\((r/K,1)\) equally across its required domains and the remainder equally
across the other domains. Within each component hierarchy, understanding
increases weakly with scope.
\end{theorem}

The planner can attain an optimum with a few parallel component
hierarchies in which responsibilities narrow down each branch but
knowledge can extend beyond the job. Cross-trained workers place at least
as much knowledge in each required domain as in each other domain, while
learning about the latter to improve their policy understanding at some
productive cost
(Figure~\ref{fig:targetable-profiles}).

\begin{figure}[H]
\centering
\includegraphics[width=\linewidth]{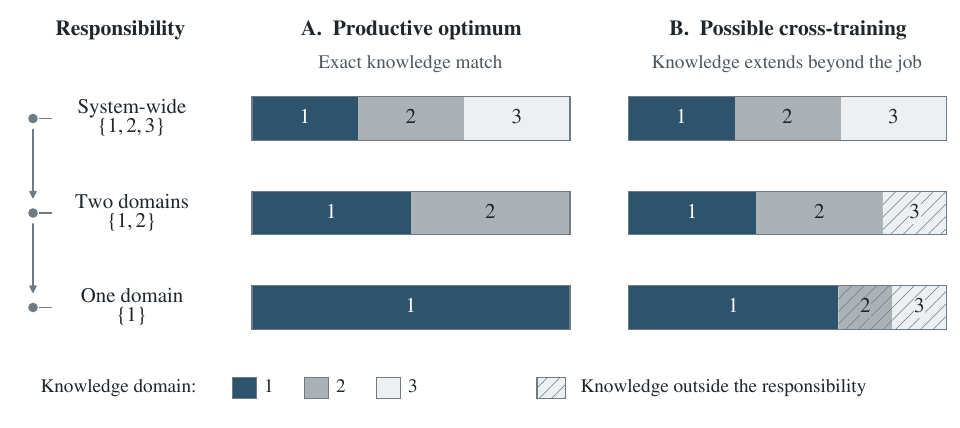}
\caption{Responsibilities and knowledge along a nested branch, from system-wide coordination at the top to specialist work at the bottom. Each bar allocates a worker's unit knowledge stock across domains. Panel A shows exact matching in a labor-minimizing organization. Panel B illustrates cross-training within the bounds of Theorem~\ref{thm:targetable-production}; the shares are illustrative, and the theorem also permits exact matching.}
\label{fig:targetable-profiles}
\end{figure}
\FloatBarrier

Cross-training need not occur at every optimum; Online Appendix OD
gives a three-domain example where the global optimum strictly
cross-trains both nonroot groups. Outside knowledge need not vary
monotonically with scope, although understanding rises weakly within
each hierarchy. 

So far the technology of production has been fixed: organizations and
institutions have varied around it. By changing which work humans
perform, technology shocks change what humans must know, and through it
both accountability and the distribution of public spending, even when
every displaced worker is re-employed. The next section analyzes a
leading case. AI substitutes for human work within responsibilities, and
where firms place it decides which knowledge remains human.

\hypertarget{artificial-intelligence-and-human-responsibility}{%
\section{Technology Shocks: The Case of AI}\label{artificial-intelligence-and-human-responsibility}}

\label{sec:ai}

AI reallocates work within knowledge hierarchies.
\citet{ide-talamas-2025} trace the effects of this reallocation on
output and labor incomes; here it also changes what the remaining human
jobs require, and so voters' expertise. At the port, automating
inspections leaves humans to coordinate; automating coordination leaves
them narrower jobs. Both raise output but move policy understanding in
opposite directions, so firms and the planner may favor different uses
of AI.

\hypertarget{the-use-of-ai-and-human-understanding}{%
\subsection{The Use of AI and Human
Understanding}\label{the-use-of-ai-and-human-understanding}}

AI substitutes for human work within a responsibility. Because it can be
matched to any domain mix, the capacity it needs depends only on scope:
automating one unit of work at scope \(r\) requires
\(a_r^{\mathrm{AI}}>0\) units, with
\(0<a_1^{\mathrm{AI}}\leq\cdots\leq a_n^{\mathrm{AI}}\) and \(n=r(q)\).
Given capacity \(A\) per unit of output, the firm chooses an
organization and the work \(\xi_v\in[0,z_v]\) to automate at each node
so as to minimize the labor left to humans:
\[
\begin{aligned}
\min_{\mathcal A\in\Omega(q),\,\xi}\quad
&\sum_{v\in T}(z_v-\xi_v)a(p_v,s_v)\\[-2pt]
\text{subject to}\quad
&0\leq\xi_v\leq z_v\quad(v\in T),
\qquad
\sum_{v\in T}a_{r(p_v)}^{\mathrm{AI}}\xi_v\leq A.
\end{aligned}
\]
Proposition~\ref{prop:ai-sorting} shows that a global minimum can always be
attained by adopting AI in a productively optimal pre-AI hierarchy, where scope \(r\)
performs \(g_r(q)\) units of work per unit of output with average
understanding \(b_r\), weakly increasing in scope by
Theorem~\ref{thm:division-understanding}.

Writing \(x_r\) for the work automated at scope \(r\), the set of feasible
allocations on this hierarchy is
\[
\mathcal X(A)
=
\left\{
x\in\mathbb R_+^n:
x_r\leq g_r(q)\ \text{for every }r,\quad
\sum_{r=1}^n a_r^{\mathrm{AI}}x_r\leq A
\right\},
\qquad
0<A<\sum_{r=1}^n a_r^{\mathrm{AI}}g_r(q).
\]
One unit of AI capacity at scope \(r\)
releases \[
\rho_r\equiv\frac{\ell_r}{a_r^{\mathrm{AI}}}
\] units of human labor. Thus \(\rho_r\) measures AI’s comparative
advantage at that scope, and the technology is \textit{scope-neutral} when all
ratios are equal. Their ordering is an empirical question: AI’s
productivity gains differ across workers and activities
(\citealt*{brynjolfsson-li-raymond-2025};
\citealt{dellacqua-et-al-2026}).

Displaced workers are re-employed in the remaining positions, exhausting
the unit labor supply, and understanding reflects the knowledge these
jobs require. Since all positions are exactly matched, the firm is
indifferent to which ones within a scope lose work; assume automated
work is drawn from them in proportion to their work, so scope \(r\)
keeps average understanding \(b_r\) (Online Appendix OE allows targeting
within a scope). Human labor per unit of output, output, and average
understanding are therefore

\[
\begin{aligned}
C_H(x)
&=\sum_{r=1}^n\ell_r[g_r(q)-x_r],
&
Y(x)&=\frac{1}{C_H(x)},\\
B(x)
&=\frac{\sum_{r=1}^n\ell_r[g_r(q)-x_r]b_r}
{C_H(x)}.
\end{aligned}
\tag{24}\label{eq:ai-outcomes}
\]

Automation expands output by reducing \(C_H(x)\), while its effect on
understanding depends on the jobs in which humans are re-employed.
Assume positive \(B(x)\) at every feasible allocation. On the fixed hierarchy, the firm maximizes output and the planner
maximizes common-interest welfare, with \(f=0\) and equal voting:

\[
\min_{x\in\mathcal X(A)}\quad C_H(x)
\qquad\text{and}\qquad
\max_{x\in\mathcal X(A)}\quad
\left\{
F\!\left(C_H(x)\right)+\epsilon\log B(x)
\right\},
\quad
F(C)=\frac1C-\log C.
\tag{25}\label{eq:ai-objectives}
\]

\begin{proposition}[AI use and human understanding]\label{prop:ai-sorting}
A globally cost-minimizing AI implementation exists on a productively optimal pre-AI hierarchy. On that hierarchy, the firm assigns capacity in descending order of \(\rho_r\); if these ratios are distinct, at most one scope is partially automated. On that fixed hierarchy, with displacement proportional within each scope, expanding capacity weakly lowers average policy understanding when AI automates broader before narrower responsibilities and weakly raises it when the order reverses. Under scope neutrality, every placement using the same amount of capacity on that hierarchy produces the same output, and the planner assigns capacity first to the scopes with the lowest \(b_r\).
\end{proposition}

The firm places AI where it releases the most labor. Automating the
broadest responsibilities first shifts human employment toward narrow
jobs, while automating component work shifts it toward coordinators.
Under scope neutrality, output is the same either way, and the planner
can increase the political knowledge in society at no cost. Which tasks remain
with labor is the question studied by \citet{acemoglu-restrepo-2019};
here those tasks also fix the expertise workers acquire, so automation
reallocates the knowledge available to voters even when every displaced
worker is re-employed.

\hypertarget{when-the-firm-and-the-planner-use-ai-differently}{%
\subsection{When the Firm and the Planner Use AI
Differently}\label{when-the-firm-and-the-planner-use-ai-differently}}

Keeping broad responsibilities human can remain socially preferable even
when AI saves more labor at broad scopes. Let
\(\rho_r(\zeta)=\rho[1+\zeta\delta_r]\) and
\(a_r^{\mathrm{AI}}(\zeta)=\ell_r/\rho_r(\zeta)\), where \(\rho>0\)
and \(0=\delta_1<\cdots<\delta_n=1\). At \(\zeta=0\), labor savings
are scope-neutral; a positive \(\zeta\) gives firms a strict incentive
to automate broader work first.

\begin{theorem}[A reversal of the firm's AI allocation]\label{thm:ai-reversal}
Suppose \(b_1<b_n\), every feasible allocation leaves positive human labor and policy understanding, and \(A\leq a_r^{\mathrm{AI}}(\zeta)g_r(q)\) for \(r\in\{1,n\}\) and every \(\zeta\) in a neighborhood of zero. For every sufficiently small \(\zeta>0\), the firm has a globally cost-minimizing allocation on a productively optimal hierarchy that assigns all capacity to the broadest scope \(n\); assigning it to scope \(1\) on the same hierarchy yields strictly higher common-interest welfare. If \(b_1<b_r\) for every \(2\leq r\leq n\), then, holding the pre-AI hierarchy fixed and maintaining proportional displacement within each scope, the planner assigns all capacity to scope \(1\) for every sufficiently small \(\zeta>0\).
\end{theorem}

When \(\zeta\) is small, automating broad rather than narrow work gains
the firm little output, while the shift of human employment toward
narrow jobs lowers understanding by an amount that does not vanish with
\(\zeta\). For fixed \(A\) and \(\epsilon>0\), automating narrow work on
the same hierarchy therefore yields higher welfare. The advantage at
scope neutrality vanishes as \(A\) or \(\epsilon\) goes to zero; Online
Appendix OE gives its limit and conditions under which further
automation reduces effective public services or welfare. AI adoption
thus affects accountability through the knowledge that human work
continues to require.

\hypertarget{conclusion}{%
\section{Conclusion}\label{conclusion}}

\label{sec:conclusion}

How firms divide responsibilities determines which parts of an
interdependent production system workers must understand and who must
learn how the parts fit together. Workers bring this knowledge to the
ballot, where its value depends on what policy demands, so the division
of labor shapes accountability before any campaign begins. Wages price
knowledge for production, not policy evaluation; markets therefore
over-specialize when cross-domain issues matter enough, and most workers
become, as Smith feared, ``incapable of judging'' cross-domain policies.

Political institutions are no free remedy: authority for the
knowledgeable can worsen inequalities, and conflicts of interest can
block knowledge sharing. Hence the civic case for liberal education
\citep{nussbaum-2010,zakaria-2015}: learning beyond one's occupation
equips citizens to judge consequences their work gives them no reason to
understand, a value wages ignore. Broad education thus complements
specialized jobs, freeing citizens from dependence on such groups.

The theory suggests that the effect of technological change on democracy depends on which knowledge it substitutes or complements: in the case of AI, automating specialized
tasks shifts employment toward coordinators and raises understanding,
automating coordination lowers it. The direction markets reward need not be right one.

Open questions remain at each link of the mechanism. The first is
whether occupational knowledge has civic returns, as schooling does
\citep{dee-2004,milligan-moretti-oreopoulos-2004}, coming from what jobs require rather
than whom they select: when reorganization or automation changes a job's
responsibilities but not the worker, do cross-domain judgment and the
responsiveness of votes to policy quality change with them? Upstream, large firms and organized occupations could internalize part of the externality, narrowing or widening it; in the long run, the
question is how occupational knowledge adjusts as work and policy change
and whether education and public information compensate. None of this
alters the premise that the knowledge citizens bring to the ballot is
produced by how the economy organizes work; whether a democracy keeps
the gains from specialization while equipping citizens to evaluate
policy is decided where production is organized.

\clearpage

\begin{spacing}{1.15}
\renewcommand{\bibsection}{%
  \hypertarget{references}{\section*{References}\label{references}}%
  \addcontentsline{toc}{section}{References}}
\bibliographystyle{apalike}
\bibliography{references}
\end{spacing}

\clearpage
\appendix
\numberwithin{equation}{section}
\counterwithin{lemma}{section}

\hypertarget{production-and-understanding}{%
\section{Production and
Understanding}\label{production-and-understanding}}

\label{app:proofs-production}

We first solve the scope-frontier problem and construct a tree attaining
all frontiers. We then characterize the minimum-handoff selection and
use its employment weights to compare understanding across scopes.

\hypertarget{proof-of-theorem}{%
\subsection{\texorpdfstring{Proof of Theorem
\ref{thm:efficient-organization}}{Proof of Theorem }}\label{proof-of-theorem}}

\label{app:proof-efficient}

Write \(I=\operatorname{supp}q\) and \(n=|I|\geq2\). Since
\(P\geq0\) and \(\mathbb E[P]=q\) imply \(P_i=0\) almost
surely outside \(I\), all coordinate calculations below are on \(I\).
For \(1\leq r\leq n\), the scope-\(r\) frontier is
\begin{equation}
\begin{aligned}
\Phi_r(q)
=\min_P\quad & \mathbb E\lVert P\rVert_2^2\\
\text{subject to}\quad
&\mathbb E[P]=q,\\
&|\operatorname{supp}P|\leq r\quad\text{almost surely}.
\end{aligned}
\label{eq:proof-scope-frontier}
\end{equation}
Here \(P\) ranges over random profiles in \(\Delta_K\). By coherence,
any cut \(\mathcal C\) satisfies \(\sum_{v\in\mathcal C}\omega_v=1\)
and \(\sum_{v\in\mathcal C}\omega_vp_v=q\). Thus selecting
\(p_v\) with probability \(\omega_v\) gives the mean restriction
in \eqref{eq:proof-scope-frontier}. Since \(d(P)=1-\|P\|_2^2\),
the greatest interface load that can remain at scope \(r\) or below
is \(1-\Phi_r(q)\). We also use work conservation:
summing the model's node-work identities over the tree cancels all
nonroot interface loads and gives \(\sum_vz_v=1+d(q)\).

\begin{lemma}[Scope frontiers and nested attainment]
\label{lem:proof-scope-frontiers}
For \(1<r<n\), let \(\theta_r>0\) be the unique solution of
\[
\sum_{i\in I}\min\{1,q_i/\theta_r\}=r.
\]
Set \(\theta_1=1\) and
\(\theta_n=\min_{i\in I}q_i\). At every scope, a frontier
minimizer \(P\) satisfies
\begin{equation}
\pi_i^r
=
\Pr(P_i>0)
=
\min\left\{1,\frac{q_i}{\theta_r}\right\},
\qquad
\sum_{i\in I}\pi_i^r=r.
\label{eq:proof-frontier-inclusion}
\end{equation}
Every frontier-attaining random profile
has exactly \(r\) active coordinates almost surely. Conditional on
\(P_i>0\), its coordinate \(i\) is almost surely equal to
\begin{equation}
 w_i^r=\frac{q_i}{\pi_i^r}=\max\{q_i,\theta_r\},
\qquad
\Phi_r(q)=\sum_{i\in I}q_iw_i^r.
\label{eq:proof-frontier-weights}
\end{equation}
Every frontier minimizer satisfies \(\|P\|_2^2=\Phi_r(q)\)
almost surely. A finite coherent tree
attains the frontiers successively from scope \(n\) to scope \(1\),
with aggregate workloads
\begin{equation}
g_1(q)=1,
\qquad
g_r(q)=\Phi_{r-1}(q)-\Phi_r(q)
\quad (2\leq r\leq n).
\label{eq:proof-frontier-increments}
\end{equation}
\end{lemma}

\bgroup \par\noindent\textbf{Proof.}\ 
We first derive and solve a lower-bound problem in inclusion probabilities.
We then implement its solution by finitely many profiles and join those
profiles across scopes into one coherent tree.

\emph{Step 1. Derive a lower bound on concentration.}
For a feasible \(P\), let \(\pi_i=\Pr(P_i>0)\). The mean
restriction gives \(q_i=\pi_i\mathbb E[P_i\mid P_i>0]\), so
\(\pi_i>0\) and \(\mathbb E[P_i\mid P_i>0]=q_i/\pi_i\).
Conditional Jensen's inequality therefore yields
\begin{equation}
\mathbb E[P_i^2] = \pi_i\mathbb E[P_i^2\mid P_i>0]
\geq
\pi_i\left(\frac{q_i}{\pi_i}\right)^2
=\frac{q_i^2}{\pi_i}.
\label{eq:proof-conditional-jensen}
\end{equation}
Equality holds if and only if \(P_i=q_i/\pi_i\) almost surely
conditional on \(P_i>0\). Moreover, \(0\leq P_i\leq1\) implies \(q_i\leq\pi_i\leq1\),
and the bound on scope implies
\[
\sum_{i\in I}\pi_i
=\mathbb E\left[\sum_{i\in I}\mathbf 1_{\{P_i>0\}}\right]
=\mathbb E|\operatorname{supp}P|\leq r.
\]
Summing \eqref{eq:proof-conditional-jensen} and imposing only these
marginal restrictions gives a lower bound; joint implementation by
profiles remains to be shown. The sum constraint binds, since if
\(\sum_i\pi_i<r\), some \(\pi_i<1\) can be increased to lower
the objective. Thus the relaxed problem is
\begin{equation}
\begin{aligned}
\min_{(\pi_i)_{i\in I}}\quad
&\sum_{i\in I}\frac{q_i^2}{\pi_i}\\
\text{subject to}\quad
&q_i\leq\pi_i\leq1\quad\text{for every }i\in I,\\
&\sum_{i\in I}\pi_i=r.
\end{aligned}
\label{eq:proof-inclusion-problem}
\end{equation}
Its feasible set is nonempty and compact: for example,
\(\pi_i=q_i+[(r-1)/(n-1)](1-q_i)\) is feasible. Strict convexity
of the objective gives a unique minimizer.

\emph{Step 2. Solve for the minimizing probabilities \(\pi_i\).}
At \(r=1\), the bounds force \(\pi_i=q_i\); at \(r=n\),
they force \(\pi_i=1\). At the latter endpoint any
\(0<\theta\leq\min_iq_i\) gives this vector; the lemma fixes
one convention. For \(1<r<n\), the function
\(F(\theta)=\sum_i\min\{1,q_i/\theta\}\) decreases continuously
and strictly from \(n\) to \(1\) on \([\min_iq_i,1]\), since
at least one term strictly decreases for \(\theta>\min_iq_i\).
Hence there is a unique \(\theta_r\) with \(F(\theta_r)=r\).
The vector \(\pi_i=\min\{1,q_i/\theta_r\}\) satisfies the
convex problem's optimality conditions. Its lower bounds are slack;
the multiplier on \(\sum_i\pi_i=r\) is \(\theta_r^2\), and
the upper-bound multiplier is \(0\) when \(\pi_i<1\) and
\(q_i^2-\theta_r^2\geq0\) when \(\pi_i=1\).
These values satisfy stationarity and complementary slackness, proving
optimality and \eqref{eq:proof-frontier-inclusion}.

\emph{Step 3. Attain the bound and characterize equality.}
Let \(F_r=\{i\in I:q_i\geq\theta_r\}\); these coordinates
are included with probability one. The remaining probabilities
\(q_i/\theta_r\) lie in \([0,1]\) and sum to the integer
\(r-|F_r|\). They are therefore a finite mixture of zero-one
vectors with that sum. To verify this, any nonintegral vector with
integer sum has two fractional coordinates. Moving mass between them
in either direction until one reaches a boundary expresses it as a
mixture of two vectors with fewer fractional coordinates. Iteration
terminates at zero-one vectors, which specify the desired subsets.
Choose such subsets, include \(F_r\), and assign the positive
weights \(w_i^r\). Equation \eqref{eq:proof-frontier-inclusion} gives
\(\sum_{i\in F_r}q_i+(r-|F_r|)\theta_r=1\),
so every profile has exactly \(r\) positive coordinates and sums to one.
Its coordinate means are \(\mathbb E[P_i]=\pi_i^rw_i^r=q_i\), and its positive weights are
constant conditional on inclusion. It attains every Jensen bound,
proving \eqref{eq:proof-frontier-weights}.

Conversely, every minimizer must attain the unique relaxed minimum
and every Jensen bound. Its support has size at most \(r\) and
expected size \(r\), hence size exactly \(r\) almost surely.
It includes \(F_r\) and has the stated positive weights, so
\[
\|P\|_2^2=\sum_{i\in F_r}q_i^2+(r-|F_r|)\theta_r^2
=\sum_iq_iw_i^r=\Phi_r(q)
\quad\text{almost surely}.
\]
In particular, \(\Phi_1(q)=1\) and \(\Phi_n(q)=\|q\|_2^2\).

\emph{Step 4. Construct a coherent tree attaining every frontier.}
Take any \(r\)-element support \(S\supseteq F_r\).
Its frontier profile puts
weight \(w_i^r\) on each \(i\in S\) and zero elsewhere.
For \(2\leq r\leq n\), given such a support, delete coordinate
\(i\in S\) with probability
\begin{equation}
\delta_i^r(S)
=1-\frac{w_i^r}{w_i^{r-1}}.
\label{eq:proof-deletion}
\end{equation}
The thresholds strictly decrease with scope, so
\(w_i^{r-1}\geq w_i^r>0\) and \(0\leq\delta_i^r(S)<1\).
Also \(F_{r-1}\subseteq F_r\subseteq S\), and coordinates in
\(F_{r-1}\) have equal weights at both scopes and cannot be deleted.
For the others, \(w_i^{r-1}=\theta_{r-1}\). Thus
\[
\sum_{i\in S}\frac{w_i^r}{w_i^{r-1}}
=|F_{r-1}|+
\frac{1-\sum_{i\in F_{r-1}}q_i}{\theta_{r-1}}=r-1,
\]
where the last equality uses \eqref{eq:proof-frontier-inclusion}
at scope \(r-1\).
Since \(|S|=r\), the deletion probabilities sum to one.
A positive-probability deletion leaves \(r-1\) coordinates containing
\(F_{r-1}\). Assign them the weights \(w^{r-1}\). For
each \(k\in S\), expected child weight is
\((1-\delta_k^r(S))w_k^{r-1}=w_k^r\).
Hence the mean child profile is the parent profile. Giving the child
after deletion of \(i\) scale \(\omega_v\delta_i^r(S)\), and
omitting zero-probability children, verifies coherence. Iterating from
\(q\) to scope one gives a finite tree, with depth \(n-1\) and
at most \(r\) children per scope-\(r\) node. Each layer attains
its frontier by the weights just established.

Every scope-\(r\) parent has squared norm \(\Phi_r(q)\), and all
its children have squared norm \(\Phi_{r-1}(q)\). For \(r\geq2\),
the parent's interface work is therefore
\begin{equation}
z_v
=\omega_v[(1-\Phi_r(q))-(1-\Phi_{r-1}(q))]
=\omega_vg_r(q).
\label{eq:proof-proportional-work}
\end{equation}
Layer scales sum to one, so aggregate work is \(g_r(q)\).
At scope one, \(d(p)=0\) and terminal work is \(z_v=\omega_v\),
giving \(g_1(q)=1\).
\hfill\rule{0.5em}{0.5em}\par\egroup

\medskip
\bgroup \par\noindent\textbf{Proof of Theorem \ref{thm:efficient-organization}.}\ 
We first establish exact matching and single-domain execution. We then
use the frontier lemma to identify the unique cost-minimizing workloads,
characterize the minimum-handoff selection, and prove the strict ordering
of workloads by scope.

\emph{Step 1. Establish exact matching and single-domain execution.}
For any responsibility and knowledge profiles, fit satisfies
\begin{equation}
\beta(s,p)
\leq
\sum_{k:p_k>0}p_k\frac{s_k}{p_k}
=\sum_{k:p_k>0}s_k
\leq1.
\label{eq:proof-exact-match}
\end{equation}
Equality requires all active ratios to be one and no knowledge outside
the responsibility's support, hence \(s=p\). Assigning each node
its own matched position attains the lower bound
\(\ell_{r(p_v)}z_v\) at every node. Since positions can be split
freely, every labor minimizer must attain it wherever \(z_v>0\).
A matched terminal of scale \(\omega\) and scope \(r>1\) costs
\(\omega\ell_r[1+d(p)]\). Instead, let it resolve its interfaces
and delegate execution to pure-domain children of scales
\(\omega p_k\). This is coherent and costs
\(\omega[\ell_rd(p)+\ell_1]\), saving
\(\omega(\ell_r-\ell_1)>0\). Thus every minimizing tree has
only scope-one terminals.

\emph{Step 2. Derive workload bounds attained by every labor minimizer.}
In such a coherent tree with exact matching, let \(Z_r\)
be the sum of the work \(z_v\) at nodes whose scope is \(r\).
Define workload tails and successive labor-cost increments by
\begin{equation}
T_h=\sum_{r=h}^nZ_r,
\qquad
\Delta_h=\ell_h-\ell_{h-1}>0.
\label{eq:proof-workload-tails}
\end{equation}
Coherence implies that child supports are contained in parent supports:
\(p_{vi}=0\) forces \(\sum_c\omega_cp_{ci}=0\).
For \(r<n\), the first node of scope at most \(r\) on each
path therefore defines a cut. Write \(P_r^{\mathrm{cut}}\) for
its scale-weighted profile; it has mean \(q\) and scope at most
\(r\). The nodes above the cut are exactly those of scope greater
than \(r\). Summing their work telescopes to root interface load
minus cut interface load, so
\begin{equation}
\sum_{j=r+1}^nZ_j
=\mathbb E\lVert P_r^{\mathrm{cut}}\rVert_2^2-\lVert q\rVert_2^2
\geq\Phi_r(q)-\lVert q\rVert_2^2
=\sum_{j=r+1}^ng_j(q).
\label{eq:proof-cut-bound}
\end{equation}
The last equality telescopes \eqref{eq:proof-frontier-increments}.
The nested tree attains every bound. Summation by parts, using
\(\sum_rZ_r=1+d(q)\), gives its labor consequence:
\begin{equation}
C
=\ell_1[1+d(q)]
+\sum_{h=2}^n\Delta_hT_h.
\label{eq:proof-tail-cost}
\end{equation}
All \(\Delta_h>0\), so the nested tree minimizes labor, and
every minimizer must bind every cut bound. Adjacent differences of
the tails, together with total work, give
\begin{equation}
Z_1=1,
\qquad
Z_r=g_r(q)
\quad(2\leq r\leq n).
\label{eq:proof-unique-workloads}
\end{equation}
Thus workloads are unique even when the tree is not. Moreover,
\begin{equation}
\sum_{r=2}^ng_r(q)
=\Phi_1(q)-\Phi_n(q)
=1-\lVert q\rVert_2^2
=d(q).
\label{eq:proof-interface-sum}
\end{equation}

\emph{Step 3. Establish the minimum handoff scale \(n-1\) among labor minimizers.}
Every first-scope-\(r\) cut in a labor minimizer attains its
frontier, whose profiles have exactly \(r\) coordinates. No path
can skip scope \(r\), since its first profile of scope at most
\(r\) would then have fewer coordinates.
Choose child \(c\) at node \(v\) with probability
\(\omega_c/\omega_v\). Each node and its incoming edge are
reached with probability \(\omega_v\), so \(J(\mathcal A)\)
is the expected number of external handoffs on this path. For an
internal node, coherence gives
\[
z_v=\sum_{c\in\operatorname{ch}(v)}
\omega_c\|p_c-p_v\|_2^2.
\]
The last node at each scope \(r>1\) on a path has a smaller-scope
child and therefore positive work. The scope-one terminal also has
positive work. Exact matching places these \(n\) nodes in distinct
positions; connectedness precludes leaving and reentering a position.
Every path therefore has at least \(n-1\) handoffs.
The nested tree, with one matched position per node, attains
\(J=n-1\). In any selected minimizer, every path must attain this
bound because it has positive probability. Its positions are therefore
exactly the \(n\) required by the positive-work nodes at distinct
scopes, with no additional position containing only zero-work nodes.

\emph{Step 4. Form a position tree preserving labor, employment, and handoffs.}
The squared-distance expression for \(z_v\) in Step 3 shows that
an internal node has zero work exactly
when all its child profiles equal its own. Suppress a nonroot
zero-work node \(v\) by attaching its children to its parent.
Coherence, work, and labor are unchanged. If this disconnects its
position, give the components separate names with the same knowledge.
If \(v\) shared its parent's position, handoff indicators are
unchanged; otherwise the old incoming handoff cost \(\omega_v\)
covers all replacement-edge costs, whose sum is at most
\(\sum_c\omega_c=\omega_v\). Thus \(J\) cannot rise and must
remain \(n-1\).
Repeat until all nonroot nodes have positive work. A remaining
zero-work root has positive-work children with profile \(q\).
Each must share its position, or its path would have a handoff before
the \(n\) required positions. Contract these children into the
root. Finally, contract each connected position: its positive-work
nodes have identical profiles by exact matching. The contracted node
inherits its top node's scale and profile and the set's boundary
children. Summing coherence and work identities over the set shows
that labor, employment, and external handoffs are preserved.

\emph{Step 5. Show adjacent-scope handoffs and work proportional to scale within each layer.}
The resulting tree has positive work everywhere and one position of
each scope on every path. Handoffs are therefore between adjacent
scopes and drop one domain. Each scope-\(r\) layer is a frontier
cut, so all its profiles have squared norm \(\Phi_r(q)\) and
all its children squared norm \(\Phi_{r-1}(q)\). Hence
\eqref{eq:proof-proportional-work} holds in every selected hierarchy,
including terminals with \(g_1=1\). This proportionality will link
responsibility scales to employment weights.

\emph{Step 6. Show that narrower scopes perform strictly more work.}
In \eqref{eq:proof-inclusion-problem}, replace \(\sum_i\pi_i=r\)
by \(\sum_i\pi_i=z\), where \(z\in[1,n]\), and denote the
minimum by \(V_q(z)\). For \(z\in[1,n)\), let \(\theta_q(z)\) be the
unique positive solution to
\(\sum_{i\in I}\min\{1,q_i/\theta_q(z)\}=z\).
Set \(\theta_q(n)=\min_iq_i\), the continuous limiting value.
The preceding solution of the relaxed problem applies unchanged at
noninteger \(z\). It shows that \(\theta_q\) is continuous,
positive, and strictly decreasing on \([1,n]\), and that
\(V_q(r)=\Phi_r(q)\) at every integer.
On an open interval in \((1,n)\) where \(\{i:q_i\geq\theta_q(z)\}\) is fixed,
let \(c\) be its size and \(S\) the remaining mass of \(q\).
Then
\[
\theta_q(z)=\frac{S}{z-c},
\qquad
V_q(z)=\sum_{q_i\geq\theta_q(z)}q_i^2+\frac{S^2}{z-c},
\qquad
V_q'(z)=-\theta_q(z)^2.
\]
This set changes only when the threshold crosses a coordinate
\(q_i\). There are finitely many such changes, and \(\theta_q\)
and \(V_q\) are continuous across them. Integration therefore gives
\begin{equation}
g_r(q)=\int_{r-1}^{r}\theta_q(z)^2\,dz
\qquad(2\leq r\leq n).
\label{eq:proof-envelope-workloads}
\end{equation}
Positivity and strict decrease of \(\theta_q\) imply
\(0<g_n<\cdots<g_2\) by comparing adjacent integration intervals.
Every simplex profile of scope at most two has squared norm at least
\(1/2\): this follows from Cauchy--Schwarz on its active coordinates.
Hence \(\Phi_2(q)\geq1/2\) and
\(g_2=1-\Phi_2(q)\leq1/2<g_1=1\). Exact matching and the
unique workloads finally give
\(C^M(q)=\sum_{r=1}^n\ell_rg_r(q)\) and
\(Y^M(q)=1/C^M(q)\).
This completes the proof.
\hfill\rule{0.5em}{0.5em}\par\egroup

\hypertarget{proof-of-corollary}{%
\subsection{\texorpdfstring{Proof of Corollary
\ref{cor:employment-hierarchy}}{Proof of Corollary }}\label{proof-of-corollary}}

\label{app:proof-employment}

\bgroup \par\noindent\textbf{Proof.}\ 
If \(q_i=1/K\) for every domain, the frontier lemma gives
\(\pi_i^r=r/K\). Conditional on being active, coordinate \(i\)
therefore equals \(q_i/\pi_i^r=1/r\). Every scope-\(r\) profile
in a selected hierarchy is thus uniform on its \(r\)-element support,
has squared norm \(\Phi_r(q)=1/r\), and generates workloads
\[
g_1=1,
\qquad
g_r=\Phi_{r-1}(q)-\Phi_r(q)=\frac{1}{r(r-1)}
\quad(2\leq r\leq K).
\]

For the employment ordering, allow any production profile \(q\) with
\(n\) active domains.
On any open interval in \((1,n)\) with \(\{i:q_i\geq\theta_q(z)\}\) fixed,
let \(c=|\{i:q_i\geq\theta_q(z)\}|\) and
\(S=\sum_{q_i<\theta_q(z)}q_i\). Then
\[
\theta_q(z)=\frac{S}{z-c},
\qquad
z\theta_q(z)=\frac{Sz}{z-c},
\qquad
\frac{d}{dz}\bigl[z\theta_q(z)\bigr]
=-\frac{Sc}{(z-c)^2}\leq0.
\]
Continuity at the finitely many changes of this set implies that
\(z\theta_q(z)\), and hence
\(H(z)=[z\theta_q(z)]^2\), is nonincreasing throughout
\([1,n]\). For \(2\leq r<n\),
\eqref{eq:proof-envelope-workloads} and
\(\theta_q(z)^2=H(z)/z^2\) give
\[
g_r\geq H(r)\int_{r-1}^{r}\frac{dz}{z^2}
=\frac{H(r)}{r(r-1)},
\qquad
g_{r+1}\leq H(r)\int_r^{r+1}\frac{dz}{z^2}
=\frac{H(r)}{r(r+1)}.
\]
The workloads and \(H(r)\) are positive. Dividing these bounds yields
\begin{equation}
\frac{g_r(q)}{g_{r+1}(q)}\geq\frac{r+1}{r-1}.
\label{eq:proof-workload-ratio}
\end{equation}
The previous proof also gives \(g_2\leq1/2\).
Since \(m_r=\ell_rg_r/C^M\), it remains to compare labor across scopes.
Under \(\ell_r=1+\kappa(r-1)\) with \(0<\kappa<1\),
\[
\begin{aligned}
\ell_2g_2&\leq\frac{1+\kappa}{2}<1=\ell_1g_1,\\
\frac{\ell_{r+1}}{\ell_r}
&=\frac{1+\kappa r}{1+\kappa(r-1)}
<\frac{r+1}{r-1}
\leq\frac{g_r}{g_{r+1}}
\quad(2\leq r<n).
\end{aligned}
\]
Thus \(\ell_rg_r>\ell_{r+1}g_{r+1}\) for every adjacent pair,
which proves the employment ordering. At \(n=2\), only the first
comparison is needed.
\hfill\rule{0.5em}{0.5em}\par\egroup

\hypertarget{proof-of-theorem-1}{%
\subsection{\texorpdfstring{Proof of Theorem
\ref{thm:division-understanding}}{Proof of Theorem }}\label{proof-of-theorem-1}}

\label{app:proof-division-understanding}

\bgroup \par\noindent\textbf{Proof.}\ 
Fix a productively optimal hierarchy and sample a path by choosing
child \(c\) at parent \(v\) with probability \(\omega_c/\omega_v\).
Let \(V_r\) be its scope-\(r\) node and \(P_r=p_{V_r}\).
Exact matching makes \(P_r\) also its occupational profile, and
\(\Pr(V_r=v)=\omega_v\).
Conditional on the parent node \(V_r=v\), coherence gives
\[
\mathbb E[P_{r-1}\mid V_r=v]
=\sum_{c\in\operatorname{ch}(v)}
\frac{\omega_c}{\omega_v}p_c=p_v.
\]
The parent profile \(P_r\) is a function of \(V_r\). Applying
the law of iterated expectations, even when several nodes share the
same profile, therefore yields
\begin{equation}
P_r=\mathbb E[P_{r-1}\mid P_r].
\label{eq:proof-profile-martingale}
\end{equation}
At scope \(r\), node \(v\) employs
\(\ell_r\omega_vg_r/C^M\) workers and the layer employs
\(\ell_rg_r/C^M\). Its conditional employment weight is therefore
\(\omega_v\), exactly its path probability. Thus
\(b_r=\mathbb E[b(P_r)]\).

For each issue \(u\), \(\gamma(s,u)=\beta(s,u)/\ell_{r(u)}\)
is concave in \(s\):
fit is a minimum of linear functions and \(\ell_{r(u)}>0\).
Conditional Jensen gives
\[
\gamma(P_r,u)
\geq\mathbb E[\gamma(P_{r-1},u)\mid P_r].
\]
Taking expectations yields
\begin{equation}
\mathbb E[\gamma(P_r,u)]
\geq
\mathbb E[\gamma(P_{r-1},u)]
\quad(2\leq r\leq n).
\label{eq:proof-layer-understanding}
\end{equation}
Integrating over the agenda gives \(b_r\geq b_{r-1}\).
The interchange is valid since \(0\leq\gamma(s,u)\leq1\).
For \(u=e_k\), where \(e_k\) is the \(k\)th unit vector,
understanding is \(\gamma(s,e_k)=s_k\).
Starting at \(P_n=q\) and taking expectations in
\eqref{eq:proof-profile-martingale} gives \(\mathbb E[P_r]=q\).
Every layer therefore has mean understanding \(q_k\) of this issue.
\hfill\rule{0.5em}{0.5em}\par\egroup

\hypertarget{proof-of-corollary-1}{%
\subsection{\texorpdfstring{Proof of Corollary
\ref{cor:coordinator-premium}}{Proof of Corollary }}\label{proof-of-corollary-1}}

\label{app:proof-coordinator-premium}

\bgroup \par\noindent\textbf{Proof.}\ 
Let \(p\) be the coordinator's profile,
\(p_j\) its children's profiles, and \(\alpha_j>0\) their
scales divided by the parent's scale. Coherence gives
\(p=\sum_j\alpha_jp_j\) and \(\sum_j\alpha_j=1\).
Since \(\gamma(s,e_k)=s_k\), each single-domain issue contributes
zero to the coordinator's premium over this weighted child average.
Only the agenda weight \(\lambda\) on \(u=p\) remains, giving
\[
\frac{\lambda}{\ell_{r(p)}}
\left[1-\sum_j\alpha_j\beta(p_j,p)\right].
\]
By \eqref{eq:proof-exact-match}, \(\beta(p_j,p)\leq1\), with
equality exactly when \(p_j=p\). Thus the premium is nonnegative,
and strictly positive exactly when \(\lambda>0\) and some
positive-share child has a different profile.
\hfill\rule{0.5em}{0.5em}\par\egroup
\hypertarget{political-equilibrium}{%
\section{Political Equilibrium}\label{political-equilibrium}}

\label{app:proofs-politics}

\hypertarget{proof-of-theorem-2}{%
\subsection{\texorpdfstring{Proof of Proposition
\ref{thm:political-equilibrium}}{Proof of Proposition }}\label{proof-of-theorem-2}}

\label{app:proof-political-equilibrium}

At zero effort, we use the following conventions. A voter with \(b_i>0\)
supports each candidate with probability one half at \((0,0)\) and
supports the sole positive-effort candidate with probability one when
exactly one effort is positive. A voter with \(b_i=0\) supports each
candidate with probability one half at every effort profile.

\bgroup \par\noindent\textbf{Proof.}\ 
Suppose \(B>0\), and let \(\mu_+>0\) be the measure of
\(I_+=\{i:b_i>0\}\). The cost assumptions imply that \(c\)
is strictly increasing on positive efforts.
At \((0,0)\), a deviation to any positive
effort raises a candidate's aggregate support by \(\mu_+/2\), whereas
its cost tends to zero with effort. Thus \((0,0)\) is not an
equilibrium. Nor can only one candidate exert positive effort in
equilibrium: she could reduce it while keeping it positive, preserving
all her support and strictly lowering her cost.
Zero is not a best response to a positive effort by the opponent either. Fix \(\bar e>0\), set \(M=\max\{1,\bar e\}\), and take \(0<e\leq1\). For every \(i\in I_+\),
\[
\pi(e,\bar e;b_i)
=\frac{e^{b_i}}{e^{b_i}+\bar e^{b_i}}
\geq\frac{e}{1+M}.
\]
The numerator is at least \(e\) and the denominator at most \(1+M\).
Voters outside \(I_+\) remain indifferent, so aggregate support rises
by at least \(\mu_+e/(1+M)\) relative to zero effort. Since
\(c'(0)=0\) implies \(c(e)=o(e)\), this gain exceeds the cost
for sufficiently small \(e>0\). Every equilibrium must therefore
be interior.

For positive efforts, write \(\pi(e,\bar e;b)=e^b/(e^b+\bar e^b)\). For \(b\in(0,1]\),
\begin{equation}
\frac{\partial\pi}{\partial e}
=\frac{b}{e}\pi(1-\pi),
\qquad
\frac{\partial^2\pi}{\partial e^2}
=\frac{b e^{b-2}\bar e^b[(b-1)\bar e^b-(b+1)e^b]}
{(e^b+\bar e^b)^3}
\leq0.
\label{eq:proof-vote-concavity}
\end{equation}
For \(b=0\), support is constant. Against a positive opponent effort,
support converges pointwise to its zero-effort value as \(e\downarrow0\).
Dominated convergence therefore extends aggregate support continuously
to zero, and its concavity extends to \([0,\infty)\). Subtracting
the strictly convex cost makes the candidate's objective strictly concave.
The cost also tends to infinity, since for any \(e_0>0\),
\(c(e)\geq c(e_0)+c'(e_0)(e-e_0)\) for \(e\geq e_0\).
Bounded support then ensures that a global best response exists; it is
positive because zero is not optimal. On any compact interval \([a,d]\subset(0,\infty)\),
\(0\leq\partial\pi/\partial e\leq1/(4a)\), uniformly in \(b\),
which justifies differentiation under the voter integral.
At positive efforts, \(\pi_{i,-j}=1-\pi_{ij}\), so
\(\pi_{ij}(1-\pi_{ij})\) is the same for both candidates.
Multiplying their first-order conditions by their respective efforts gives
\begin{equation}
e_jc'(e_j)
=\int_0^1b_i\pi_{ij}(1-\pi_{ij})\,di
=e_{-j}c'(e_{-j}).
\label{eq:proof-paired-foc}
\end{equation}
The function \(h(e)=ec'(e)\) is strictly increasing on \((0,\infty)\), since \(h'(e)=c'(e)+ec''(e)>0\). Thus every equilibrium is symmetric.
At equal positive efforts, \(\pi_{ij}=1/2\) for every citizen, and the first-order condition becomes
\[
e^*c'(e^*)=\frac14\int_0^1b_i\,di=\frac B4.
\]
The left-hand side is continuous and tends to zero as \(e\downarrow0\).
It also tends to infinity: for any fixed \(e_0>0\), strict convexity
gives \(c'(e)\geq c'(e_0)>0\) whenever \(e\geq e_0\).
Thus the equation has exactly one positive solution. At that solution,
each candidate satisfies her first-order condition against the other;
strict concavity makes it her unique global best response. This proves
both existence and uniqueness among pure strategies. Implicit
differentiation gives \(\frac{de^*}{dB}=\frac{1}{4h'(e^*)}>0\).
Thus average understanding alone determines equilibrium effort. If
\(B=0\), then \(b_i=0\) almost everywhere: support is constant, and
each candidate uniquely minimizes her cost at zero effort.
\hfill\rule{0.5em}{0.5em}\par\egroup
\hypertarget{competition-and-planning}{%
\section{Competition and Planning}\label{competition-and-planning}}

\label{app:proofs-planning}

\hypertarget{proof-of-theorem-3}{%
\subsection{\texorpdfstring{Proof of Proposition
\ref{thm:competitive-implementation}}{Proof of Proposition }}\label{proof-of-theorem-3}}

\label{app:proof-competitive}

\bgroup \par\noindent\textbf{Proof.}\ 
For the measurable equilibrium wage schedule \(w\), let \(\Lambda\)
be the unit-mass distribution of workers' profiles and \(\Xi\) the
output measure on organizations, so
\(Y=\Xi(\Omega(q))>0\). The labor demand per unit of output from
organization \(\mathcal A\) is the measure
\(L_{\mathcal A}=\sum_{v:z_v>0}z_va(p_v,s_v)\delta_{s_v}\), where
\(\delta_s\) is unit mass at \(s\). Demands add at the same profile
when nodes share a position. Its labor requirement is
\(C_{\mathcal A}=L_{\mathcal A}(\Delta_K)\), and its payroll is
\(c_w(\mathcal A)=\int_{\Delta_K}w(s)\,dL_{\mathcal A}(s)\).
Free entry requires \(c_w(\mathcal A)\geq1\) for every feasible
organization, with equality \(\Xi\)-almost everywhere. Workers'
maximizing choices imply that \(\bar w=\max_s w(s)\) is attained
and that \(\Lambda(E)=0\) for \(E=\{s:w(s)<\bar w\}\).
Market clearing gives
\(0=\Lambda(E)=\int_{\Omega(q)}L_{\mathcal A}(E)\,d\Xi(\mathcal A)\).
Nonnegativity implies \(L_{\mathcal A}(E)=0\) for
\(\Xi\)-almost every organization. All its positive labor demand
therefore earns \(\bar w\), and zero profit gives
\(1=c_w(\mathcal A)=\bar wC_{\mathcal A}\); in particular,
\(\bar w>0\).
By Theorem~\ref{thm:efficient-organization}, a labor minimizer
\(\widetilde{\mathcal A}\) exists. For almost every active organization,
free entry and \(w(s)\leq\bar w\) imply
\[
1
\leq c_w(\widetilde{\mathcal A})
\leq\bar wC^M
\leq\bar wC_{\mathcal A}
=1.
\]
Every inequality binds. Thus \(\Xi\)-almost every organization uses
\(C^M\) units of labor per unit of output, and occupied profiles earn
\(\bar w=1/C^M\). Aggregate labor clearing gives
\(1=\int C_{\mathcal A}\,d\Xi(\mathcal A)=C^M Y\), hence
\(Y=1/C^M\).

Conversely, fix a labor minimizer \(\widetilde{\mathcal A}\), set
\(w(s)=1/C^M\) for all profiles, and produce \(1/C^M\) units
using this organization. Choose the workers' profile distribution as
\(\Lambda=L_{\widetilde{\mathcal A}}/C^M\), a probability measure
because \(L_{\widetilde{\mathcal A}}(\Delta_K)=C^M\). Workers
maximize wages, labor markets clear, and the active organization earns
zero profit. Any alternative has payroll
\(c_w(\mathcal A)=C_{\mathcal A}/C^M\geq1\), so entry cannot
be profitable.
\hfill\rule{0.5em}{0.5em}\par\egroup

\hypertarget{proof-of-theorem-4}{%
\subsection{\texorpdfstring{Proof of Theorem
\ref{thm:optimal-responsibility}}{Proof of Theorem }}\label{proof-of-theorem-4}}

\label{app:proof-optimal-responsibility}

The proof bounds work at narrow scopes and the understanding obtainable
from staffing it, then constructs matched cutoff organizations attaining
the bounds. Throughout, \(q\) is uniform. Write
\(N(\mathcal A)=C(\mathcal A)B(\mathcal A)
=\sum_vz_va(p_v,s_v)b(s_v)\) for the understanding numerator.
For the labor-minimizing workload, set \(g_1=1\),
\(g_r=\frac1{r(r-1)}\) for \(r=2,\ldots,K\), and
\(G_r=\sum_{t=1}^rg_t=2-\frac1r\).

\begin{lemma}[Workload-prefix bounds]
\label{lem:proof-workload-prefix}
Let \(X_t=\sum_{v:r(p_v)=t}z_v\) be the workload performed under
scope-\(t\) responsibilities. Every finite coherent organization satisfies
\begin{equation}
\sum_{t=1}^rX_t\leq G_r
\quad(r<K),
\qquad
\sum_{t=1}^KX_t=G_K.
\label{eq:proof-workload-prefix}
\end{equation}
Consequently, every feasible workload vector can be obtained by moving
work from each scope to weakly broader scopes: there are amounts
\(x_{tr}\geq0\), for \(1\leq t\leq r\leq K\), with
\[
\sum_{r=t}^Kx_{tr}=g_t,
\qquad
\sum_{t=1}^rx_{tr}=X_r.
\]
Here \(x_{tr}\) moves benchmark work from source scope \(t\) to
destination scope \(r\geq t\). The cutoff transports that retain
every source through some \(j\) at its original scope and send all
broader sources to scope \(K\) are feasible, as is every convex
combination of adjacent cutoffs.
\end{lemma}

\bgroup \par\noindent\textbf{Proof.}\ 
\textit{Step 1. Work conservation bounds every scope prefix.}
Coherence implies that a child's support is contained in its parent's,
so scope can only fall along a branch. Fix \(r<K\), and mark the
first node at scope \(r\) or below on each branch that reaches it.
The marked subtrees are disjoint and their scales sum to at most one:
a random path taking child \(c\) with probability \(\omega_c/\omega_v\)
visits node \(v\) with probability \(\omega_v\), and visits at most
one marked node.
A marked responsibility with scale \(\omega\) and profile \(p\)
contains \(\omega[1+d(p)]\) units of work, including its descendants.
Since \(r(p)\leq r\), Cauchy--Schwarz gives
\(1+d(p)=2-\lVert p\rVert^2\leq2-1/r=G_r\).
All work at scope \(r\) or below lies in these subtrees, so summing
proves the prefix bound. Work conservation at the uniform root gives
\(\sum_tX_t=1+d(q)=G_K\).

\medskip
\noindent\textit{Step 2. The prefix bounds characterize upward transport.}
Fill destinations in increasing order. At destination \(r\), the
unassigned mass from sources at most \(r\) is
\(G_r-\sum_{k<r}X_k\geq X_r\).
Assign \(X_r\) of this mass to \(r\), splitting sources if needed.
At \(r=K\), equality of total mass exhausts all sources, yielding the
stated row and column sums. Conversely, such a transport preserves total
mass and satisfies
\(\sum_{r\leq h}X_r=\sum_{r\leq h}\sum_{t\leq r}x_{tr}
\leq\sum_{t\leq h}g_t=G_h\).
This characterizes the transport relaxation; the cutoffs used below have
finite organizational implementations.

\medskip
\noindent\textit{Step 3. Implement cutoffs and adjacent mixtures with finite coherent trees.}
For \(j\in\{1,\ldots,K-1\}\), give the root one child for each
\(j\)-domain subset, with scale \(1/\binom Kj\) and profile uniform
on that subset. The mean child profile has every coordinate equal to
\(\binom{K-1}{j-1}/[j\binom Kj]=1/K\), so the division is coherent.
For \(j=1\), the children are terminal; otherwise use below each child
the nested hierarchy of Lemma~\ref{lem:proof-scope-frontiers} for a
uniform \(j\)-domain profile. Assign each node its own matched position,
which satisfies connectedness. Child scales sum to one, so their
hierarchies generate aggregate workload \(g_t\) at each \(t\leq j\).
The root performs \(\frac1j-\frac1K=\sum_{t=j+1}^Kg_t\)
units of work. At \(j=0\), use a single matched terminal root,
performing all \(G_K\) units.
To mix cutoffs \(j\) and \(j-1\) with weights \(1-x\) and \(x\),
where \(0<x<1\), add a root with profile \(q\) and two children
with profile \(q\), scales \(1-x\) and \(x\), and correspondingly
scaled copies of the cutoff trees. The new root is coherent, has zero
work, and occupies its own matched position; keep positions in the two
branches distinct. Workloads, labor requirements, and understanding
numerators are then the corresponding convex combinations. Mean
understanding is their numerator-to-labor ratio. At \(x=0,1\), use
the endpoint tree directly.
\hfill\rule{0.5em}{0.5em}\par\egroup

Let \(\nu_h\) be the agenda probability of an issue requiring \(h\)
domains. By symmetry, each \(h\)-domain set has conditional probability
\(1/\binom Kh\). A knowledge profile uniform on \(r\) domains has
understanding
\begin{equation}
b_r
=\sum_{h=1}^r
\frac{\nu_h}{\ell_h}
\frac{\binom{r-1}{h-1}}{\binom Kh},
\qquad
n_r:=\ell_rb_r.
\label{eq:proof-br}
\end{equation}
Indeed, an \(h\)-domain issue is covered with probability
\(\binom rh/\binom Kh\), with fit \(h/r\), and
\(\binom rh\,h/r=\binom{r-1}{h-1}\). We set \(\binom ab=0\)
for \(a\geq0\) and \(b<0\) or \(b>a\). The formula implies that
\(b_r\) is nondecreasing and
\(b_K=\sum_h\nu_hh/(K\ell_h)>0\). One unit of matched scope-\(r\)
work contributes \(n_r\) to the understanding numerator.

\begin{lemma}[Assignment envelope and ordered returns]
\label{lem:proof-assignment-envelope}
Suppose \(\ell_r=1+\kappa(r-1)\), with \(0<\kappa<1\).

\textit{(i)} For any \(s\in\Delta_K\), order its domains so that \(s_{(1)}\geq\cdots\geq s_{(K)}\), breaking ties consistently, and set \(s_{(K+1)}=0\). Define \(\alpha_r=r[s_{(r)}-s_{(r+1)}]\). If \(u_r\) is uniform on the first \(r\) domains in this ordering, then
\begin{equation}
s=\sum_{r=1}^K\alpha_ru_r,
\qquad
b(s)=\sum_{r=1}^K\alpha_rb_r,
\qquad
M_t(s):=\sum_{j=1}^ts_{(j)}
=\sum_{r=1}^K\alpha_r\min\left\{1,\frac tr\right\}.
\label{eq:proof-profile-decomposition}
\end{equation}

\textit{(ii)} If a responsibility \(p\) spans \(t\) knowledge domains,
then, for every profile \(s\) with \(\beta(s,p)>0\) and every scalar
\(\bar b\geq b_K\),
\begin{equation}
a(p,s)[b(s)-\bar b]
\leq
\max_{t\leq r\leq K}\ell_r[b_r-\bar b].
\label{eq:proof-assignment-envelope}
\end{equation}

\textit{(iii)} The sequence \(n_r\) is discretely convex. The secant slopes
to scope \(K\) therefore satisfy
\begin{equation}
\sigma_r:=\frac{n_K-n_r}{\ell_K-\ell_r}
\quad(r=1,\ldots,K-1)
\qquad\text{with}\qquad
\sigma_1\leq\cdots\leq\sigma_{K-1},
\quad \sigma_r\geq b_K.
\label{eq:proof-sigma-order}
\end{equation}
\end{lemma}

\bgroup \par\noindent\textbf{Proof.}\ 
\textit{Step 1. Decompose ordered profiles and understanding linearly.}
For part (i), \(\alpha_r\geq0\) and
\(\sum_r\alpha_r=\sum_rs_{(r)}=1\). The \(j\)th coordinate of
the proposed decomposition is
\(\sum_{r=j}^K\frac{\alpha_r}{r}
=\sum_{r=j}^K[s_{(r)}-s_{(r+1)}]=s_{(j)}\).
For an issue uniform on an \(h\)-domain set \(U\), fit is
\(h\min_{k\in U}s_k\). Under a fixed coordinate ordering, the
minimum selects the coordinate with the largest rank index in \(U\), so
understanding is linear on profiles with that ordering. The profiles
\(s,u_1,\ldots,u_K\) share the ordering, proving the decomposition
of \(b(s)\). Summing the first \(t\) coordinates gives \(M_t(s)\).
In particular, \(b(s)\leq b_K\).

\medskip
\noindent\textit{Step 2. Bound staffing returns by the best matched return.}
For part (ii), define
\(Q_t(\bar b)=\max_{t\leq r\leq K}\ell_r[b_r-\bar b]\leq0\),
where the sign follows from \(b_r\leq b_K\leq\bar b\).
If \(T=\operatorname{supp}p\), summing
\(\beta(s,p)p_k\leq s_k\) over \(k\in T\) gives
\begin{equation}
\beta(s,p)\leq\sum_{k\in T}s_k\leq M_t(s).
\label{eq:proof-beta-mass}
\end{equation}
For \(r<t\), monotonicity gives
\(b_r-\bar b\leq b_t-\bar b\leq Q_t(\bar b)/\ell_t\).
For \(r\geq t\), \(\ell_r/r=\kappa+(1-\kappa)/r\) decreases,
so \(1/\ell_r\geq t/(r\ell_t)\). Multiplying by the nonpositive
\(Q_t(\bar b)\) reverses this inequality:
\[
b_r-\bar b
\leq\frac{Q_t(\bar b)}{\ell_r}
\leq\frac tr\frac{Q_t(\bar b)}{\ell_t}.
\]
Combining both cases, multiplying by \(\alpha_r\), and summing gives,
by \eqref{eq:proof-profile-decomposition},
\[
b(s)-\bar b\leq\frac{M_t(s)}{\ell_t}Q_t(\bar b).
\]
Multiplying by \(a(p,s)=\ell_t/\beta(s,p)\) and using
\eqref{eq:proof-beta-mass} gives
\[
a(p,s)[b(s)-\bar b]
\leq\frac{M_t(s)}{\beta(s,p)}Q_t(\bar b)
\leq Q_t(\bar b),
\]
where the last inequality again uses \(Q_t(\bar b)\leq0\).
This proves \eqref{eq:proof-assignment-envelope}.

\medskip
\noindent\textit{Step 3. Returns to moving work to the root rise with source scope.}
For part (iii), the contribution of \(h\)-domain issues, before weighting
by \(\nu_h\), is
\(n_r(h)=\ell_r\binom{r-1}{h-1}/[\ell_h\binom Kh]\).
The identity
\[
[1+\kappa(r-1)]\binom{r-1}{h-1}
=[1+\kappa(h-1)]\binom{r-1}{h-1}
+\kappa h\binom{r-1}{h}
\]
writes its numerator as a nonnegative combination of binomial sequences.
Pascal's identity gives their second differences:
\[
\binom{r+1}{m}-2\binom r m+\binom{r-1}{m}
=\binom{r-1}{m-2}\geq0.
\]
The zero conventions include the points where a sequence first becomes
positive and the cases \(m=0,1\). Division by a positive constant
and mixing with weights \(\nu_h\) preserve convexity. Thus
\[
\sigma_r=\frac1{K-r}\sum_{t=r}^{K-1}
                   \frac{n_{t+1}-n_t}{\kappa}
\]
is an average of nondecreasing increments; increasing \(r\) removes
a weakly smallest term, so \(\sigma_r\) is nondecreasing. Finally,
\(\sigma_r=b_K+\frac{\ell_r(b_K-b_r)}{\ell_K-\ell_r}\geq b_K\).
\hfill\rule{0.5em}{0.5em}\par\egroup

\bgroup \par\noindent\textbf{Proof of Theorem \ref{thm:optimal-responsibility}.}\ 
The lemmas yield a bound for arbitrary organizations. We solve the bound
by cutoffs, show that they trace the whole relevant frontier, and establish
a global optimum on that frontier.

\medskip
\noindent\textit{Step 1. Bound arbitrary organizations by a transport problem.}
Fix an organization with outcomes \((C,N)\), a coefficient
\(\bar b\geq b_K\), and write
\(w_r(\bar b)=n_r-\bar b\ell_r\). Lemma
\ref{lem:proof-assignment-envelope} bounds the contribution of one unit
of scope-\(t\) work to \(N-\bar bC\) by
\(\max_{t\leq r\leq K}w_r(\bar b)\). Choose a maximizer \(R(t)\)
and assign actual scope-\(t\) work to that scope for the purpose of
the bound. The resulting vector \(Z_r=\sum_{t:R(t)=r}X_t\) preserves
total work and, because \(R(t)\geq t\), satisfies
\(\sum_{r\leq h}Z_r\leq\sum_{t\leq h}X_t\leq G_h\). Hence
\begin{equation}
N-\bar bC\leq\sum_{r=1}^KZ_rw_r(\bar b)
\leq\sum_{t=1}^Kg_t\max_{t\leq r\leq K}w_r(\bar b).
\label{eq:proof-support-bound}
\end{equation}
For the second inequality, represent \(Z\) by upward transport amounts:
\(\sum_rZ_rw_r=\sum_t\sum_{r\geq t}x_{tr}w_r
\leq\sum_tg_t\max_{r\geq t}w_r\).

\medskip
\noindent\textit{Step 2. Cutoff transports maximize the relaxed bound.}
Since \(n_r\) is convex and \(\ell_r\) affine, \(w_r(\bar b)\)
is discretely convex. For \(t<r<K\),
\[
w_r(\bar b)\leq
\frac{K-r}{K-t}w_t(\bar b)+\frac{r-t}{K-t}w_K(\bar b)
\leq\max\{w_t(\bar b),w_K(\bar b)\}.
\]
Thus each source \(t\) is optimally retained at \(t\) or sent to
scope \(K\). For \(t<K\), the gain from sending it to the root is
\begin{equation}
w_K(\bar b)-w_t(\bar b)
=(\ell_K-\ell_t)(\sigma_t-\bar b).
\label{eq:proof-root-comparison}
\end{equation}
Retaining source \(t\) is optimal when \(\sigma_t\leq\bar b\),
and sending it to the root is optimal when \(\sigma_t\geq\bar b\),
with either choice allowed at equality. Since \(\sigma_t\) is
nondecreasing, ties can be resolved so that sources sent to the root
form an upper interval. Source \(K\) always remains there.

\medskip
\noindent\textit{Step 3. Feasible cutoffs form a concave chain of outcomes.}
Lemma~\ref{lem:proof-workload-prefix} implements each maximizing cutoff
with matched staffing, attaining the bound. To describe the full chain,
at cutoff \(j\in\{0,\ldots,K-1\}\) write
\begin{equation}
Z_r^j=
\begin{cases}
g_r,&r\leq j,\\
0,&j<r<K,\\
\displaystyle\sum_{t=j+1}^Kg_t,&r=K.
\end{cases}
\label{eq:proof-cutoff-vector}
\end{equation}
The associated labor requirement and understanding numerator are
\begin{equation}
C_j=\sum_{r=1}^jg_r\ell_r+\ell_K\sum_{r=j+1}^Kg_r,
\qquad
N_j=\sum_{r=1}^jg_rn_r+n_K\sum_{r=j+1}^Kg_r.
\label{eq:proof-cutoff-outcomes}
\end{equation}
For \(j\in\{1,\ldots,K-1\}\), if a fraction \(x\in[0,1]\) of the scope-\(j\) source workload is additionally assigned to the root,
\begin{equation}
C_{j,x}=C_j+xg_j(\ell_K-\ell_j),
\qquad
N_{j,x}=N_j+xg_j(n_K-n_j).
\label{eq:proof-partial-cutoff}
\end{equation}
This joins cutoff \(j\) to cutoff \(j-1\), with
\(C_{j-1}-C_j=g_j(\ell_K-\ell_j)>0\) and slope \(\sigma_j\)
in \((C,N)\)-space. From cutoff \(K-1\) to cutoff zero, labor
increases and the slopes \(\sigma_{K-1},\ldots,\sigma_1\) weakly
decrease. The chain is therefore concave.

\medskip
\noindent\textit{Step 4. Show that the chain maximizes understanding at each labor requirement.}
The chain spans \([C_{K-1},C_0]=[C^M,\ell_KG_K]\).
For a point on segment \(j\), choose \(\bar b=\sigma_j\geq b_K\).
Then source \(j\) is indifferent between its original scope and the
root. For \(t<j\), \(\sigma_t\leq\bar b\), so retaining it is
optimal; for \(j<t<K\), \(\sigma_t\geq\bar b\), so sending it
to the root is optimal. Both endpoint cutoffs, and every mixture between
them, therefore attain \eqref{eq:proof-support-bound}. This includes
tied slopes and the outer endpoints, each of which belongs to a segment.
If a point on the chain has outcomes \((C,\widehat N)\) and an arbitrary
feasible organization has the same \(C\) and numerator \(N\), the
support bound gives \(N-\bar bC\leq\widehat N-\bar bC\),
hence \(N\leq\widehat N\).
The bound ranges over all feasible trees and knowledge assignments,
including cross-training. Thus the chain maximizes \(N\), and hence
\(B=N/C\), at every \(C\in[C^M,C_0]\). At its left endpoint,
it also proves that \(B^M=N_{K-1}/C^M>0\) is attained.

\medskip
\noindent\textit{Step 5. The chain contains a global welfare optimum.}
No organization has \(C<C^M\). If \(C>C_0\), the profile
decomposition gives \(b(s_v)\leq b_K\), hence \(B\leq b_K\).
The all-root organization has \(C=C_0\) and \(B=b_K\), so it
has strictly greater output and weakly greater understanding. Since
\(W_0(Y,B)=Y+\log Y+\epsilon\log B\) increases in both arguments,
and \(B=0\) gives welfare \(-\infty\), every feasible organization
is weakly dominated in welfare by a point on the chain.
The chain is a finite union of closed segments. Each cutoff has positive
root workload staffed at \(b_K>0\); mixtures therefore retain positive
numerators and understanding. Welfare is continuous on this compact
chain, so a global maximum is attained, possibly at an endpoint.
Its finite implementation has matched staffing and at most two adjacent
cutoff branches, splitting at most one source workload between its original
scope and system-wide positions.
\hfill\rule{0.5em}{0.5em}\par\egroup

\hypertarget{proof-of-corollary-2}{%
\subsection{\texorpdfstring{Proof of Corollary
\ref{cor:strict-over-specialization}}{Proof of Corollary }}\label{proof-of-corollary-2}}

\label{app:proof-strict-over-specialization}

\bgroup \par\noindent\textbf{Proof.}\ 
The preceding proof attains \((C,N)=(C^M,C^MB^M)\), with \(B^M>0\).
The partial-cutoff construction moves a fraction \(x\in[0,1]\) of
scope-\(K-1\) workload to system-wide positions. Writing
\(\Delta C=g_{K-1}(\ell_K-\ell_{K-1})>0\), its outcomes are
\[
C_x=C^M+x\Delta C,
\qquad
N_x=C^MB^M+x\Delta C\,\sigma_{K-1},
\qquad
B_x=N_x/C_x.
\]
In these coordinates, welfare is
\(\widetilde W_0(C,N)=1/C+\epsilon\log N-(1+\epsilon)\log C\).
Differentiating along the segment gives
\[
\left.\frac{d\widetilde W_0(C_x,N_x)}{dx}\right|_{x=0}
=
\frac{\Delta C}{C^M}
\left[
\epsilon\frac{\sigma_{K-1}-B^M}{B^M}
-\left(1+\frac1{C^M}\right)
\right]>0
\]
under the stated condition. A sufficiently small positive reassignment
therefore improves on the best labor-minimizing outcome, and hence on
every labor minimizer, whose understanding is at most \(B^M\).
Theorem \ref{thm:optimal-responsibility} supplies an optimal cutoff whose
welfare is at least that of the improving reassignment. It cannot be the
labor-minimizing cutoff, so it assigns a positive amount of scope-\(K-1\)
source work to system-wide positions. Every labor minimizer has the
canonical workload vector \(g\); the optimal cutoff therefore assigns
strictly more work to system-wide positions than any labor minimizer.

If every issue requires all \(K\) domains, \(n_r=0\) for \(r<K\)
and \(n_K=1\). Hence \(\sigma_{K-1}=\frac1\kappa\) and
\(B^M=\frac{g_K}{C^M}=\frac1{K(K-1)C^M}\).
As \(\kappa\downarrow0\), \(C^M=\sum_r\ell_rg_r\to G_K>0\)
and \(B^M\to g_K/G_K>0\). For fixed \(\epsilon>0\), the left
side of the corollary's condition, \(\epsilon(1/\kappa-B^M)\),
diverges, while \(B^M(1+1/C^M)\) stays finite. The inequality holds
for all sufficiently small positive \(\kappa\).
\hfill\rule{0.5em}{0.5em}\par\egroup
\hypertarget{targeted-public-spending-and-institutions}{%
\section{Targeted Public Spending and
Institutions}\label{targeted-public-spending-and-institutions}}

\label{app:proofs-targetability}

\hypertarget{proof-of-proposition}{%
\subsection{\texorpdfstring{Proof of Proposition
\ref{prop:institutions}}{Proof of Proposition }}\label{proof-of-proposition}}

\label{app:proof-institutions}

\bgroup \par\noindent\textbf{Proof.}\ 
For positive \(f\), we establish an interior authority optimum and
characterize it by a unique scalar equation, then derive authority
rankings and sharing incentives.

\emph{Step 1. Establish an interior authority optimum for \(f>0\).}
With the proposition's occupational distribution and output fixed,
substituting equilibrium public provision into welfare and dropping
terms independent of authority gives
\[
J_f(v)=(\epsilon-f)\log B_v+f\sum_gm_g\log v_g,
\qquad B_v=\sum_gm_gv_gb_g.
\]
The term multiplied by \(f\) is omitted at \(f=0\). The planner therefore solves
\begin{equation}
\begin{aligned}
\max_{(v_g)_g}\quad & J_f(v)\\
\text{subject to}\quad
&v_g\geq0\quad\text{for every }g,\\
&\sum_gm_gv_g=1.
\end{aligned}
\label{eq:proof-authority-problem}
\end{equation}
At \(f=0\), maximizing \(B_v\) puts all authority on groups with
maximal understanding, since \(B_v\) is a convex combination with
weights \(m_gv_g\).

For \(f>0\), the feasible simplex is compact and
\(\min_g b_g\leq B_v\leq\max_g b_g\), so \(\log B_v\)
remains bounded even when \(f>\epsilon\). If any \(v_g\) approaches
zero, \(fm_g\log v_g\to-\infty\), while the other terms are bounded
above. Extending \(J_f\) to minus infinity on the boundary makes it
upper semicontinuous. The feasible vector \(v_g=1\) has finite value;
hence a global optimum exists and is interior.

\emph{Step 2. Characterize the unique global optimum.}
If \(\varpi\) is the multiplier on the authority constraint,
every optimum must satisfy
\begin{equation}
(\epsilon-f)\frac{b_g}{B_v}+\frac{f}{v_g}=\varpi
\qquad\text{for every }g.
\label{eq:proof-authority-foc}
\end{equation}
Multiplying by \(m_gv_g\), summing over groups, and using both normalizations gives \(\varpi=\epsilon\). Thus
\begin{equation}
v_g^*
=\frac{f}{\epsilon-(\epsilon-f)b_g/B_v^*}.
\label{eq:proof-authority-implicit}
\end{equation}
Define \(\vartheta=(\epsilon-f)/(\epsilon B_v^*)\). Equation \eqref{eq:proof-authority-implicit} becomes
\begin{equation}
v_g^*=\frac{f/\epsilon}{1-\vartheta b_g},
\qquad
\sum_g\frac{m_g}{1-\vartheta b_g}=\frac{\epsilon}{f},
\qquad
\vartheta<\frac{1}{\max_g b_g}.
\label{eq:proof-authority-scalar}
\end{equation}
On \(\vartheta\in(-\infty,1/\max_gb_g)\), every denominator is
positive, and the left-hand side rises continuously and strictly from
zero to infinity. Thus the scalar equation has a unique solution, which
generates a positive, normalized vector \(v^*\). Its implied
\(B_v^*\) satisfies
\[
\vartheta B_v^*
=\frac f\epsilon\sum_gm_g\frac{\vartheta b_g}{1-\vartheta b_g}
=\frac{\epsilon-f}{\epsilon},
\]
so the definition of \(\vartheta\) also holds, including at
\(f=\epsilon\). Every optimum satisfies these conditions, and they
identify a unique feasible vector. Existence therefore makes this vector
the unique global optimum, including when \(f>\epsilon\); concavity
is not needed for this argument.

\emph{Step 3. Order formal authority and effective influence.}
Since the left-hand side of the scalar equation in \eqref{eq:proof-authority-scalar} equals one at \(\vartheta=0\), its solution is positive when \(f<\epsilon\), zero when \(f=\epsilon\), and negative when \(f>\epsilon\). Equation \eqref{eq:proof-authority-scalar} then gives greater authority to better-informed groups when \(f<\epsilon\), equal authority \(v_g^*=1\) when \(f=\epsilon\), and less authority to better-informed groups when \(f>\epsilon\).
Within this fixed allocation, \(\vartheta\) is common to all groups,
so effective influence remains increasing in understanding:
\begin{equation}
v_g^*b_g=\frac f\epsilon\frac{b_g}{1-\vartheta b_g},
\qquad
\frac{d}{db}\frac{b}{1-\vartheta b}=\frac{1}{(1-\vartheta b)^2}>0.
\label{eq:proof-effective-influence}
\end{equation}

\emph{Step 4. Compare sharing and withholding at every cost.}
Under the proposition's information-sharing comparison, the occupational
distribution, authority, output, \(f\), and the source group's own
understanding and private income are fixed. A member \(i\) of the source group \(h\) receives
\(u_i=y_i+\log Y+(\epsilon-f)\log B_v+f\log(v_hb_h)\).
The term multiplied by \(f\) is omitted at \(f=0\); for \(f>0\),
the assumption \(v_h>0\) makes this utility finite. Raising the other groups' understanding changes \(B_v\) but leaves every other term in the source member's payoff fixed. The gross gain is therefore
\((\epsilon-f)\log(B_v'/B_v)\).
Sharing is strictly preferred when this gain exceeds \(k\), withholding
when it is below \(k\), and the source is indifferent at equality.
Because \(B_v'>B_v\), the source shares for sufficiently small costs
when \(f<\epsilon\). At \(f=\epsilon\), it withholds if \(k>0\)
and is indifferent if \(k=0\). At \(f>\epsilon\), it strictly
withholds for every \(k\geq0\).
\hfill\rule{0.5em}{0.5em}\par\egroup

\hypertarget{proof-of-theorem-5}{%
\subsection{\texorpdfstring{Proof of Theorem
\ref{thm:targetable-production}}{Proof of Theorem }}\label{proof-of-theorem-5}}

\label{app:proof-targetable-production}

\bgroup \par\noindent\textbf{Proof.}\ 
Maintain the theorem's equal voting weights, uniform \(q\), symmetric
dimensional agenda with \(\nu_1>0\), affine labor intensities, and
\(f\in[0,\epsilon]\). Steps 1--2 obtain a welfare optimum in the
closed outcome set. Steps 3--6 bound linear objectives and construct
finite hierarchies attaining those bounds. Steps 7--8 recover the
optimum as a mixture of at most three hierarchies, or two at the endpoints.

\emph{Step 1. Express welfare in three outcomes and implement labor-share mixtures.}
For \(\mathcal A\in\Omega(q)\), let \(x_v=z_va(p_v,s_v)\)
be the labor required to perform the assignment at node \(v\) per unit of output, and define
\begin{equation}
C_{\mathcal A}=\sum_vx_v,
\qquad
N_{\mathcal A}=\sum_vx_vb(s_v),
\qquad
\Gamma_{\mathcal A}=\sum_vx_v\log b(s_v).
\label{eq:proof-target-mode-statistics}
\end{equation}
The corresponding economy-wide outcomes are
\(Y_{\mathcal A}=C_{\mathcal A}^{-1}\),
\(B_{\mathcal A}=N_{\mathcal A}/C_{\mathcal A}\), and
\(S_{\mathcal A}=\Gamma_{\mathcal A}/C_{\mathcal A}\).
Node \(v\) employs mass \(x_v/C_{\mathcal A}\); summing across
nodes of the same position gives its occupational share. Thus
\(S_{\mathcal A}\) is mean log understanding with the same employment
weights as \(B_{\mathcal A}\).
Finite labor-share combinations are feasible organizations. Given organizations
\(\mathcal A_a\) and positive labor shares \(\theta_a\) summing to one,
write \(C_a=C_{\mathcal A_a}\) and put
\(Y=\sum_a\theta_a/C_a\) and \(\alpha_a=(\theta_a/C_a)/Y\).
Add a root \((1,q)\) with children \((\alpha_a,q)\), identifying
each child with the root of an \(\alpha_a\)-scaled copy of
\(\mathcal A_a\). The new root is coherent and performs zero work:
\(d(q)-\sum_a\alpha_a d(q)=0\). Branch \(a\)'s labor requirement
is \(\alpha_aC_a\), so
\[
C=\sum_a\alpha_aC_a=1/Y,
\qquad
\frac{\alpha_aC_a}{C}=\theta_a.
\]
Give the new root its own matched zero-employment position and keep
branch positions distinct. Within each branch, relative occupational shares are unchanged,
and the resulting finite organization has
\begin{equation}
Y=\sum_a\theta_aY_a,
\qquad
B=\sum_a\theta_aB_a,
\qquad
S=\sum_a\theta_aS_a,
\qquad
\theta_a\geq0\quad\text{for every }a,
\qquad
\sum_a\theta_a=1.
\label{eq:proof-target-labor-mixture}
\end{equation}
Branches with zero weight are simply omitted. Consequently, the set of
finite-organization outcome triples is convex. Let \(\mathcal O\) denote
its closure. The planner maximizes
\begin{equation}
Y+\log Y+(\epsilon-f)\log B+fS.
\label{eq:proof-target-objective}
\end{equation}
For \(0\leq f\leq\epsilon\), this objective is concave and increasing in \(Y\) and weakly increasing in \(B\).

\emph{Step 2. Obtain a welfare optimum and its supporting gradient.}
Because a domain-specific issue in domain \(k\) is understood at level \(s_k\), the symmetric single-domain component contributes \((\nu_1/K)\sum_ks_k=\nu_1/K\) to every worker's understanding. Thus
\(B\geq\nu_1/K>0\) and \(S\geq\log(\nu_1/K)\), while \(B\leq1\), \(S\leq0\), and \(0\leq Y\leq Y^M\). These bounds are
preserved under closure, so \(\mathcal O\) is compact. Extending welfare
to minus infinity at \(Y=0\) makes it upper semicontinuous. Since a finite
exactly matched organization has finite welfare, a maximizer exists and
has \(Y>0\).
At a welfare optimum \(z^*=(Y^*,B^*,S^*)\), the gradient
\[
c^*=\left(1+\frac1{Y^*},\frac{\epsilon-f}{B^*},f\right)
\]
satisfies \(c^*\cdot(z-z^*)\leq0\) for every \(z\in\mathcal O\):
the segment toward \(z\) lies in the convex set \(\mathcal O\),
and welfare's right derivative along it cannot be positive. Thus
\(z^*\) maximizes the linear objective \(c^*\cdot z\).

\emph{Step 3. Bound the linear objective and establish a zero maximal residual.}
Fix \(\eta>0\) and \(\lambda\geq0\), and let
\begin{equation}
\begin{aligned}
H
&=\sup_{\mathcal A\in\Omega(q)}\quad
\{\eta Y_{\mathcal A}+\lambda B_{\mathcal A}+fS_{\mathcal A}\}\\
&=\sup_{\mathcal A\in\Omega(q)}\quad
\frac{\eta+\lambda N_{\mathcal A}+f\Gamma_{\mathcal A}}
{C_{\mathcal A}}.
\end{aligned}
\label{eq:proof-target-support}
\end{equation}
At the planner's optimum, \(\eta=1+1/Y\) and \(\lambda=(\epsilon-f)/B\). Define \(h(b)=\lambda b+f\log b\).
The residual is \(\eta+\lambda N_{\mathcal A}+f\Gamma_{\mathcal A}
-HC_{\mathcal A}\), or \(C_{\mathcal A}\) times the gap between
the ratio and \(H\). We will show that at the support value,
\begin{equation}
0=\sup_{\mathcal A\in\Omega(q)}\quad
\{\eta+\lambda N_{\mathcal A}+f\Gamma_{\mathcal A}-HC_{\mathcal A}\}.
\label{eq:proof-target-dinkelbach}
\end{equation}
Let \(G=2-1/K\) denote the total work required to produce one unit of
output under the uniform production profile. An organization in which the
exactly matched root performs all the work has
\(C_{\mathcal A}=G\ell_K\), \(B_{\mathcal A}=b_K\), and
\(S_{\mathcal A}=\log b_K\). Hence
\begin{equation}
H\geq \frac{\eta}{G\ell_K}+h(b_K)>h(b_K).
\label{eq:proof-target-root-bound}
\end{equation}
Averaging \(s\) over all domain permutations gives \(q\), so symmetry
and concavity imply \(b(s)\leq b_K\). Since \(h\) is nondecreasing,
\(H-h(b(s))\geq\delta:=H-h(b_K)>0\). This strict gap also
justifies \eqref{eq:proof-target-dinkelbach}. Along any sequence whose ratio approaches \(H\), the identity
\[
\frac{\eta+\sum_vx_vh(b(s_v))}{C_{\mathcal A}}
=\frac{\eta}{C_{\mathcal A}}+
\frac{\sum_vx_vh(b(s_v))}{C_{\mathcal A}}
\leq\frac{\eta}{C_{\mathcal A}}+h(b_K)
\]
implies \(C_{\mathcal A}\leq2\eta/\delta\) once the ratio is at
least \(H-\delta/2\). Multiplying its vanishing gap from \(H\) by
this bounded labor requirement makes the residual converge to zero.
Every residual is nonpositive by definition of \(H\).

\emph{Step 4. Bound assignment returns using uniform responsibilities.}
For any scope-\(K\) assignment, \(\beta(s,p)\leq1\) and
\(b(s)\leq b_K\). Since the bracketed support term is negative,
\[
\frac{\ell_K}{\beta(s,p)}[h(b(s))-H]
\leq \ell_K[h(b_K)-H],
\]
and the uniform root attains equality. This bounds the residual return from any scope-\(K\) assignment by the return at a uniform, exactly matched root.

Now consider a nonroot responsibility with profile \(p\) and scope
\(t<K\). Its required domains form a set \(T\). For an assigned
occupational profile \(s\), let
\(\alpha=\sum_{k\in T}s_k\). Summing \(\beta(s,p)p_k\leq s_k\) over \(k\in T\) gives
\(0<\beta(s,p)\leq\alpha\) at each positive-work assignment. For a bound on this assignment's return,
average its profiles over permutations within \(T\) and its complement.
The resulting responsibility is uniform on \(T\), and the occupational
profile becomes \(s_{T,\alpha}\), with coordinates
\begin{equation}
\bigl[s_{T,\alpha}\bigr]_k=
\begin{cases}
\alpha/t,&k\in T,\\
(1-\alpha)/(K-t),&k\notin T.
\end{cases}
\label{eq:proof-target-two-level-profile}
\end{equation}
The symmetrized assignment has fit \(\alpha\); symmetry and concavity
give \(b(s_{T,\alpha})\geq b(s)\). Since \(h(b)-H<0\), the
larger fit and weakly greater understanding imply
\begin{equation}
\frac{\ell_t}{\alpha}
[h(b(s_{T,\alpha}))-H]
\geq
\frac{\ell_t}{\beta(s,p)}
[h(b(s))-H].
\label{eq:proof-target-two-level-dominance}
\end{equation}
These are bounds on the return per unit of assigned work. Below, a coherent hierarchy with uniform responsibilities attains the bounds simultaneously, so this symmetrization does not require independently replacing nodes of the original tree.
Values \(\alpha<t/K\) can be discarded: replacing them by \(\alpha=t/K\) strictly raises fit and weakly raises understanding, which strictly improves the negative residual return. At \(\alpha=t/K\), the worker is uniform across all domains. Performing scope-\(t\) work with this profile requires \(K\ell_t/t\) units of labor per unit of work, compared with \(\ell_K\) at the root. The affine schedule gives
\[
\frac{K\ell_t}{t}-\ell_K
=(1-\kappa)\left(\frac Kt-1\right)>0.
\]
For \(\alpha\in[t/K,1]\), the two-level profile lies on the segment joining the uniform profiles on \(T\) and on all \(K\) domains. Required-domain knowledge is at least outside knowledge on this interval.
An issue uniform on a \(j\)-domain set \(U\) therefore has fit
\(j\alpha/t\) if \(U\subseteq T\), and \(j(1-\alpha)/(K-t)\)
otherwise. Dividing by \(\ell_j\) and averaging over the agenda makes
understanding affine, with endpoints \(b_K\) and \(b_t\):
\begin{equation}
\bar b_t(\alpha)
=b_K-d_t\left(\alpha-\frac tK\right),
\qquad
d_t=\frac{K(b_K-b_t)}{K-t},
\label{eq:proof-target-understanding-segment}
\end{equation}
The matched understanding levels in \eqref{eq:proof-br} satisfy
\begin{equation}
b_t
=\sum_{j=1}^t
\frac{\nu_j}{\ell_j}
\frac{\binom{t-1}{j-1}}{\binom Kj}.
\label{eq:proof-target-bt}
\end{equation}
For each nonroot scope \(t<K\), choose
\[
\alpha_t
\in
\arg\max_{\alpha\in[t/K,1]}
\quad
\frac{\ell_t}{\alpha}
[h(\bar b_t(\alpha))-H].
\]
The objective is continuous on \([t/K,1]\), since
\(\bar b_t(\alpha)\geq\nu_1/K>0\). Thus a maximizer exists with
\begin{equation}
\frac tK\leq\alpha_t\leq1.
\label{eq:proof-target-own-support}
\end{equation}

\emph{Step 5. Show that optimizing profiles' understanding weakly increases with scope.}
Pascal's identity in \eqref{eq:proof-target-bt} gives nonnegative first
and second differences of \(b_t\): each dimension's contributions are
proportional to \(\binom{t-1}{j-2}\) and \(\binom{t-1}{j-3}\),
respectively, with out-of-range coefficients zero. Hence \(b_t\) is
weakly increasing and discretely convex.

If \(b_t=b_K\), then \(\bar b_t(\alpha)=b_K\) throughout the segment. Because the bracketed support term is negative, \(\alpha=1\) uniquely maximizes the return for scope \(t\). Monotonicity of \(b_t\) implies that every broader scope then also has understanding \(b_K\). To establish the ordering, it remains to compare scopes with \(b_t<b_K\).

For \(f>0\) and \(b_t<b_K\), write
\(R_t(\alpha)=\ell_t[h(\bar b_t(\alpha))-H]/\alpha\).
Differentiating gives
\[
R_t'(\alpha)=\frac{\ell_t}{\alpha^2}
\{H-h(\bar b_t(\alpha))-\alpha d_th'(\bar b_t(\alpha))\}.
\]
An interior maximizer satisfies
\begin{equation}
H-h(\bar b_t(\alpha_t))
=\alpha_td_th'(\bar b_t(\alpha_t)).
\label{eq:proof-target-alpha-foc}
\end{equation}
At \(\alpha=t/K\) a maximizer requires \(R_t'(t/K)\leq0\);
at \(\alpha=1\) it requires \(R_t'(1)\geq0\). Define
\(D_t=t(b_K-b_t)/(K-t)\).
Using \(A_t=b_K+D_t\), write
\(b=\bar b_t(\alpha)=A_t-d_t\alpha\). The return in the
understanding coordinate \(b\in[b_t,b_K]\) is
\[
\widetilde R_t(b)=\ell_td_t\frac{h(b)-H}{A_t-b},
\qquad A_t-b=d_t\alpha>0.
\]
Define
\begin{equation}
Q_t(b)
=[D_t+b_K-b]h'(b)-[H-h(b)].
\label{eq:proof-target-understanding-root}
\end{equation}
Interior stationarity is \(Q_t(b)=0\). The derivative
\(\widetilde R_t'(b)=\ell_td_tQ_t(b)/(A_t-b)^2\) has the sign
of \(Q_t\), and
\(Q_t'(b)=[D_t+b_K-b]h''(b)<0\) since \(h''(b)=-f/b^2<0\). Consequently the unique maximizer is
\(b_t\) if \(Q_t(b_t)\leq0\), \(b_K\) if
\(Q_t(b_K)\geq0\), and otherwise the unique zero in
\((b_t,b_K)\). This also identifies a unique \(\alpha_t\).
The secant \((b_K-b_t)/(K-t)\) averages the successive differences
\(b_{u+1}-b_u\), \(u=t,\ldots,K-1\). Discrete convexity makes
these averages weakly increasing in \(t\), so their nonnegative
products with \(t\), namely \(D_t\), are also weakly increasing.
For \(t<t'\) with \(b_{t'}<b_K\), on the common interval,
\[
Q_{t'}(b)-Q_t(b)=(D_{t'}-D_t)h'(b)\geq0.
\]
If scope \(t\)'s optimum lies below \(b_{t'}\), ordering is
immediate. Otherwise the upward shift in \(Q_t\) cannot move its
maximizing point downward: an interior zero or the condition
\(Q_t(b_K)\geq0\) moves weakly right, and a lower-endpoint optimum
is already no greater than \(b_{t'}\). Scopes with \(b_t=b_K\)
have been handled above. Therefore
\begin{equation}
\bar b_1(\alpha_1)
\leq\bar b_2(\alpha_2)
\leq\cdots\leq b_K.
\label{eq:proof-target-understanding-order}
\end{equation}
At \(f=0\), Step 7 selects matched-cutoff maximizers from Theorem
\ref{thm:optimal-responsibility}, whose understanding levels \(b_t\)
have the same ordering.

\emph{Step 6. Attain all assignment and workload bounds with a finite coherent hierarchy.}
Distribute the constant \(\eta\) over the \(G\) units of work and
define the best residual return at each scope:
\begin{equation}
r_t(H)
=\frac{\eta}{G}
+\max_{\alpha\in[t/K,1]}\quad
\frac{\ell_t}{\alpha}
[h(\bar b_t(\alpha))-H]
\quad(t<K),
\qquad
r_K(H)=\frac{\eta}{G}+\ell_K[h(b_K)-H].
\label{eq:proof-target-scope-return}
\end{equation}
If \(X_t\) is an organization's scope-\(t\) workload, its residual
is at most \(\sum_tX_tr_t(H)\). Lemma
\ref{lem:proof-workload-prefix} represents these workloads as an upward
transport of \(g_1=1\) and \(g_t=1/[t(t-1)]\) for \(t\geq2\).
That is, nonnegative \(T_{js}\), \(j\leq s\), satisfy
\(\sum_{s\geq j}T_{js}=g_j\) and \(\sum_{j\leq s}T_{js}=X_s\).
Hence
\[
\sum_sX_sr_s(H)=\sum_j\sum_{s\geq j}T_{js}r_s(H)
\leq\sum_jg_j\max_{s\geq j}r_s(H).
\]
This workload bound is an accounting representation; a feasible tree
attaining it is constructed next. Write the bound as
\begin{equation}
\mathcal D(H)
=\sum_{j=1}^Kg_j\max_{j\leq s\leq K}r_s(H).
\label{eq:proof-target-exact-residual}
\end{equation}
For each source scope \(j\), let \(k(j)\) be the smallest scope that maximizes \(r_s(H)\) over \(s\geq j\):
\[
k(j)
=\min\left\{
s\in\{j,\ldots,K\}:
r_s(H)=\max_{t\in\{j,\ldots,K\}}r_t(H)
\right\}.
\]
If \(k(j)>j\), removing \(j\) from the maximizing set leaves
\(k(j+1)=k(j)\); if \(k(j)=j\), the next maximizer is larger.
Thus \(k\) is weakly increasing and its distinct values form a chain
\[
0=\widehat r_0<\widehat r_1<\cdots<\widehat r_m=K,
\qquad
k(j)=\widehat r_a\quad
\text{for }\widehat r_{a-1}<j\leq\widehat r_a.
\]
Start with the uniform root of scale one. A node with required set
\(T\), of size \(\widehat r_a\) with \(a\geq2\), delegates equally among all
\(\widehat r_{a-1}\)-element subsets, each child uniform on its subset.
Every domain belongs to a fraction
\(\widehat r_{a-1}/\widehat r_a\) of those subsets, so its average
child weight is \(1/\widehat r_a\): each division is coherent.
Scope \(\widehat r_1\) is terminal. The tree has finitely many levels
and children per node, with total scale one at each active scope.
Nonterminal scope \(\widehat r_a\) therefore performs
\[
\frac1{\widehat r_{a-1}}-\frac1{\widehat r_a}
=\sum_{j=\widehat r_{a-1}+1}^{\widehat r_a}g_j
\qquad(a\geq2)
\]
units of work; terminal scope \(\widehat r_1\) performs \(2-1/\widehat r_1=\sum_{j=1}^{\widehat r_1}g_j\). Give each node its own connected position, match the root exactly,
and assign \(s_{T,\alpha_t}\) at each active nonroot node.
Coherence constrains responsibilities, so these knowledge choices are
feasible independently across positions. Symmetry gives every node of
scope \(t\) its optimized return; each workload block goes to its
maximizing destination. Thus this finite organization attains
\eqref{eq:proof-target-exact-residual}.
No active nonroot scope can use \(\alpha_t=t/K\). At that endpoint, \(\bar b_t(\alpha_t)=b_K\) and
\[
r_t(H)
=\frac{\eta}{G}+\frac{K\ell_t}{t}[h(b_K)-H]
<\frac{\eta}{G}+\ell_K[h(b_K)-H]
=r_K(H),
\]
where the strict inequality follows from \(h(b_K)-H<0\) and
\(K\ell_t/t>\ell_K\). Since \(K\) is an available destination
whenever \(t\) is, such a scope cannot be chosen by \(k(j)\).
Every active nonroot scope therefore satisfies \(t/K<\alpha_t\leq1\).
Thus the hierarchy is pruned and broad-to-narrow, with each active
nonroot position matched or cross-trained as stated.

Each \(r_t(H)\) is continuous by maximization over a fixed compact
interval. If \(H_2>H_1\), each candidate falls by at least
\(H_2-H_1\), so
\[
\mathcal D(H_2)\leq\mathcal D(H_1)-G(H_2-H_1).
\]
Therefore \(\mathcal D\) is continuous, strictly decreasing, and tends
to minus infinity. At \(H=h(b_K)\), the root gives
\(\mathcal D(H)\geq\eta>0\). At the support value,
\eqref{eq:proof-target-dinkelbach} and the bound just attained give
\(\mathcal D(H)=0\). Hence this unique zero is attained by a finite
hierarchy, with the profile ordering from Step 5 when \(f>0\).

\emph{Step 7. The same compact set of hierarchies attains every
supporting maximum.}
Let \(\mathcal P\subset\mathbb R^3\) be the outcomes of the following finite hierarchies. For every chain
\(0=\widehat r_0<\widehat r_1<\cdots<\widehat r_m=K\), assign each source workload \(g_j\) with \(\widehat r_{a-1}<j\leq\widehat r_a\) to scope \(\widehat r_a\), implement these assignments using the canonical pruned, broad-to-narrow hierarchy above, and allow each active nonroot scope \(t\) to use any \(\alpha_t\in[t/K,1]\). Keep the root exactly matched. There are finitely many chains, each
indexed by a subset of \(\{1,\ldots,K-1\}\). For a fixed chain,
parameters range over a finite product of compact intervals; fits are
bounded away from zero and understanding is at least \(\nu_1/K\).
Thus \(C,N,\Gamma\), and hence outcomes, vary continuously.
Their finite union \(\mathcal P\) is compact and consists entirely
of finite-organization outcomes.

For \(\gamma>0\), equations \eqref{eq:proof-target-support}--\eqref{eq:proof-target-exact-residual} apply with \(h(b)=\lambda b+\gamma\log b\). When \(\gamma=0\) and \(\lambda>0\), let \(H\) be the corresponding support value and set \(\bar b=H/\lambda\). The all-root organization gives \(H>\lambda b_K\), so \(\bar b>b_K\); Lemma \ref{lem:proof-assignment-envelope} and the full matched-cutoff construction in the proof of Theorem \ref{thm:optimal-responsibility} then supply a maximizer in \(\mathcal P\). These full cutoffs are the finitely many organizations indexed by
\(j=0,\ldots,K-1\) in \eqref{eq:proof-cutoff-vector}; their chains
are included in \(\mathcal P\), with \(\alpha_t=1\) throughout.
The assignment envelope sends source work either to itself or the root,
and ties can be resolved at a full cutoff. When
\(\lambda=\gamma=0\), a labor-minimizing cutoff supplies a maximizer.
For every \(\eta>0\) and \(\lambda,\gamma\geq0\), we therefore have
\begin{equation}
\begin{aligned}
\mathsf H(\eta,\lambda,\gamma)
&:=
\sup_{(Y,B,S)\in\mathcal O}
\{\eta Y+\lambda B+\gamma S\}\\
&=
\max_{(Y,B,S)\in\mathcal P}
\{\eta Y+\lambda B+\gamma S\}.
\end{aligned}
\label{eq:proof-target-compact-oracle}
\end{equation}
Taking the closure preserves linear suprema, while the residual equation and the construction in Step 6 supply a member of \(\mathcal P\) that attains the unrestricted supremum.

\emph{Step 8. A welfare optimum combines at most three hierarchies,
or two at the endpoints.}
For a nonzero coefficient vector \(c\), write the sets of maximizers
over \(\mathcal O\) and \(\mathcal P\) as
\[
F_{\mathcal O}(c)=\{z\in\mathcal O:c\cdot z=\mathsf H(c)\},
\qquad M(c)=\operatorname*{arg\,max}_{z\in\mathcal P}c\cdot z.
\]
If the two support functions agree on a neighborhood of \(c\), then
\[
F_{\mathcal O}(c)=\operatorname{co}M(c).
\]
To prove the nontrivial inclusion, fix \(z\in F_{\mathcal O}(c)\)
and any direction \(d\). For small \(t>0\), compactness supplies
\(z_t\in\mathcal P\) maximizing \((c+td)\cdot z_t\). Local
support equality gives
\[
\mathsf H(c)+t\,d\cdot z
\leq\mathsf H(c+td)
=c\cdot z_t+t\,d\cdot z_t
\leq\mathsf H(c)+t\,d\cdot z_t.
\]
A subsequence as \(t\downarrow0\) converges to
\(\bar z\in M(c)\), by compactness and continuity of the support
function. Thus \(d\cdot z\leq\max_{w\in M(c)}d\cdot w\) for
every \(d\). Separation from the compact convex set
\(\operatorname{co}M(c)\) proves the inclusion. The reverse inclusion
follows from \(\mathcal P\subseteq\mathcal O\) and convexity.
Thus equality on a neighborhood identifies the whole set of maximizers.

For \(0<f<\epsilon\), the supporting welfare gradient \(c^*\)
in Step 2 has three positive coordinates. Equation
\eqref{eq:proof-target-compact-oracle} holds on a neighborhood of
\(c^*\), so the preceding argument gives
\(z^*\in\operatorname{co}M(c^*)\).
Every hierarchy generating a point of \(M(c^*)\) has zero residual
and binds the assignment and workload bounds: all these inequalities
bound its residual above by \(\mathcal D(H)=0\), so any strict
inequality would contradict its zero residual.
For \(f>0\), each scope with \(b_t<b_K\) has the unique optimizing
understanding characterized by \(Q_t\); a scope with \(b_t=b_K\)
uses \(\alpha_t=1\). Strict root dominance excludes
\(\alpha_t=t/K\) at active nonroot scopes. These equality conditions
give the profile bounds and understanding ordering in every selected
hierarchy; the root is matched by definition of \(\mathcal P\).
All points of \(M(c^*)\) lie in the supporting plane
\(c^*\cdot z=\mathsf H(c^*)\), of affine dimension at most two.
Carath\'eodory's theorem on this plane represents \(z^*\) as a
convex combination of at most three such outcomes.

At \(f=0\), use the finite set of full matched-cutoff organizations constructed in the proof of Theorem \ref{thm:optimal-responsibility}. For every \(\eta,\lambda>0\), the residual equation above and the assignment-envelope bound in that proof show that this set attains
\[
\sup_{(Y,B)\in\operatorname{proj}_{Y,B}\mathcal O}
\{\eta Y+\lambda B\}.
\]
The support function of the \((Y,B)\) outcomes generated by these organizations therefore agrees, in a neighborhood of \((1+1/Y^*,\epsilon/B^*)\), with that of \(\operatorname{proj}_{Y,B}\mathcal O\). The same exposed-face argument places \((Y^*,B^*)\) in the convex hull of exact support maximizers lying on one line, so two outcomes suffice; \(S\) does not enter welfare. At \(f=\epsilon\), apply the argument to the restriction \(\mathsf H(\eta,0,\gamma)\) at \((\eta,\gamma)=(1+1/Y^*,\epsilon)\). Support equality holds on a neighborhood in this two-dimensional coefficient space. The \((Y,S)\) projections of its exact maximizers again lie on one line and their convex hull contains \((Y^*,S^*)\), so two outcomes suffice; \(B\) does not enter welfare. At \(f=0\), the selected components are matched and inherit the
ordering of \(b_t\); at \(f=\epsilon>0\), the same equality
conditions as above give the profile bounds and ordering.
Finally, \eqref{eq:proof-target-labor-mixture} implements each convex
combination as one finite organization, preserving the component
profiles. Its welfare equals the relaxed optimum, so it is a global
optimum over the original feasible class.
\hfill\rule{0.5em}{0.5em}\par\egroup
\hypertarget{artificial-intelligence}{%
\section{Artificial Intelligence}\label{artificial-intelligence}}

\label{app:proofs-ai}

\hypertarget{proof-of-proposition-1}{%
\subsection{\texorpdfstring{Proof of Proposition
\ref{prop:ai-sorting}}{Proof of Proposition }}\label{proof-of-proposition-1}}

\label{app:proof-ai-sorting}

Use the notation and assumptions of Section~\ref{sec:ai}. We first allow
the firm to choose both the organization and the work assigned to AI,
and show that a global cost minimum is attained on a productively optimal
pre-AI hierarchy. The understanding and planner comparisons then hold
that hierarchy fixed, with proportional displacement within each scope.

For the unrestricted problem, write capacity use and remaining human
labor as
\[
I(\mathcal A,\xi)
=\sum_{v\in T}a_{r(p_v)}^{\mathrm{AI}}\xi_v,
\qquad
C_H(\mathcal A,\xi)
=\sum_{v\in T}(z_v-\xi_v)
\frac{\ell_{r(p_v)}}{\beta(s_v,p_v)}.
\]
Here \(\mathcal A\in\Omega(q)\), \(0\leq\xi_v\leq z_v\), and
global feasibility requires \(I(\mathcal A,\xi)\leq A\).
On a productively optimal pre-AI hierarchy, human knowledge is exactly
matched and the scope workloads are \(g_r(q)>0\). Aggregating automated
work gives \(x_r=\sum_{v:r(p_v)=r}\xi_v\). Conversely, every
\(x\in\mathcal X(A)\) is implemented by automating the fraction
\(x_r/g_r(q)\) at each scope-\(r\) node. Proportional displacement
preserves the layer's average understanding \(b_r\) while it retains
human work; a fully automated layer has zero employment weight. For
the proof, denote the remaining understanding numerator by
\[
N_H(x)=\sum_{r=1}^n\ell_r[g_r(q)-x_r]b_r,
\qquad B(x)=\frac{N_H(x)}{C_H(x)},
\]
and write \(I(x)=\sum_ra_r^{\mathrm{AI}}x_r\). The maintained
capacity restriction gives \(C_H>0\) on \(\mathcal X(A)\), and the
assumed positivity of \(B\) gives \(N_H>0\).

\begin{lemma}[Production hierarchy and AI allocation]
\label{lem:proof-ai-separation}
The minimum human labor requirement over all feasible pairs
\((\mathcal A,\xi)\) is attained on a productively optimal pre-AI
hierarchy. On that hierarchy, the firm solves
\begin{equation}
\min_{x\in\mathcal X(A)}\quad
\sum_{r=1}^n\ell_r[g_r(q)-x_r].
\label{eq:proof-ai-firm}
\end{equation}
Equivalently, with \(A_r=a_r^{\mathrm{AI}}x_r\) denoting capacity
assigned to scope \(r\), it solves
\begin{equation}
\begin{aligned}
\max_{(A_r)_r}\quad
&\sum_{r=1}^n\rho_rA_r\\
\text{subject to}\quad
&\sum_rA_r\leq A,\\
&0\leq A_r\leq a_r^{\mathrm{AI}}g_r(q)
\quad\text{for every }r.
\end{aligned}
\label{eq:proof-ai-knapsack}
\end{equation}
\end{lemma}

\bgroup \par\noindent\textbf{Proof.}\ 
We establish a lower bound for all organizations and show that the
selected hierarchy attains it. Linear-programming duality then restores
the capacity constraint, and an exchange argument gives the allocation
rule.

\emph{Step 1. Bound every organization's Lagrangian cost.}
For \(\lambda\geq0\), attach the term \(\lambda(I-A)\) to labor
cost and remove the aggregate capacity constraint, retaining
\(0\leq\xi_v\leq z_v\). Thus \(\lambda\) values one unit of AI
capacity in labor units. The cheaper way to perform one unit of
scope-\(r\) work has cost
\(c_r(\lambda)=\min\{\ell_r,\lambda a_r^{\mathrm{AI}}\}\).
Both arguments are nondecreasing in scope, so \(c_r(\lambda)\) is
nondecreasing. Since \(\beta(s_v,p_v)\leq1\), unmatched human
knowledge cannot lower this cost. With
\(Z_r(\mathcal A)=\sum_{v:r(p_v)=r}z_v\), every organization and
nodewise-feasible allocation therefore satisfy
\begin{equation}
C_H(\mathcal A,\xi)+\lambda I(\mathcal A,\xi)
=\sum_{v\in T}\left\{
(z_v-\xi_v)\frac{\ell_{r(p_v)}}{\beta(s_v,p_v)}
+\lambda a_{r(p_v)}^{\mathrm{AI}}\xi_v
\right\}
\geq\sum_rc_r(\lambda)Z_r(\mathcal A).
\label{eq:proof-ai-global-lower-bound}
\end{equation}
Each human or automated unit at scope \(r\) costs at least
\(c_r(\lambda)\); together these units account for the node's full
workload.

\emph{Step 2. Show that one pre-AI hierarchy attains the bound.}
First refine an arbitrary tree until every terminal profile is pure.
At a terminal profile \(p\) of scale \(\omega\) and scope \(r>1\),
retain \(p\) as an internal node and add a pure-domain child \(e_i\)
of scale \(\omega p_i\) for each \(p_i>0\). This is coherent:
the child scales sum to \(\omega\), and their scale-weighted profiles
sum to \(\omega p\). The parent handles \(\omega d(p)\) interface
work, while the children execute \(\omega\) component work. Hence
the contribution to the bound changes from
\(\omega c_r(\lambda)[1+d(p)]\) to
\(\omega[c_r(\lambda)d(p)+c_1(\lambda)]\), a weak reduction.
It suffices to bound costs on the refined trees.
For such a tree, let \(T_h=\sum_{r=h}^nZ_r\) be workload at scope
\(h\) or above. Work conservation gives \(\sum_rZ_r=1+d(q)\).
Expanding \(c_r=c_1+\sum_{h=2}^r(c_h-c_{h-1})\) yields
\[
\sum_{r=1}^n c_r(\lambda)Z_r
=c_1(\lambda)[1+d(q)]
+\sum_{h=2}^n[c_h(\lambda)-c_{h-1}(\lambda)]T_h
\geq\sum_{r=1}^n c_r(\lambda)g_r(q).
\]
The inequality uses the cut bound \eqref{eq:proof-cut-bound} at
scope \(h-1\), namely \(T_h\geq\sum_{r=h}^ng_r(q)\), with
nonnegative coefficients, and \(\sum_rg_r(q)=1+d(q)\).
A productively optimal hierarchy attains all these tail bounds at once.
Matching human knowledge and assigning every unit to the cheaper
technology also attains equality in
\eqref{eq:proof-ai-global-lower-bound}; either technology or a mixture
is permissible at a tie. Consequently, for every \(\lambda\geq0\),
\[
\begin{aligned}
\inf_{\mathcal A\in\Omega(q),\ 0\leq\xi_v\leq z_v}
\{C_H(\mathcal A,\xi)+\lambda[I(\mathcal A,\xi)-A]\}
&=\min_{0\leq x_r\leq g_r(q)}
\{C_H(x)+\lambda[I(x)-A]\}\\
&=\sum_{r=1}^n c_r(\lambda)g_r(q)-\lambda A.
\end{aligned}
\]
At \(\lambda=0\), all work can be automated in this relaxed problem,
whose aggregate capacity constraint has been removed.

\emph{Step 3. Transfer optimality to the constrained global problem.}
Problem \eqref{eq:proof-ai-firm} is a feasible finite linear program
with compact feasible set. Strong duality provides a minimizer \(x^*\)
and a multiplier \(\lambda^*\geq0\) such that
\[
C_H(x^*)
=\min_{0\leq x_r\leq g_r(q)}
\{C_H(x)+\lambda^*[I(x)-A]\},
\qquad
\lambda^*[I(x^*)-A]=0.
\]
The condition \(\lambda^*[I(x^*)-A]=0\) requires full use of
capacity when \(\lambda^*>0\). For any globally feasible pair
\((\widetilde{\mathcal A},\widetilde\xi)\),
\[
\begin{aligned}
C_H(\widetilde{\mathcal A},\widetilde\xi)
&\geq C_H(\widetilde{\mathcal A},\widetilde\xi)
  +\lambda^*[I(\widetilde{\mathcal A},\widetilde\xi)-A]\\
&\geq C_H(x^*)+\lambda^*[I(x^*)-A]
=C_H(x^*).
\end{aligned}
\]
Feasibility gives the first inequality, Step 2 gives the second, and
\(\lambda^*[I(x^*)-A]=0\) gives the equality. Thus the hierarchy's
constrained optimum also minimizes labor over all organizations.
Its capacity constraint binds: if capacity were unused, the restriction
\(A<\sum_ra_r^{\mathrm{AI}}g_r(q)\) leaves some human work, and a
small feasible increase in its automation strictly reduces labor.

\emph{Step 4. Rank scopes by labor savings per unit of capacity.}
Substitution gives \(C_H(x)=C_H(0)-\sum_r\rho_rA_r\), proving
\eqref{eq:proof-ai-knapsack}. If \(\rho_r>\rho_t\), scope \(r\)
is below its cap, and \(A_t>0\), moving
\(0<\delta\leq\min\{a_r^{\mathrm{AI}}g_r(q)-A_r,\ A_t\}\)
units of capacity from \(t\) to \(r\) preserves feasibility and
saves \(\delta(\rho_r-\rho_t)>0\) additional labor. Hence optimal
allocations fill scopes in descending order of \(\rho_r\), with
arbitrary allocation among ties. Conversely, relative to this ordering,
any other full-capacity allocation moves capacity from weakly higher
returns to weakly lower returns, so cannot improve the objective.
With distinct ratios, only the last scope receiving capacity can be
partially automated.
\hfill\rule{0.5em}{0.5em}\par\egroup

\begin{lemma}[Average understanding]
\label{lem:proof-ai-understanding}
For \(x\in\mathcal X(A)\) with \(C_H(x),N_H(x)>0\),
\begin{equation}
\frac{\partial B(x)}{\partial x_r}
=\frac{\ell_r}{C_H(x)}[B(x)-b_r],
\qquad
\frac{\partial B(x)}{\partial A_r}
=\rho_rY(x)[B(x)-b_r].
\label{eq:proof-ai-understanding}
\end{equation}
\end{lemma}

\bgroup \par\noindent\textbf{Proof.}\ 
With the hierarchy and its \(b_r\)'s fixed, differentiate
\(B=N_H/C_H\), using \(\partial N_H/\partial x_r=-\ell_rb_r\)
and \(\partial C_H/\partial x_r=-\ell_r\). Division by
\(a_r^{\mathrm{AI}}\) gives the second expression. The same formulas
give feasible one-sided changes at boundaries. Automation reduces the
layer's share of human work, so understanding rises precisely when that
layer's understanding is below the current mean.
\hfill\rule{0.5em}{0.5em}\par\egroup

\medskip
\bgroup \par\noindent\textbf{Proof of Proposition \ref{prop:ai-sorting}.}\ 
Lemma~\ref{lem:proof-ai-separation} proves the global firm claim and
allocation rule. On the fixed hierarchy with proportional displacement,
Theorem~\ref{thm:division-understanding} gives
\(b_1\leq\cdots\leq b_n\). When broader scopes are automated
first, the layer currently receiving AI has the greatest \(b_r\)
among layers retaining human employment. Thus \(b_r\geq B\), and
\eqref{eq:proof-ai-understanding} makes \(B\) weakly decreasing
while that layer is automated. Continuity of \(B=N_H/C_H\), with
\(C_H>0\), joins these comparisons across layer exhaustion points.
With narrower scopes first, the current layer has the lowest remaining
\(b_r\), so \(b_r\leq B\) and understanding weakly increases.
The comparisons concern these specified orderings; tied firm returns
alone do not select an ordering among the tied scopes.

For the planner's objective \(W=Y+\log Y+\epsilon\log B\),
Lemma~\ref{lem:proof-ai-understanding} and
\(\partial Y/\partial A_r=\rho_rY^2\) give
\[
\frac{\partial W}{\partial A_r}
=\rho_rY\left[Y+1+\epsilon\left(1-\frac{b_r}{B}\right)\right].
\]
This derivative is strictly positive whenever \(b_r\leq B\).
If capacity is unused, choose the lowest-understanding layer retaining
human work. Its \(b_r\) is at most the employment-weighted mean
\(B\), and a small increase in its automation is feasible and raises
welfare. Thus every planner optimum uses all capacity, whether or not
AI is scope-neutral.
At neutrality, \(\rho_r=\rho>0\) for all \(r\). Write
\(C_H^0=C_H(0)\) and \(N_H^0=N_H(0)\). Every allocation with
\(\sum_rA_r=A\) satisfies
\begin{equation}
C_H=C_H^0-\rho A,
\qquad
N_H=N_H^0-\rho\sum_rb_rA_r.
\label{eq:proof-ai-neutral}
\end{equation}
Output is therefore constant across these allocations. At fixed
\(C_H>0\), welfare increases with \(N_H\), since
\(\partial W/\partial N_H=\epsilon/N_H>0\). The planner minimizes
\(\sum_rb_rA_r\). If a higher-\(b_r\) layer receives capacity
while a lower-\(b_r\) layer is below its cap, transferring capacity
to the latter strictly reduces this sum. Hence the planner fills scopes
in ascending order of \(b_r\), with arbitrary allocation among ties.
\hfill\rule{0.5em}{0.5em}\par\egroup

\hypertarget{proof-of-theorem-6}{%
\subsection{\texorpdfstring{Proof of Theorem
\ref{thm:ai-reversal}}{Proof of Theorem }}\label{proof-of-theorem-6}}

\label{app:proof-ai-reversal}

\bgroup \par\noindent\textbf{Proof.}\ 
We first construct a feasible welfare improvement on a global firm
optimum. Under the stronger understanding-gap hypothesis, a uniform
comparison then establishes optimality against all allocations on the
fixed hierarchy.
Fix a productively optimal pre-AI hierarchy. Let \(C^0\) and
\(N^0\) be its no-AI labor requirement and understanding numerator.
The hypothesis \(b_1<b_n\) implies \(n\geq2\). Consider
\begin{equation}
\rho_r(\zeta)=\rho[1+\zeta\delta_r],
\qquad
0=\delta_1<\delta_2<\cdots<\delta_n=1.
\label{eq:proof-ai-tilt}
\end{equation}
At \(\zeta=0\), \(a_r^{\mathrm{AI}}(0)=\ell_r/\rho\)
is positive and strictly increasing in \(r\). Since there are finitely
many scopes, both properties persist for small \(\zeta\geq0\),
so Lemma~\ref{lem:proof-ai-separation} applies. By hypothesis, the
fixed capacity \(A>0\) can be absorbed by either endpoint layer:
\(A\leq a_s^{\mathrm{AI}}(\zeta)g_s(q)\) for \(s\in\{1,n\}\)
throughout a neighborhood of zero. Both endpoint allocations below are
therefore feasible, with positive labor and understanding.

\emph{Step 1. Exhibit a welfare improvement on a global firm optimum.}
For \(\zeta>0\), \(\rho_n(\zeta)\) is uniquely largest, and
scope \(n\) can absorb all capacity. Thus the firm's unique optimal
capacity vector on the fixed hierarchy has \(A_n=A\), with all other
coordinates zero. By Lemma~\ref{lem:proof-ai-separation}, this
implementation also attains a global cost minimum over organizations.
For \(s\in\{1,n\}\), assigning all capacity to endpoint \(s\)
gives labor, understanding numerator, and welfare
\[
\begin{aligned}
C_s(\zeta)&=C^0-\rho_s(\zeta)A,
&
N_s(\zeta)&=N^0-\rho_s(\zeta)Ab_s,\\
W_s(\zeta)&=\frac{1}{C_s(\zeta)}-\log C_s(\zeta)
+\epsilon\log\!\left(\frac{N_s(\zeta)}{C_s(\zeta)}\right).
\end{aligned}
\]
At neutrality, both endpoints have \(C_s(0)=C^0-\rho A\), so the
material component of welfare and the denominator of understanding
coincide. Since \(b_1<b_n\),
\begin{equation}
W_1(0)-W_n(0)
=\epsilon\log\frac{N^0-\rho Ab_1}{N^0-\rho Ab_n}>0.
\label{eq:proof-ai-neutral-gap}
\end{equation}
The positive endpoint outcomes vary continuously with \(\zeta\), so
this strict ranking persists for sufficiently small \(\zeta>0\).
Scope-\(1\) automation on the same hierarchy therefore strictly
improves welfare over the displayed global firm optimum. This feasible
improvement establishes the reversal without solving the planner's
problem over organizations.

\emph{Step 2. Prove scope-\(1\) optimality on the fixed hierarchy.}
Suppose now \(b_1<b_r\) for all \(r\geq2\), and set
\(\Delta_b=\min_{2\leq r\leq n}(b_r-b_1)>0\). The comparison
must cover allocations assigning arbitrarily little capacity elsewhere,
so the endpoint ranking alone is insufficient. Write
\(A'=(A'_1,\ldots,A'_n)\) for an alternative capacity vector,
with feasible set
\[
\mathcal K(\zeta)=
\left\{A'\in\mathbb R_+^n:
A'_r\leq a_r^{\mathrm{AI}}(\zeta)g_r(q)\ \text{for every }r,
\quad \sum_r A'_r\leq A\right\}.
\]
Its outcomes satisfy
\[
C_H(A';\zeta)=C^0-\sum_r\rho_r(\zeta)A'_r,
\qquad
N_H(A';\zeta)=N^0-\sum_r\rho_r(\zeta)b_rA'_r.
\]
Welfare is continuous on this compact feasible set and attains a
maximum. The capacity-exhaustion argument in the proof of
Proposition~\ref{prop:ai-sorting} applies at each \(\zeta\), so
every maximum has \(\sum_rA'_r=A\).
For any such vector, let \(A_{>1}=\sum_{r=2}^nA'_r\) be capacity
assigned outside scope \(1\); then \(A'_1=A-A_{>1}\). Since
\(\delta_1=0\), comparison with scope-\(1\) automation gives
\[
\begin{aligned}
C_H(A';\zeta)-C_1(\zeta)
&=-\rho\zeta\sum_{r=2}^n\delta_rA'_r,
\\
N_H(A';\zeta)-N_1(\zeta)
&=-\rho\sum_{r=2}^n
   [b_r-b_1+\zeta\delta_rb_r]A'_r
\leq-\rho\Delta_b A_{>1}.
\end{aligned}
\]
Here \(0\leq\delta_r\leq1\) and \(b_r\geq0\). Thus the labor
gain is at most \(\rho\zeta A_{>1}\), while the numerator loss is
at least \(\rho\Delta_b A_{>1}\).

We next bound their welfare effects uniformly over feasible vectors.
Positivity and compactness at zero give positive minima of \(C_H\)
and \(N_H\). These remain bounded away from zero for small
\(\zeta\geq0\): otherwise a sequence of feasible vectors with
\(\zeta\downarrow0\) has a convergent subsequence whose limit is
feasible at zero, by continuity of the capacity bounds, and has
\(C_H=0\) or \(N_H=0\), a contradiction. Hence for some
\(\underline C,\underline N,\bar\zeta>0\), all outcomes with
\(0\leq\zeta\leq\bar\zeta\) lie in
\([\underline C,C^0]\times[\underline N,N^0]\).
This rectangle also contains the horizontal and vertical segments
joining any alternative outcome to the scope-\(1\) outcome.
On it, \(\partial W/\partial C=-1/C^2-(1+\epsilon)/C<0\)
and \(\partial W/\partial N=\epsilon/N>0\).
Therefore \(|\partial W/\partial C|\leq L_C\), where
\(L_C=\underline C^{-2}+(1+\epsilon)\underline C^{-1}\), and
\(\partial W/\partial N\geq\epsilon/N^0\).
Abbreviate the alternative outcome by \((C',N')\). Splitting the
welfare difference at \((C',N_1)\) and integrating these bounds yields
\[
\begin{aligned}
W(C',N')-W(C_1,N_1)
&=[W(C',N')-W(C',N_1)]
 +[W(C',N_1)-W(C_1,N_1)]\\
&\leq-\frac{\epsilon}{N^0}(N_1-N')+L_C(C_1-C')\\
&\leq\rho A_{>1}\left(L_C\zeta-
              \frac{\epsilon\Delta_b}{N^0}\right).
\end{aligned}
\]
For \(0<\zeta<\min\{\bar\zeta,\epsilon\Delta_b/(L_CN^0)\}\),
every allocation with \(A_{>1}>0\) therefore has strictly lower
welfare. A full-capacity vector with \(A_{>1}=0\) assigns all
capacity to scope \(1\). This is the planner's unique optimal capacity
vector on the fixed hierarchy with proportional displacement, whereas
the firm assigns all capacity to scope \(n\).
\hfill\rule{0.5em}{0.5em}\par\egroup

\end{document}


\setstretch{1.5}
\hypersetup{
  pdftitle={The Division of Understanding: Online Appendix},
  pdfauthor={Giampaolo Bonomi},
  linkcolor=LJETBlue, citecolor=LJETBlue,
  filecolor=LJETBlue, urlcolor=LJETBlue
}
\title{Online Appendix\\[3pt] to \textit{The Division of Understanding:\\
Specialization and Democratic Accountability}}
\author{Giampaolo Bonomi\thanks{Princeton University. Email: bonomi@princeton.edu}}
\date{September 2026}
\maketitle
\begin{center}
\begin{minipage}{0.82\textwidth}
\small
This appendix develops additional implications of the model, establishes
robustness to alternative production and political technologies, and reports
the empirical evidence. Proofs of the results stated in the paper appear in
its Appendices A--E. Detailed survey questions and variable inventories are
provided in the accompanying empirical documentation.
\end{minipage}
\end{center}
\makeatletter
\renewcommand*\l@section{\@dottedtocline{1}{1.5em}{3.5em}}
\renewcommand*\l@subsection{\@dottedtocline{2}{5.0em}{4.7em}}
\makeatother
\setcounter{tocdepth}{1}
\tableofcontents
\newpage
\appendix
\renewcommand{\thesection}{O\Alph{section}}
\renewcommand{\thesubsection}{\thesection.\arabic{subsection}}
\section{Production and policy understanding}
\label{processing-and-the-organization-of-production}
\label{app:a-production}

Main Appendix A proves the production and understanding results. Here I
examine their robustness and the political role of selection among labor
minimizers. Profiles lie in
\(\Delta_K\), \(r(p)\) is the number of positive coordinates of
\(p\), and \(1=\ell_1<\cdots<\ell_K\).

\subsection{Assignment conventions and selection}
\label{app:a-productive-selection}

The assignment conventions in Section 3 preserve labor and the occupational
distribution. A
positive-work responsibility with zero productive fit requires infinite
labor. A zero-work nonroot node can be suppressed, while a zero-work root
can be given a separately matched position at no labor cost. Hence fit
can be required to be positive everywhere. Disconnected responsibilities
assigned to one position can likewise be recorded as separate positions
with the same occupational profile. Splitting them into their connected
components preserves labor, employment, and the scale of external
handoffs.

Labor minimization fixes work and employment at each scope, but need not
fix expertise within a scope. Let production and the sole policy issue have profile
\(q=u=(1/2,1/2)\). Direct delegation to the two pure-domain components
assigns \(1/2\) unit of interface work to a matched coordinator with
profile \(q\). Alternatively, divide the root equally into
\(p^+=(0.6,0.4)\) and \(p^-=(0.4,0.6)\), then delegate each
intermediate responsibility to its components. The root now resolves
\(0.02\) units of work and the intermediate positions resolve
\(0.48\). Both organizations therefore require the same labor,
\(C=1+\ell_2/2\), and are labor minimizers.

Yet their workers understand the policy differently. The intermediate
profiles have fit \(\beta(p^+,q)=\beta(p^-,q)=0.8\), while the
single-domain specialists have zero fit. Since a matched scope-two
worker's understanding is \(\beta(p,q)/\ell_2\), aggregate
understanding is \(0.50/C\) under direct delegation and
\([0.02+0.48(0.8)]/C=0.404/C\) under the refinement. The firm is
indifferent although the electorate is better informed under direct
delegation. Other same-scope refinements need not reduce
employment-weighted understanding.

Write \(C^M\) for the minimum labor requirement and \(J(\mathcal A)\)
for the total scale of handoffs. Main Appendix A proves that a
minimum-handoff labor minimizer \(\mathcal A^M\) exists, with
\(J(\mathcal A^M)=r(q)-1\). For each \(\tau>0\), choose
a feasible \(\mathcal A_\tau\) within \(\tau^2\) of the infimum of
\(C+\tau J\). Comparison with \(\mathcal A^M\) gives
\[
\begin{aligned}
0\leq C(\mathcal A_\tau)-C^M
&\leq\tau J(\mathcal A^M)+\tau^2,\\
J(\mathcal A_\tau)&\leq J(\mathcal A^M)+\tau.
\end{aligned}
\]
Consequently, any feasible limit along which \(C\) and \(J\) are
continuous minimizes labor and then handoffs. Thus a vanishing handoff
cost delivers the selection when such a limit exists, without restricting
the planner's choices.

\subsection{Complementary knowledge beyond weakest-link processing}
\label{app:a-production-robustness}

The ordering of understanding across layers rests on concavity. In a
productively optimal hierarchy, draw a root-to-terminal path by choosing
each child according to its share of the parent's responsibility. Let
\(P_r\) be the occupational profile at scope \(r\) along that path.
Coherence gives \(P_r=\mathbb E[P_{r-1}\mid P_r]\), and main
Appendix A establishes that responsibility weights coincide with
employment weights within a layer. Therefore any concave
issue-understanding function \(\gamma(s,u)\) satisfies
\[
\mathbb E[\gamma(P_r,u)]\geq
\mathbb E[\gamma(P_{r-1},u)].
\]
The advantage of holding complementary knowledge together thus extends
beyond the minimum operator. It also holds across any two occupational
distributions related by a mean-preserving spread.

Production can be generalized at the same time. For each requirement
profile \(x\), let \(M_x\) aggregate the required-domain ratios
\((s_k/x_k)_{k:x_k>0}\). Assume it is increasing, concave, homogeneous
of degree one, and essential: it vanishes when any required-domain ratio
is zero. Normalize \(M_x(\mathbf1)=1\), and require
\[
M_x(y)\leq\sum_{k:x_k>0}x_ky_k,
\]
with equality only when the ratios are constant. Set
\(\beta_M(s,x)=M_x((s_k/x_k)_{k:x_k>0})\), use
\(\ell_{r(p)}/\beta_M(s,p)\) for productive labor intensity, and
\(\beta_M(s,u)/\ell_{r(u)}\) for issue understanding. Then
\(\beta_M(s,p)\leq\sum_{k:p_k>0}s_k\leq1\), with equality only at
\(s=p\). Exact matching and the resulting labor-minimizing hierarchy
are therefore unchanged. Concavity gives the same layer ordering.

Under uniform production and the symmetric dimensional agenda introduced
in Section 6, a matched scope-\(r\) worker covering an \(h\)-domain
issue has the same ratio \(h/r\) in every required domain. Homogeneity
gives fit \(h/r\); essentiality gives zero otherwise. Thus the formula
for \(b_r\) in main Appendix C is unchanged as well. Complementary
power means provide examples:
\[
\beta_\rho(s,x)=
\left[\sum_{k:x_k>0}x_k(s_k/x_k)^\rho\right]^{1/\rho},
\qquad \rho<0,
\]
extended continuously to zero when required knowledge is absent. The
geometric mean is the limit as \(\rho\to0\), and weakest-link fit is
the limit as \(\rho\to-\infty\). These extensions preserve the
production-to-understanding mechanism; the planner's assignment rule and
cutoff characterization use the stronger structure of weakest-link fit.

\subsection{Interface technologies beyond the quadratic specification}
\label{separable-concave-interface-resolution}

Any continuous concave interface function with \(d(e_k)=0\) preserves
nonnegative internal work and work conservation. A separable subclass also
preserves the efficient hierarchy itself. Suppose
\(d_\varphi(p)=\sum_k\varphi(p_k)\), where \(\varphi\) is
continuously differentiable and strictly concave on \([0,1]\), with
\(\varphi(0)=\varphi(1)=0\). The same function applies to every
domain. Let \(I=\operatorname{supp}q\), \(n=|I|\), and define
\[
V_r^\varphi(q)=\max\{
\mathbb E[d_\varphi(P)]:\mathbb E[P]=q,\ r(P)\leq r
\text{ almost surely}\},\quad 1\leq r\leq n.
\]

\begin{lemma}[Robustness of the efficient hierarchy]
\label{lem:app-a-separable-hierarchy}
The frontier-maximizing profiles and their nested deletion construction
are the same as in main Appendix A. The resulting aggregate workloads
and minimum labor requirement are
\[
g_1^\varphi=1,\qquad
g_r^\varphi=V_r^\varphi(q)-V_{r-1}^\varphi(q)\ (r\geq2),
\qquad C_\varphi^M(q)=\sum_{r=1}^n\ell_rg_r^\varphi.
\]
Minimum-handoff labor minimizers have adjacent-scope handoffs, exact
matching, and the same layer ordering of understanding as in Theorem 2.
\end{lemma}

\begin{proof}
Write \(\pi_i=\Pr(P_i>0)\). Conditional Jensen gives
\(\mathbb E[\varphi(P_i)]\leq\pi_i\varphi(q_i/\pi_i)\), with
equality only if positive coordinate values are constant. Maximize the
sum of these bounds over \(q_i\leq\pi_i\leq1\) and
\(\sum_i\pi_i\leq r\). Its derivative in \(\pi_i\) is
\(\Psi(q_i/\pi_i)\), where
\(\Psi(t)=\varphi(t)-t\varphi'(t)\) is positive and strictly
increasing for \(t>0\). The inclusion budget therefore binds, and the
first-order conditions give
\[
\pi_i^r=\min\{1,q_i/\theta_r\},\qquad
\sum_{i\in I}\pi_i^r=r,\qquad
w_i^r=q_i/\pi_i^r=\max\{q_i,\theta_r\}.
\]
As in main Appendix A, set \(\theta_1=1\) and
\(\theta_n=\min_{i\in I}q_i\). The same finite distributions over
supports attain these bounds, and the same conditional deletion rule
makes them coherent across scopes. Every frontier profile includes
\(F_r=\{i:q_i\geq\theta_r\}\) and has common interface load
\[
V_r^\varphi(q)=\sum_{i\in F_r}\varphi(q_i)
+(r-|F_r|)\varphi(\theta_r).
\]
The cut bounds and labor-cost decomposition in main Appendix A now
apply with \(V_r^\varphi\) in place of \(1-\Phi_r\), proving
optimality and the workload formula. Strict concavity supplies the same
equality conditions and zero-work pruning argument. Minimum handoffs
therefore select adjacent scopes. Common interface loads within each
layer imply \(z_v=\omega_vg_r^\varphi\), so conditional employment
weights are again responsibility weights. Conditional Jensen then gives
the understanding ordering.
\end{proof}

Changing \(\varphi\) changes the amount of coordination work and hence
employment, average understanding, and the planner's and AI comparisons;
the unchanged object is the efficient division of responsibilities.
Domain-specific or nonseparable interface technologies need not preserve
this division. Nor does the extension alone preserve the baseline
employment ordering or the planner's cutoff.

\section{Electoral accountability}
\label{political-competition-and-welfare}

This section extends the electoral game in Section 5. Maintain
\(b_i\in[0,1]\),
\(B=\int_0^1b_i\,di\), and the baseline effort-cost assumptions:
\(c\in C^1([0,\infty))\cap C^2((0,\infty))\),
\(c(0)=c'(0)=0\), and \(c''(e)>0\) for \(e>0\).

\subsection{From votes to the probability of office}
\label{timing-and-corner-conventions}

Let \(V\in[0,1]\) be candidate 1's realized vote share before a
common popularity shock \(\xi\), independent of evaluation errors and
uniform on \([-1/2,1/2]\). She wins when \(V+\xi\geq1/2\), so
her conditional winning probability is \(V\). Averaging gives
\(\Pr(\text{win})=\mathbb E[V]=\int_0^1\pi_{i1}\,di\).
Normalizing office value to one yields the candidate objective, even with
dependent evaluation errors. Other uniform ranges covering all feasible
margins give an affine function of expected support, whose slope can be
absorbed into the effort cost.

\subsection{General policy technologies}
\label{general-policy-technologies}

Suppose output and effort deliver policy quality \(G(Y,e)\). Voters
compare its logarithm across the two candidates, subject to the same
logistic evaluation errors. The technology can be nonfiscal and can make
effort incentives depend on resources. For fixed
\(Y>0\), define
\[
\Lambda(e,Y)=\frac{c'(e)}{\partial_e\log G(Y,e)}
=c'(e)\frac{G(Y,e)}{G_e(Y,e)}.
\]

\begin{lemma}[General policy technology]
\label{lem:app-b-general-policy-technology}
Suppose \(G(Y,\cdot)\) is continuous on \([0,\infty)\), twice
continuously differentiable on \((0,\infty)\), and satisfies
\(G(Y,0)=0\), \(G_e>0\), and \(G_{ee}\leq0\). Assume
\(\Lambda(e,Y)\to0\) as \(e\downarrow0\) and
\(\Lambda(e,Y)\to\infty\) as \(e\to\infty\). For \(B>0\),
the unique pure-strategy equilibrium is symmetric and positive, with
\[
\Lambda(e^*,Y)=B/4.
\]
Effort increases strictly with \(B\). If \(G\) is twice continuously
differentiable in both arguments, the sign of
\(\partial e^*/\partial Y\), holding \(B\) fixed, equals the sign
of \(\partial_{eY}^2\log G(Y,e^*)\).
\end{lemma}

\begin{proof}
Against positive opponent quality \(\bar G\), a voter's support is
\(G(Y,e)^b/[G(Y,e)^b+\bar G^b]\). For \(b\in(0,1]\), this is
the composition of an increasing concave function of quality with the
concave technology \(G\). Expected support is therefore concave in
effort, and the candidate's objective is strictly concave. The boundary
arguments in main Appendix B apply: a sole positive-effort candidate can
reduce her effort without losing support; against positive opponent
effort, concavity bounds the electoral gain from below by a positive
multiple of small own effort, whereas \(c(e)=o(e)\). Both-zero effort
is defeated by an arbitrarily small positive effort when \(B>0\).

At an interior equilibrium, the two first-order conditions give
\[
\begin{aligned}
\Lambda(e_j,Y)&=\int_0^1b_i\pi_{ij}(1-\pi_{ij})\,di
=\Lambda(e_{-j},Y),\\
\Lambda_e&=c''\frac{G}{G_e}
+c'\left(1-\frac{GG_{ee}}{G_e^2}\right)>0.
\end{aligned}
\]
Thus equilibrium is symmetric. The limits of \(\Lambda\) give a
unique solution to \(\Lambda=B/4\), and strict concavity makes it a
global best response. Finally, with \(g=\log G\),
\(\Lambda_Y=-c'g_{eY}/g_e^2\). Implicit differentiation proves both
comparative statics.
\end{proof}

Better understanding still strengthens accountability. Resources also
raise effort if they increase the marginal return to design in
log-quality units, and reduce it if they lower that return. With a
multiplicatively separable technology \(G(Y,e)=Q(Y)H(e)\), this
interaction vanishes; the baseline \(Ye\) has precisely this property.
The result generalizes the electoral mechanism, without asserting that
the later planner and AI formulas remain unchanged.

\subsection{Alternative voter responses}
\label{alternative-voter-responses}

The sufficient statistic \(B\) also survives beyond logistic errors.
Return to the baseline policy technology and let a voter's support for
candidate \(j\) be
\(\pi_{ij}=F(b_i\log(e_j/e_{-j}))\). Assume \(F\) is twice
continuously differentiable and strictly increasing, satisfies
\(F(-z)=1-F(z)\), has \(F'(0)>0\), and tends to zero and one at
the respective tails. Use the zero-effort conventions in main Appendix B.

\begin{lemma}[Alternative response rule]
\label{lem:app-b-alternative-response}
If \(bF''(z)\leq F'(z)\) for all \(b\in(0,1]\) and
\(z\in\mathbb R\), then for \(B>0\) the unique pure-strategy
equilibrium is symmetric and positive, with
\[
e^*c'(e^*)=F'(0)B.
\]
\end{lemma}

\begin{proof}
Fix opponent effort \(\bar e>0\) and set
\(z=b\log(e/\bar e)\). The second derivative of support is
\(b[bF''(z)-F'(z)]/e^2\leq0\), so each candidate's objective is
strictly concave. The same boundary arguments as above exclude zero
effort. At a positive equilibrium,
\[
e_jc'(e_j)=\int_0^1b_iF'(z_i)\,di,
\qquad
e_{-j}c'(e_{-j})=\int_0^1b_iF'(-z_i)\,di.
\]
Symmetry makes \(F'\) even, so the right-hand sides coincide.
Strict increase of \(ec'(e)\) forces equal efforts, and the
first-order condition at a tie gives the displayed equation. Its
left-hand side rises continuously from zero to infinity; strict
concavity makes its unique solution a global best response.
\end{proof}

Under the curvature restriction, mean understanding remains sufficient
for equilibrium effort: the error distribution changes the strength of
incentives through \(F'(0)\). Logistic errors satisfy the restriction
and have \(F'(0)=1/4\). In both extensions,
\(B=0\) makes support independent of effort and the unique equilibrium
has zero effort.

\section{Common-interest planning and competition}
\label{common-interest-planning-and-competition}

The main appendix proves competitive implementation and the global
delegation-cutoff result. I first show that productive efficiency survives
when sizeable firms maximize their workers' welfare and internalize political
outcomes; the remaining subsections develop the planner's comparative statics.

\subsection{Sizeable firms and free entry}
\label{oa:workforce-entry}

Suppose firms are owned by their workers, may choose their scale freely, and
return all profits to their workforce as dividends. For an organization
\(\mathcal A\), let \(Y_{\mathcal A}=1/C(\mathcal A)\), let
\(\lambda_{\mathcal A}\) be its employment distribution over profiles, and
let \(B_{\mathcal A}=\int b(s)\,d\lambda_{\mathcal A}(s)\). Firm \(j\)
employs mass \(m_j>0\) under \(\mathcal A_j\). At the common wage schedule
\(w(s)\), it pays a nonnegative dividend \(D_j(s)\) and distributes all
revenue to its workers:

\[
\int_{\Delta_K}\bigl[w(s)+D_j(s)\bigr]\,
d\lambda_{\mathcal A_j}(s)=Y_{\mathcal A_j}.
\]

Thus its workers' average private income is \(Y_{\mathcal A_j}\). With a unit
workforce, \(\sum_jm_j=1\),
\(Y=\sum_jm_jY_{\mathcal A_j}\), and
\(B=\sum_jm_jB_{\mathcal A_j}\). The firm maximizes
\(Y_{\mathcal A_j}+\log Y+\epsilon\log B\), accounting for the
economy-wide outcomes and profile prices induced by its choices. Full payout
makes prices irrelevant for average workforce income: they divide revenue
between wages and dividends but do not change its total.

An individual worker takes political outcomes as given and chooses the firm
and profile offering the highest private income, \(w(s)+D_j(s)\). Free
mobility and market clearing therefore give a common private income \(\rho\)
at every occupied job; because any profile can be supplied at its quoted wage,
\(w(s)\leq\rho\) for every \(s\).

Free entry allows any positive-mass group of workers to form a firm. An
equilibrium is \emph{entry-stable} if no such firm can attract its members at
weakly higher private incomes and make them strictly better off after the
resulting market-clearing reallocation.
\begin{proposition}[Free entry and productive efficiency]
\label{prop:oa-workforce-entry}
In any entry-stable equilibrium with \(Y,B>0\), only productively efficient
organizations are implementable: every active firm \(j\) satisfies
\(C(\mathcal A_j)=C^M\). Consequently, aggregate output and every worker's
private income equal \(Y^M=1/C^M\).
\end{proposition}

\begin{proof}
Full profit distribution and the common private income imply
\(Y_{\mathcal A_j}=\rho\) at every active firm, and hence \(Y=\rho\). Suppose
some active firm \(j\) uses a labor-inefficient organization. Then
\(\delta\equiv Y^M-\rho>0\).

Choose a labor-minimizing organization \(\mathcal A^M\) and a mass
\(x\in(0,m_j)\) of firm \(j\)'s workers. Take the same fraction from every
position, so constant returns allow the incumbent to contract proportionally,
and let these workers choose the profiles required by \(\mathcal A^M\). The
entrant pays \(D^M(s)=Y^M-w(s)>0\). These dividends exhaust its profits and
give every recruit private income \(Y^M\). Because the departing group is a
scaled copy of the incumbent's workforce and workers may freely change
profiles, the reallocation is feasible and profile markets continue to clear.
Its aggregate consequences are
\[
Y_x=Y+x(Y^M-\rho),\qquad
B_x=B+x\bigl(B_{\mathcal A^M}-B_{\mathcal A_j}\bigr).
\]
Each recruit's utility gain is therefore
\[
\delta+\log\frac{Y_x}{Y}
+\epsilon\log\frac{B_x}{B}
\longrightarrow\delta>0
\qquad\text{as }x\downarrow0.
\]
For sufficiently small \(x\), the entrant makes all its workers better off,
contradicting entry stability. The remaining conclusions follow immediately.
\end{proof}

The productivity gain remains positive for each recruit, while a vanishingly
small firm's political effect disappears. The argument therefore extends to
\(u_i=y_i+H(Y,\mu)\), where \(\mu\) is the social distribution of
understanding and the common term \(H\) is continuous under a vanishing-mass
reallocation.

\subsection{The general frontier and existence}
\label{the-general-frontier-relaxed-existence-and-the-organizational-change-hurdle}

For a finite organization \(\mathcal A\in\Omega(q)\), its labor and
labor-weighted understanding requirements are
\[
C_{\mathcal A}=\sum_vz_va(p_v,s_v),
\qquad N_{\mathcal A}=\sum_vz_va(p_v,s_v)b(s_v).
\]
Its output and understanding are
\(Y_{\mathcal A}=1/C_{\mathcal A}\) and
\(B_{\mathcal A}=N_{\mathcal A}/C_{\mathcal A}\).
The cost frontier in the paper therefore gives the finite-organization
problem the value
\begin{equation}
\sup_{\bar B>0:\,\mathcal C_q(\bar B)<\infty}
\left\{\frac1{\mathcal C_q(\bar B)}
-\log\mathcal C_q(\bar B)+\epsilon\log\bar B\right\}.
\label{eq:appC-reduced-planner}
\end{equation}
Indeed, every organization supplies a feasible frontier target, and a
sequence approaching the cost infimum at any target approaches at least
the displayed welfare. Neither direction requires attainment.

Finite labor-share combinations are already feasible organizations by
the parallel-branch construction in main Appendix D. For
existence, it therefore remains only to take limits. Define
\begin{equation}
\mathcal O_q=\overline{
\{(Y_{\mathcal A},B_{\mathcal A}):\mathcal A\in\Omega(q)\}}.
\label{eq:appC-outcome-set}
\end{equation}
The set inside the closure is convex. Since
\(\mathcal O_q\subseteq[0,1]^2\), it is compact. Assigning full-support
knowledge supplies a strictly positive outcome, while
\(W_0(Y,B)=Y+\log Y+\epsilon\log B\) tends to minus infinity at either
zero coordinate. Strict concavity thus gives a unique relaxed optimum.
This is uniqueness of the outcome, not of the organization; the closure
can also contain outcomes that no finite organization attains. The
example below shows why that distinction matters. Under the assumptions
of the paper's cutoff theorem, its main-appendix proof instead supplies
finite attainment by one organization with at most two adjacent-cutoff
branches.

\subsection{Policy complexity and the delegation cutoff}
\label{the-delegation-frontier-and-comparative-statics}

Maintain the cutoff theorem's uniform production profile, symmetric
dimensional agenda, and affine labor intensities.
Let \(g_r\) be the efficient workloads, \(b_r\) the understanding of a
matched scope-\(r\) worker, and \(n_r=\ell_rb_r\). At a complete cutoff
\(j\in\{0,\ldots,K-1\}\), the main appendix establishes
\begin{equation}
(C_j,N_j)=\sum_{r\leq j}g_r(\ell_r,n_r)
+\sum_{r>j}g_r(\ell_K,n_K).
\label{eq:appC-cutoff-outcomes}
\end{equation}
The frontier runs from the labor-minimizing cutoff \(j=K-1\) toward
\(j=0\), with adjacent cutoffs joined by partial reassignment of one
source workload. On segment \(j\), its slope and average understanding are
\begin{equation}
\sigma_j=\frac{n_K-n_j}{\ell_K-\ell_j},
\qquad N=\sigma_jC-D_j,
\qquad B=\frac NC,
\label{eq:appC-sigma}
\end{equation}
where \(D_j\geq0\) is constant along the segment. The main appendix
proves \(\sigma_1\leq\cdots\leq\sigma_{K-1}\) and
\(\sigma_j\geq b_K\geq B\): broad source workloads offer the largest
gain in understanding per additional unit of labor.

To vary the political need for breadth, let
\(\nu_\chi=(1-\chi)\nu^{\mathrm{dom}}+\chi\nu^{\mathrm{cross}}\),
where \(\nu^{\mathrm{dom}}\) weights the \(K\) domain-specific issues
equally and \(\nu^{\mathrm{cross}}\) is a fixed dimensional agenda of
cross-domain issues. Every profile has understanding \(b_0=1/K\) under
\(\nu^{\mathrm{dom}}\).

\begin{proposition}[Policy complexity and broader responsibility]
\label{prop:appC-cutoff-comparative-statics}
The least and greatest amounts of system-wide work among optimal cutoff
allocations weakly increase with either \(\epsilon\) or \(\chi\). The
increase is strict within a unique interior partial-cutoff segment. At
\(\chi=0\), labor minimization is socially optimal. The planner's
maximized welfare advantage over the labor-minimizing cutoff also
weakly increases with \(\chi\), strictly when the optimum assigns
additional work to the root and thereby strictly raises understanding
under the cross-domain agenda.
\end{proposition}

\noindent\textbf{Proof.}
Along a segment with \(D_j>0\), differentiation of
\(W_0(1/C,N/C)\) shows that increasing \(C\) raises welfare exactly when
\begin{equation}
\epsilon>\Theta_j(C):=
\frac{B(1+1/C)}{\sigma_j-B},
\qquad
\Theta_j'(C)=\frac{\sigma_j}{D_j}+\frac1{C^2}>0.
\label{eq:appC-threshold-monotone}
\end{equation}
If \(D_j=0\), understanding is constant and welfare strictly falls with
\(C\); set \(\Theta_j=+\infty\). At a vertex, the next segment has a
weakly smaller \(\sigma_j\), so its hurdle weakly rises. Welfare is
therefore single-peaked along the entire frontier. Raising
\(\epsilon\) moves its maximizing labor requirement weakly rightward,
including at vertices and when adjacent slopes tie. The amount of
system-wide work strictly increases along this same path.

Write \(B_{\mathrm{cross}}(C)\) and \(\sigma_{j,\mathrm{cross}}\) for
understanding and the slope under the cross-domain agenda. Linearity in
agenda weights gives
\(B_\chi(C)=b_0+\chi[B_{\mathrm{cross}}(C)-b_0]\) and
\(\sigma_{j,\chi}=b_0+\chi(\sigma_{j,\mathrm{cross}}-b_0)\).
Consequently, whenever a segment offers a strict understanding gain,
\begin{equation}
\Theta_{j,\chi}(C)=
\left(1+\frac1C\right)
\frac{b_0+\chi[B_{\mathrm{cross}}(C)-b_0]}
{\chi[\sigma_{j,\mathrm{cross}}-B_{\mathrm{cross}}(C)]}
\label{eq:appC-cross-domain-hurdle}
\end{equation}
strictly decreases with \(\chi>0\). Thus the maximizing point moves
weakly rightward with \(\chi\), including across vertices. On an
interior partial-cutoff segment, the strict derivatives in
\eqref{eq:appC-threshold-monotone} and
\eqref{eq:appC-cross-domain-hurdle} give strict comparative statics. At
\(\chi=0\), understanding is constant across all organizations, so
only output matters.

Let \(M\) be the labor-minimizing cutoff. For any fixed cutoff point,
\begin{equation}
\frac{\partial}{\partial\chi}
\log\frac{B_\chi(C)}{B_\chi(M)}
=\frac{b_0[B_{\mathrm{cross}}(C)-B_{\mathrm{cross}}(M)]}
{B_\chi(C)B_\chi(M)}\geq0.
\label{eq:appC-welfare-gap-derivative}
\end{equation}
Its production cost is independent of \(\chi\). Maximizing these
weakly increasing welfare differences proves the welfare-gap claim;
holding an optimizer fixed proves strictness whenever the numerator is
positive.\hfill\(\square\)

Greater policy interdependence therefore strengthens the case for
preserving system-wide responsibilities and increases the cost of the
market's specialization. This comparison concerns the organizational
distortion: welfare itself can fall as issues become harder to evaluate.

\subsection{Uneven knowledge requirements within an issue}
\label{policy-issues-with-uneven-knowledge-requirements}

The cutoff result relies on equally weighted requirements within each
issue, not just symmetry across domain names. For example, let \(K=2\),
\(\ell_2=6/5\), and give equal probability to
\(u^1=(9/10,1/10)\) and \(u^2=(1/10,9/10)\). A worker with
\(s=(9/10,1/10)\) performing the specialized responsibility \(p=e_1\)
has
\begin{equation}
a(e_1,s)=\frac{10}{9},
\qquad b(s)=\frac1{2\ell_2}\left(1+\frac19\right)=\frac{25}{54}.
\label{eq:appC-uneven-specialist}
\end{equation}
A matched uniform root worker has the same understanding but costs
\(6/5>10/9\) per unit of work. Thus partly broadening a specialist's
knowledge can provide the same political understanding at lower
processing cost than placing work under system-wide responsibility.
The general frontier accommodates this possibility; the exact-matching
cutoff need not. The coordinator's understanding advantage under a
coherent split still follows from concavity of \(b(s,u)\), without
equal weights within issues.

\subsection{Why a finite optimum need not exist}
\label{two-domain-solution-and-finite-nonattainment}

Finite attainment can fail because successive refinements change the
profiles of coordinating jobs without narrowing their scope. The next
example establishes this directly, without a general recursion.

\begin{proposition}[A finite optimum need not exist]
\label{prop:appC-nonattainment}
Let \(q=(1/4,3/4)\), \(\ell_2=6/5\), and let the agenda weight the two
domain-specific issues by \(1/4\) and \(3/4\). For
\(\epsilon=299/464\), the unique relaxed optimum is
\[
Y^*=\frac{20}{29},\qquad B^*=\frac{299}{464},
\]
but no finite organization attains it.
\end{proposition}

\noindent\textbf{Proof.}
Index both responsibility and knowledge profiles by their first
coordinate. Then \(b(s)=3/4-s/2\). Set \(H=1083/464\). The least cost
\(a(p,s)[H-b(s)]\) of one unit of work is attained at exact matching:
\([H-b(s)]/s\) decreases in \(s\), whereas
\([H-b(s)]/(1-s)\) increases. Its minimized value is
\begin{equation}
k(0)=\frac{735}{464},\quad k(1)=\frac{967}{464},\quad
k(x)=a+bx\ (0<x<1),\qquad
a=\frac{441}{232},\quad b=\frac35.
\label{eq:appC-nonattainment-interior}
\end{equation}
Define the following lower bound on the cost of processing a complete
responsibility:
\[
V(x)=(1-x)k(0)+xk(1)+2ax(1-x)+bx^2(1-x).
\]
Writing \(d(x)=2x(1-x)\), direct expansion gives, for every child law
with \(\mathbb E X'=x\),
\begin{equation}
k(x)[d(x)-\mathbb E d(X')]+\mathbb E V(X')-V(x)
=b\,\mathbb E[(X'-x)^2(1-X')]\geq0.
\label{eq:appC-certificate-identity}
\end{equation}
For interior \(x\), equality requires a degenerate split: the only
possible contacts are \(x\) and \(1\), whose mean can equal \(x<1\)
only if all mass stays at \(x\). Stopping at an interior responsibility
also has strictly positive slack,
\[
[1+d(x)]k(x)-V(x)=\frac{147}{464}+\frac{x}{10}
+\frac35x^2(1-x)>0.
\]
Summing these inequalities over a finite tree rooted at \(1/4\) yields
\(HC-N>V(1/4)=49/20\): after any degenerate divisions, the tree must
either stop or make a nondegenerate split. This bound also covers
unmatched assignments, whose unit costs are weakly higher. Therefore
every finite outcome satisfies \((49/20)Y+B<H\).

To approach the bound, take a chain with interior profiles
\(x_j=j/(4n)\), \(j=n,\ldots,1\). Each node divides between
\(x_{j-1}\) and \(1\), with the unique weights preserving its mean;
\(x_0=0\). Its scale is \((3/4)/(1-x_j)\), so each interior node
performs \(3/(8n)\) units of interface work. Exact matching and terminal
scales \(1/4,3/4\) give
\begin{equation}
C_n=\frac{29}{20},\qquad
N_n=\frac{299-9/n}{320},\qquad
B_n=\frac{299-9/n}{464}.
\label{eq:appC-finite-outcomes}
\end{equation}
Their limit is the stated outcome and attains the supporting line.
There the welfare gradient is
\((1+1/Y^*,\epsilon/B^*)=(49/20,1)\). Concavity therefore makes it a
relaxed optimum; strict concavity makes it unique, and the strict
finite-outcome bound proves nonattainment.\hfill\(\square\)

The example keeps exact matching, positive understanding, and an affine
labor-intensity schedule. Its failure of attainment comes entirely from
costless refinements within the same scope. It explains the closure in
\eqref{eq:appC-outcome-set}; it does not arise in the finite cutoff
implementation proved in the main appendix.

\section{Targeted spending and the value of knowledge}
\label{targeted-spending-and-institutions}
\label{app-e-targetable-services}

This section supplies the remaining results used in Section 7 and an
example of optimal cross-training.

\subsection{Effort and the allocation of public services}
\label{app-e-targetable-allocation}

Maintain the groups and preferences of Section 7: group \(g\) has mass
\(m_g>0\), understanding \(b_g\in(0,1]\), and authority per person
\(v_g\geq0\), with \(\sum_gm_g=\sum_gm_gv_g=1\). Write
\(B_v=\sum_gm_gv_gb_g\). A candidate chooses effort \(e_j\) and
targeted services \(t_{gj}\), subject to
\(\sum_gm_gt_{gj}=fYe_j\).

\begin{lemma}[Effort and targeted provision]
\label{lem:app-e-targetable-allocation}
The candidates' game has a unique pure-strategy equilibrium. It is
symmetric, and its common effort satisfies
\begin{equation}
ec'(e)=B_v/4.
\label{eq:app-e-policy-quality-general}
\end{equation}
Under the isoelastic effort cost, \(e=B_v^\epsilon\) and
\(R=YB_v^\epsilon\). For \(f>0\), equilibrium provision is
\begin{equation}
t_g=fR\frac{v_gb_g}{B_v}.
\label{eq:app-e-targeted-spending-share}
\end{equation}
At \(f=0\), all provision benefits every group equally.
\end{lemma}

\begin{proof}
At \(f=0\), the proof of the paper's Proposition 1 applies with electoral
weights \(m_gv_g\). For \(f>0\), define
\[
Q_{gj}=(Ye_j)^{1-f}(t_{gj}/f)^f,
\qquad
\pi_{gj}=\frac{Q_{gj}^{b_g}}{Q_{gj}^{b_g}+Q_{g,-j}^{b_g}}.
\]
If both effective-service quantities are zero, set support to one half;
if only one is positive, that candidate receives the group's support.
At \(f=1\), define \(Q_{gj}=t_{gj}\). The candidate maximizes \(\sum_gm_gv_g\pi_{gj}-c(e_j)\) subject to
her budget. Against a positive opposing strategy, \(Q_{gj}^{b_g}\)
is concave in \((e_j,t_{gj})\): its Cobb--Douglas exponents sum to
\(b_g\leq1\). Since \(z/(z+a)\) is increasing and strictly concave
for \(a>0\), and \(c\) is strictly convex, the problem is strictly
concave in the politically relevant choices.

Zero effort cannot be an equilibrium choice: when both candidates exert
none, a small positive choice wins informed voters; against positive
effort, the electoral gain dominates the cost near zero because
\(b_g\leq1\) and \(c'(0)=0\). Every positive-authority group must
receive positive provision from both candidates. If both allocations
were zero, a small positive allocation would yield a discrete gain in
support; if only one were positive, that candidate could reduce it
without losing the group's support. Zero-authority groups receive none.

Let \(\mu_j\) be the budget multiplier and set
\(S_j=\sum_gm_gv_gb_g\pi_{gj}(1-\pi_{gj})\). The first-order
conditions are
\[
\frac{1-f}{e_j}S_j-c'(e_j)+\mu_jfY=0,
\qquad
\frac{v_gb_gf}{t_{gj}}\pi_{gj}(1-\pi_{gj})=\mu_j.
\]
Multiplying the allocation conditions by \(m_gt_{gj}\), summing, and
using the budget gives \(\mu_jYe_j=S_j\), hence
\(e_jc'(e_j)=S_j\). Because \(\pi_{g,-j}=1-\pi_{gj}\), both
candidates have the same \(S_j\) and therefore the same effort and
multiplier. Their allocation conditions then give identical provision
in every group. Thus \(\pi_{gj}=1/2\) and \(S_j=B_v/4\), yielding
the stated formulas. They satisfy the first-order conditions of the
strictly concave problems, proving existence and uniqueness.
\end{proof}

Thus \(B_v\) governs the quality of policy, while \(v_gb_g\) governs
per-capita targeted provision. Under equal voting, groups with greater
understanding receive more services.

\subsection{The social return to more equal understanding}
\label{app-e-social-communication}
\label{app-e-dispersion-comparison}

Under equal voting, let \(B=\sum_gm_gb_g\) and
\(\Delta=\log B-\sum_gm_g\log b_g\). The paper's equation (23)
gives welfare as \(W_f=W_0-f\Delta\), where
\(W_0=Y+\log Y+\epsilon\log B\).

\begin{lemma}[Social value of communication]
\label{lem:app-e-social-communication}
Holding output and group masses fixed, the marginal social value per
member of raising group \(g\)'s understanding is
\begin{equation}
\frac{1}{m_g}\frac{\partial W_f}{\partial b_g}
=\frac{\epsilon-f}{B}+\frac{f}{b_g}>0
\quad\text{whenever }b_g\leq B.
\label{eq:app-e-social-communication}
\end{equation}
More generally, a change that raises \(B\) and weakly lowers
\(\Delta\) strictly raises welfare at fixed output.
\end{lemma}

\begin{proof}
Differentiate \(W_f=Y+\log Y+(\epsilon-f)\log B+
f\sum_gm_g\log b_g\). For \(b_g\leq B\), the derivative per member
is at least \(\epsilon/B>0\). For a finite change, subtracting welfare
before and after gives
\(W_f'-W_f=\epsilon\log(B'/B)-f(\Delta'-\Delta)\).
\end{proof}

When \(f>\epsilon\), an informed source loses from raising others'
understanding because its share of targeted spending falls. Yet helping
a below-average group has a positive marginal social return: welfare
also counts the recipients' gains. Thus public provision of analysis or
support for its dissemination can have value even when credibility is
assured. The same concern for distribution changes the planner's choice
of occupational knowledge.

\begin{proposition}[Targeted spending and the dispersion of understanding]
\label{prop:app-e-targetability-dispersion}
Hold the feasible set of organizations fixed, with positive
understanding in every positive-mass group. If optima exist at
\(0\leq f<f'\leq1\), every optimum \(\mathcal A'\) at \(f'\) and
every optimum \(\mathcal A\) at \(f\) satisfy
\[
\Delta(\mathcal A')\leq\Delta(\mathcal A),
\qquad W_0(\mathcal A')\leq W_0(\mathcal A).
\]
At fixed output and mean understanding, a feasible nontrivial
mean-preserving equalization strictly raises welfare for \(f>0\).
\end{proposition}

\begin{proof}
Optimality gives
\[
\begin{aligned}
W_0(\mathcal A)-f\Delta(\mathcal A)
&\geq W_0(\mathcal A')-f\Delta(\mathcal A'),\\
W_0(\mathcal A')-f'\Delta(\mathcal A')
&\geq W_0(\mathcal A)-f'\Delta(\mathcal A).
\end{aligned}
\]
Adding yields the dispersion ordering; substituting it into the first
inequality yields the ordering of \(W_0\). A nontrivial
mean-preserving equalization raises mean log understanding by strict
concavity of the logarithm and thus lowers \(\Delta\).
\end{proof}

As more provision becomes targetable, the planner accepts a lower
common-interest payoff to reduce inequality in understanding. This
comparison needs neither the symmetric production assumptions used
below nor a unique or differentiable optimum.

\subsection{An optimum that cross-trains specialists}
\label{app-e-targetable-cutoff-failure}

Theorem 4 permits workers to acquire knowledge beyond their productive
responsibilities. The following example shows that cross-training can
be necessary for optimality among all feasible organizations.

\begin{example}[Optimal cross-training]
\label{ex:app-e-certified-cutoff-failure}
Take \(K=3\), \(q=(1/3,1/3,1/3)\), and a symmetric dimensional
agenda with
\[
(\nu_1,\nu_2,\nu_3)=(.15,.05,.80),
\qquad (\kappa,\epsilon,f)=(.8,.30,.25).
\]
Then \((\ell_1,\ell_2,\ell_3)=(1,1.8,2.6)\), the canonical
workloads are \((g_1,g_2,g_3)=(1,1/2,1/6)\), and matched
understanding is
\((b_1,b_2,b_3)=(1/20,8/135,2641/7020)\).
An optimal hierarchy uses all three scopes. Its knowledge shares in the
assigned domains, \(\alpha_t\), and resulting understanding are
\[
\begin{array}{c|ccc}
\text{Responsibility scope }t&1&2&3\\ \hline
\alpha_t&.939503518214&.864766636259&1\\
\bar b_t&.079601910959&.187846538884&.376210826211
\end{array}
\]
Both nonroot groups are strictly cross-trained. Scope-two workers devote
more knowledge outside their jobs than scope-one workers, although their
understanding remains higher. Thus an optimal organization can retain
specialized responsibilities while broadening the knowledge used to
evaluate policy; the amount of broadening need not be ordered by scope.
\end{example}

To verify global optimality, write
\(G=\sum_tg_t=5/3\), \(S=\sum_tm_t\log\bar b_t\), and
\(h(b)=\lambda b+f\log b\). For \(t<3\), the two-level profiles in
the proof of Theorem 4 give
\[
\bar b_t(\alpha)=b_3-d_t(\alpha-t/3),
\qquad d_t=\frac{3(b_3-b_t)}{3-t},
\qquad t/3\leq\alpha\leq1.
\]
For any \(\eta>0\), \(\lambda\geq0\), and \(H>h(b_3)\), define
\begin{equation}
\begin{aligned}
r_t(H)&=\frac{\eta}{G}
+\max_{\alpha\in[t/3,1]}\frac{\ell_t}{\alpha}
   [h(\bar b_t(\alpha))-H] &&(t<3),\\
r_3(H)&=\frac{\eta}{G}+\ell_3[h(b_3)-H].
\end{aligned}
\label{eq:app-e-scope-return}
\end{equation}
Appendix D's workload argument bounds the support residual of every
organization by \(\sum_{j=1}^3g_j\max_{t\geq j}r_t(H)\), and a
pruned hierarchy attains that bound. Hence a zero residual certifies
that \(H\) is the global maximum of \(\eta Y+\lambda B+fS\).

For reproduction, put \(\eta=1+\tau\), evaluate the scope returns at
\(\alpha_t\), and solve for \((\alpha_1,\alpha_2,\tau,\lambda,H)\):
\begin{equation}
\begin{aligned}
H-h(\bar b_t)&=\alpha_td_t[\lambda+f/\bar b_t] &&(t=1,2),\\
\tau&=C,\qquad \lambda B=\epsilon-f,
&r_1+\tfrac12r_2+\tfrac16r_3&=0,
\end{aligned}
\label{eq:app-e-certified-residual}
\end{equation}
where \(C=\sum_tg_t\ell_t/\alpha_t\),
\(m_t=g_t\ell_t/(\alpha_tC)\), \(Y=1/C\), and
\(B=\sum_tm_t\bar b_t\), with \(\alpha_3=1\). The solution is
\[
(\tau,\lambda,H)
=(2.538468507031,.286345756920,.965543368930),
\]
and its scope returns are
\((r_1,r_2,r_3)=(.446207134309,-.644852030882,-.742686713206)\).
They are strictly decreasing, so each source workload is assigned to
its own scope. The interior profile conditions give unique maxima:
the derivative of
\(\alpha d_th'(\bar b_t)-[H-h(\bar b_t)]\) is
\(\alpha f d_t^2/\bar b_t^2>0\).

The displayed parameters are rounded values of a unique exact root in
the box of radius \(10^{-8}\) around them. The accompanying
\texttt{verify\_cross\_training.py} verifies this claim by exact
rational interval arithmetic: a preconditioned Newton map has
contraction factor below \(3.454\times10^{-7}\) and maps the box
strictly into itself. It also verifies interior profile choices,
\(r_1-r_2>1.091\), \(r_2-r_3>.0978\), and \(H-h(b_3)>1.102\).

The resulting outcomes are
\((Y,B,S)=(.393938,.174614,-1.913580)\), with
\(W_f=-1.103276\). Since
\((\eta,\lambda,f)=(1+1/Y,(\epsilon-f)/B,f)\) is the welfare
gradient, attainment of the support bound proves global optimality by
concavity of welfare. The strict scope and profile choices exclude an
optimum consisting only of matched-knowledge cutoff organizations.

\section{AI adoption and the social value of human knowledge}
\label{app:ai}

The paper's Appendix E proves that the firm can attain a global cost
minimum on a productively optimal pre-AI hierarchy and establishes the
reversal in Theorem 5. This section gives the social returns to automation
away from scope neutrality and shows why a small technological advantage
need not imply a small welfare loss. Unless stated otherwise, the hierarchy
is fixed and displacement is proportional within each scope, as in
Section 8 of the paper.

\subsection{Productivity gains and losses in policy understanding}
\label{subsec:ai-planner-comparisons}

Write \(A_r=a_r^{\mathrm{AI}}x_r\) for capacity assigned to scope \(r\)
and \(\rho_r=\ell_r/a_r^{\mathrm{AI}}\) for the labor it releases per
unit of capacity. At a feasible allocation, let \(C_H>0\) be the remaining
human labor requirement and \(N_H>0\) its understanding numerator, so
\(Y=1/C_H\), \(B=N_H/C_H\), and \(R=YB^\epsilon\). Automating scope
\(r\) changes \(C_H\) by \(-\rho_r\) and \(N_H\) by
\(-\rho_rb_r\) per unit of capacity. Differentiating gives
\begin{align}
\frac{\partial B}{\partial A_r}
&=\rho_rY(B-b_r), \label{eq:app-ai-understanding-derivative}\\
\frac{\partial W}{\partial A_r}
&=\rho_rY\left[Y+1+\epsilon\left(1-\frac{b_r}{B}\right)\right],
\label{eq:app-ai-welfare-derivative}\\
\frac{\partial\log R}{\partial A_r}
&=\rho_rY\left[1+\epsilon\left(1-\frac{b_r}{B}\right)\right].
\label{eq:app-ai-policy-resource-derivative}
\end{align}
Thus the firm ranks uses of capacity by \(\rho_r\), whereas the planner's
marginal ranking also accounts for the knowledge associated with the jobs
that remain. The productivity gain from automation can coexist with a
fall in effective public services: along a feasible marginal expansion
at scope \(r\), services fall exactly when
\begin{equation}
\epsilon\left(\frac{b_r}{B}-1\right)>1;
\qquad
\text{total welfare falls exactly when}
\quad\epsilon\left(\frac{b_r}{B}-1\right)>1+Y.
\label{eq:app-ai-policy-resource-loss}
\end{equation}
These conditions concern the effect of additional automation, beyond
the comparison between two placements in Theorem 5. The electoral loss
must outweigh the gain in public resources to reduce services; it must
also outweigh the gain in private income to reduce welfare.

The same reasoning applies within a scope when equal-cost responsibilities
require different combinations of knowledge. For a responsibility whose
workers have understanding \(b\), replace \(b_r\) by \(b\) in the
derivatives above. Conditional on the amount automated in that layer,
output is unchanged by its placement, so the planner automates the
lowest-understanding responsibilities first. The firm is indifferent
among them. This comparison allows selection within a layer without
claiming that the proportional-displacement results continue to describe
every such selection.

\subsection{A nonvanishing welfare gap near neutrality}
\label{subsec:ai-neutral-gap}

Maintain the assumptions of Theorem 5. In particular,
\(\rho_r(\zeta)=\rho(1+\zeta\delta_r)\), with
\(0=\delta_1<\cdots<\delta_n=1\), and each endpoint scope can absorb all
capacity \(A\) for small \(\zeta\geq0\). Human labor and understanding
remain positive. Let \(C^0,N^0\) denote the no-AI labor requirement and
understanding numerator. Assigning all capacity to endpoint
\(s\in\{1,n\}\) gives
\[
\begin{gathered}
C_s(\zeta)=C^0-\rho_s(\zeta)A,
\qquad N_s(\zeta)=N^0-\rho_s(\zeta)Ab_s,\\
W_s(\zeta)=\frac1{C_s(\zeta)}-\log C_s(\zeta)
+\epsilon\log\frac{N_s(\zeta)}{C_s(\zeta)}.
\end{gathered}
\]
For sufficiently small \(\zeta>0\), let
\(W^M(\zeta)=W_n(\zeta)\) be welfare in the
globally cost-minimizing implementation displayed in that theorem, and
let \(W^*(\zeta)\) be the supremum over all feasible organizations and
AI assignments. Since assigning capacity to scope 1 is feasible,
\(W^*(\zeta)\geq W_1(\zeta)\). Continuity and the equal output of the
two endpoint allocations at neutrality therefore imply
\begin{equation}
\liminf_{\zeta\downarrow0}
\bigl[W^*(\zeta)-W^M(\zeta)\bigr]
\geq
\epsilon\log\frac{N^0-\rho Ab_1}{N^0-\rho Ab_n}>0.
\label{eq:app-ai-nonvanishing-gap}
\end{equation}
The extra output from automating broad work vanishes with the
technological tilt, while the difference in occupational knowledge need
not. For fixed primitives, the displayed bound does vanish with \(A\)
or \(\epsilon\). This is a welfare comparison with a particular global
cost minimum; it does not characterize the planner's joint choice of
organization, human knowledge, and AI.

\section{Evidence on occupational knowledge and policy understanding}
\label{app:empirical-evidence}

The two analyses reported in Section 4 of the paper examine different
parts of the occupational-knowledge mechanism. The Dutch LISS panel asks
whether a worker knows more about policies close to her occupational
expertise. The ECB Consumer Expectations Survey asks whether workers with
coordinating responsibilities know more about inflation and the ECB's
role. This section describes the measures, comparisons, and results;
the accompanying \texttt{empirical\_documentation.pdf} supplies the full
question inventories, answer keys, domain assignments, and source fields.

\subsection{Occupational expertise and policy-specific accuracy}
\label{app:empirical-liss}

I link the 2021 Work and Schooling and Family Policy modules of the Dutch
Longitudinal Internet Studies for the Social Sciences (LISS), a
probability-based internet panel \citep{streefkerk-2022,bartova-2022}.
The sample contains 503 paid workers in 168 four-digit ISCO-08
occupations and 12,553 respondent--question observations. The outcome
\(C_{ij}\) is 100 if respondent \(i\) answers component \(j\) correctly
and zero otherwise. The 26 components concern maternity and parental
leave, childcare allowances, and child benefits under the rules in force
in 2021. Scoring incorporates forced-guess follow-ups; unresolved
\emph{do not know} responses are incorrect. Questions not administered
or technically missing are excluded. Each respondent
receives one of two randomized childcare-allowance vignettes. Mean
accuracy in the estimation sample is 49.3 percent.

Occupational expertise comes from the 33 knowledge domains in O*NET
version 25.0 \citep{onet-2020}. For occupation \(o\) and domain \(k\),
let \(IM_{ok}\in[1,5]\) and \(LV_{ok}\in[0,7]\) be the rated importance
and required level. I form knowledge stocks
\(z_{ok}=[(IM_{ok}-1)/4](LV_{ok}/7)\) and normalize them to shares
\(\widehat s_{ok}=z_{ok}/\sum_dz_{od}\). Profiles are mapped from U.S.
SOC to four-digit ISCO occupations using the official crosswalk
\citep{bls-2012}, giving available SOC profiles equal weight within
each ISCO occupation. An aggregate O*NET occupational record is used
when available; otherwise its specialties receive equal weight.

For each component \(j\), a set \(\mathcal R_j\) records the two to four
knowledge domains its content draws on. For example, a question about
the duration of partner leave draws on Law and Government and Personnel
and Human Resources; a question about earnings replacement also draws
on Economics and Accounting. The regressor
\(M_{oj}=\sum_{k\in\mathcal R_j}\widehat s_{ok}\) measures the share of
occupational knowledge devoted to these domains. It equals the coverage
measure \(\sum_k\min\{\widehat s_{ok},\widehat u_{jk}\}\), where
\(\widehat u_{jk}=1/|\mathcal R_j|\) on \(\mathcal R_j\) and zero
otherwise: the cap never
binds in this sample. The assignments are common across respondents and
occupations and yield eight distinct requirement profiles. All include
Law and Government, whose occupational share is absorbed by the
respondent fixed effects.

Writing \(o(i)\) for respondent \(i\)'s occupation, I estimate
\begin{equation}
C_{ij}=\alpha_i+\lambda_j+\theta\widetilde M_{o(i)j}+\varepsilon_{ij},
\label{eq:app-empirical-liss-regression}
\end{equation}
where \(\widetilde M\) is \(M\) divided by its residual standard
deviation after projection on respondent and component fixed effects.
The respondent effects \(\alpha_i\) absorb stable differences in
ability, education, and occupation; the component effects \(\lambda_j\)
absorb common differences in difficulty. Thus \(\theta\) asks whether
the same person performs relatively better on questions whose domains
are better represented in her occupational knowledge. Estimation is
unweighted because these modules provide no analysis weight. Standard
errors use a finite-sample-adjusted two-way cluster-robust covariance
estimator, clustered by four-digit occupation and by the eight requirement
profiles, with a \(t(7)\) reference distribution.

The coefficient is 1.628 percentage points per residual standard
deviation of coverage (Table~\ref{tab:empirical-results}). Excluding
occupations directly related to law, human resources, finance, statutory
advice, or childcare leaves a coefficient of 1.508 (\(p=.024\)). The
association therefore also appears among workers whose occupations are
less directly connected to the policies being tested.

\subsection{Coordinating responsibilities and economic-policy knowledge}
\label{app:empirical-ecb}

The ECB Consumer Expectations Survey \citep{ecb-2026} links reported
workplace responsibilities in May to knowledge questions in August for
workers across industries and employer sectors in 11 European countries.
Pooling the four comparable waves, 2022--2025, gives 36,099
respondent-year observations. The outcome \(K_{ict}\) is the percentage
of eight equally weighted questions answered correctly by person \(i\)
in country \(c\) and year \(t\). One asks what inflation measures; seven
ask whether particular objectives or tasks belong to the ECB, including
its inflation target, bank supervision, and national taxation or
spending. The score counts \emph{do not know} as incorrect and is missing
if any of the eight answers is technically missing.

The focal count \(R_{it}\in\{0,1,2\}\) adds two reported duties:
supervising or managing a team, and involvement in company investment
or future strategy. The former coordinates other people's work; the
latter connects decisions across the firm's activities and over time.
To distinguish these duties from other workplace authority, the regression
also includes \(P_{it}\), the count of price- and wage-setting duties,
and \(H_{it}\), the count of senior- and junior-staff hiring duties.
Each count requires both of its constituent indicators to be observed.
The pooled specification is
\begin{equation}
K_{ict}=\beta R_{it}+\delta P_{it}+\eta H_{it}
+X_{it}'\gamma+\lambda_{ct}+\varepsilon_{ict}.
\label{eq:app-empirical-ecb-regression}
\end{equation}
Here \(\lambda_{ct}\) are country-by-year fixed effects. The controls
\(X_{it}\) include indicator sets for age group, gender, education,
household-income quintile, employment type, industry, employer sector,
and firm size, plus a four-question financial-literacy score covering
compound interest, real savings returns, diversification, and compound
debt. I use positive August population weights, normalized to mean one
within each wave, and cluster standard errors by respondent.

Each focal duty is associated with 2.574 percentage points greater
accuracy. Adding one-digit occupation fixed effects in 2023--2025, when
the comparable occupation field is available, gives 2.185 points.
Adding respondent fixed effects and retaining time-varying controls
for age, income, employment type, industry, employer sector, and firm
size gives 1.298 points. This last comparison uses changes within the
same person and contains 19,648 observations for 9,810 people.
The price/wage-setting and hiring counts have smaller coefficients
(.52, \(p=.178\), and .13, \(p=.682\)); the focal coefficient exceeds
each (\(p=.0003\) and \(p<.001\)).

\begin{table}[!htbp]
\centering
\caption{Occupational knowledge and policy understanding}
\label{tab:empirical-results}
\small
\begin{tabular*}{\textwidth}{@{\extracolsep{\fill}}p{2.85in}rrrr@{}}
\toprule
Specification & Coefficient & SE & \(p\)-value & Observations \\
\midrule
\multicolumn{5}{l}{\textit{Panel A. LISS family-policy accuracy}} \\
Occupational coverage of question domains & 1.628 & .419 & .006 & 12,553 \\
\addlinespace
\multicolumn{5}{l}{\textit{Panel B. ECB institutional knowledge}} \\
Pooled specification & 2.574 & .327 & \(<.001\) & 36,099 \\
2023--2025, occupation fixed effects & 2.185 & .388 & \(<.001\) & 26,981 \\
Respondent fixed effects & 1.298 & .482 & .007 & 19,648 \\
\bottomrule
\end{tabular*}
\vspace{3pt}
\noindent\begin{minipage}{\textwidth}
\footnotesize\setstretch{1.0}
\textit{Notes:} Coefficients are percentage points per residual standard
deviation of occupational coverage in Panel A and per additional focal
duty in Panel B. Observations are respondent--component pairs in Panel A
and respondent-years in Panel B. Both LISS fixed effects are included;
ECB controls and weights are described above. LISS inference clusters by
occupation and eight requirement profiles and uses \(t(7)\); ECB
inference clusters by respondent. Clustering the pooled ECB regression
by country-by-year cell instead gives SE .432 (\(p<.001\)).
\end{minipage}
\end{table}

These observational comparisons do not separate learning through work
from selection into occupations or changing roles. They nevertheless
support the two patterns emphasized in the paper: workers know more
about policies close to their occupational expertise, and those with
coordinating responsibilities know more about system-level economic
questions. LISS measures occupational coverage of policy domains; the
ECB questions measure institutional and inflation knowledge. Neither
directly measures the ability to trace a policy's consequences, voting,
or candidate effort.

\clearpage

\begin{spacing}{1.10}
\setlength{\bibhang}{1em}
\setlength{\bibsep}{0pt plus .3pt}
\renewcommand{\bibsection}{%
  \section*{References}\addcontentsline{toc}{section}{References}}
\bibliographystyle{plainnat}
\bibliography{references}
\end{spacing}